\documentclass[11pt]{amsart}
\usepackage{amsmath, amsthm, amssymb, amsfonts,empheq}
\usepackage[colorlinks=true,linkcolor=blue,urlcolor=blue]{hyperref}
\usepackage{dsfont}
\usepackage{times}
\usepackage{color}
\usepackage{todonotes}
\usepackage{geometry}
\usepackage{graphicx}
\usepackage{verbatim}
\usepackage{caption}
\usepackage{subcaption}
\usepackage[normalem]{ulem}

\DeclareMathOperator*{\argmin}{arg\,min}
\def \cC{\mathcal{C}}
\def \cD{\mathcal{D}}
\def \cI{\mathcal{I}}
\def \cH{\mathcal{H}}

\def \cH{\mathcal{H}}
\def \cF{\mathcal{F}}

\def \cL{\mathcal{L}}

\def \cT{\mathcal{T}}
\def \cA{\mathcal{A}}
\def \cO{\mathcal{O}}
\def \P{\mathsf{P}}

\def \E{\mathsf{E}}

\def \R{\mathbb{R}}

\def \ud{\textrm{d}}
\def \gg{\hat \gamma}
\def \GG{\widehat \Gamma}
\def\am{\hat{a}}
\def\Tx{\cT_{x}}
\def\wron{w}

\newcommand{\eps}{\varepsilon}

\newcommand{\rom}[1]{\uppercase\expandafter{\romannumeral #1\relax}}
\newcommand{\closure}[2][3]{%
  {}\mkern#1mu\overline{\mkern-#1mu#2}}

\newtheorem{theorem}{Theorem}[section]
\newtheorem{lemma}[theorem]{Lemma}

\newtheorem{proposition}[theorem]{Proposition}
\newtheorem{definition}[theorem]{Definition}
\newtheorem{remark}[theorem]{Remark}

 \newtheorem{assumption}{Assumption}

\title[Hedging with a single trade]{Optimal hedging of a perpetual American put \\ with a single trade}
\author[C.~Cai, T.~De Angelis, J.~Palczewski]{Cheng Cai \and Tiziano De Angelis \and Jan Palczewski}
\address{School of Mathematics, University of Leeds, LS2 9JT, UK}
\email{\href{mailto:mmcca@leeds.ac.uk}{mmcca@leeds.ac.uk} (Cheng Cai)}
\email{\href{mailto:t.deangelis@leeds.ac.uk}{t.deangelis@leeds.ac.uk} (Tiziano De Angelis)}
\email{\href{mailto:j.palczewski@leeds.ac.uk}{j.palczewski@leeds.ac.uk} (Jan Palczewski)}

\date{\today}

\numberwithin{equation}{section}

\begin{document}

\begin{abstract}
It is well-known that using delta hedging to hedge financial options is not feasible in practice. Traders often rely on discrete-time hedging strategies based on fixed trading times or fixed trading prices (i.e., trades only occur if the underlying asset's price reaches some predetermined values). Motivated by this insight and with the aim of obtaining explicit solutions, we consider the seller of a perpetual American put option who can hedge her portfolio once until the underlying stock price leaves a certain range of values $(a,b)$. We determine optimal trading boundaries as functions of the initial stock holding, and an optimal hedging strategy for a bond/stock portfolio. Optimality here refers to the variance of the hedging error at the (random) time when the stock leaves the interval $(a,b)$. Our study leads to analytical expressions for both the optimal boundaries and the optimal stock holding, which can be evaluated numerically with no effort.   
\end{abstract}

\maketitle

\section{Introduction}

In this paper we construct a hedging strategy in a Black and Scholes market for the seller of a perpetual American put option, who holds a bond/stock hedging portfolio and can rebalance her position only once until the stock price leaves a predetermined interval $(a,b)$ with $0<a<b<+\infty$. The aim of the trader is to minimise the variance of the hedging error at the (random) time at which the stock leaves the above interval.

We reduce the problem to an optimal stopping problem corresponding to the timing of a single rebalancing opportunity and an optimisation problem for the choice of the initial portfolio (c.f. \cite{martini1999variance, trabelsi2002,trabelsi2003}). The former has a payoff function of a very complicated form, which prevents the use of a guess-and-verify approach. We prove the existence of two optimal trading boundaries. When the stock price reaches either of the two optimal trading boundaries the hedging portfolio must be rebalanced; we give an analytical formula for the optimal stock holding after the trade, which in general is different from that prescribed by the classical delta-hedging.

The stopping boundaries can be calculated from analytical formulae up to a solution of algebraic equations. Those algebraic equations cannot be solved explicitly and do not reveal any further properties of the boundaries. Instead, we employ delicate probabilistic arguments to show that those boundaries exhibit monotone and continuous dependence on the initial stock holding (see Figure \ref{p4.0} for an illustration).

We prove that the value function $V$ of the stopping problem is a unique solution of a free boundary problem associated with the optimal trading boundaries. We show that $V$ is everywhere continuously differentiable with respect to the initial stock holding and the initial stock price. Furthermore, discontinuities in the second order derivative with respect to the initial stock price occur only at the optimal trading boundaries.

The question of rebalancing portfolios with a limited number of trades has a long history and has been addressed in various ways. In practice, continuous trading as prescribed by the classical delta-hedging is not viable due to several reasons: first, continuous trading is a mathematical abstraction; second, well-documented discrepancies between the Black and Scholes model and real markets (e.g., volatility smiles and transaction costs) somewhat curb the applicability of Black and Scholes model. Practitioners have adopted a broad range of simple rules for their rebalancing strategies as, e.g., rebalancing at fixed times or rebalancing at fixed values of the underlying asset's price (see, e.g., \cite{sinclair2011}). The latter strategy, in particular, inspired our work: we determine {\em optimal} values of the asset price at which a trade should be made and also an optimal trade. 

Of course, it would be desirable to extend our setting to allow the trader multiple trades (not just one), as in, e.g., \cite{martini1999variance}, \cite{trabelsi2002} and \cite{trabelsi2003}, but such an extension inevitably leads to more abstract results than ours. Indeed, the above papers aim for a general setup and obtain mostly results on the existence of optimal strategies (via viscosity theory, in \cite{martini1999variance}, and martingale methods, in \cite{trabelsi2002}). These results do not allow to determine analytically shapes of the trading regions and to compute efficiently trading strategies. If rebalancing once in the entire lifetime of the option is certainly too restrictive, our assumption of rebalancing once prior to the stock price leaving a given interval $(a,b)$ improves on real-life strategies, where traders set target values of the stock price at which they reassess their position (our $a$ and $b$). Considering a perpetual option is convenient because it guarantees that the problem is time-homogeneous and allows for explicit calculations. This is also a reasonable approximation for options far from their maturity. In Section \ref{sec:concl} we compare the performance of our optimal hedging strategy to some frequently used ad-hoc strategies. Our optimal hedging strategy produces the variance of the tracking error which is up to 4 times smaller than the other strategies.

The literature around optimal hedging is very rich and branches out in several directions. From a mathematical point of view it motivated important work on approximation of stochastic integrals. Since the variance of the tracking error for the hedging portfolio is an $L^2$-distance for stochastic processes, its minimisation is referred to as quadratic hedging. Besides the seminal work by F\"ollmer and Schweizer \cite{follmer1991hedging}, which laid the foundations of quadratic hedging for claims in incomplete markets, we also mention here work by Schweizer \cite{schweizer1995variance}, Sch\"al \cite{schal1994quadratic} and Mercurio and Vorst \cite{mercurio1996option}, who focus on approximation of random variables (representing European claims at maturity) via stochastic integrals for discrete-time processes. More recently in the mathematical literature we find numerous papers concerning the asymptotic optimality of discrete-time hedging strategies as the number of hedging opportunities tends to infinity (see, e.g., Fukasawa \cite{Fukasawa}, Gobet and Landon \cite{gobet2014almost}, Rosenbaum and Tankov \cite{rosenbaum2014}, Cai et al.~\cite{cai2016optimal}). Those papers also approach the problem by approximating random variables with stochastic integrals for discrete-time processes. Finally, Ekren, Liu and Muhle-Karbe \cite{ekren2018optimal} study optimal hedging frequency in the asymptotic limit of small transaction costs for portfolio with multiple assets. The methodology and the nature of the results in those papers are rather far from our work. 

In the finance literature we find work by Ahn and Wilmott \cite{ahn2009note}, who illustrate numerically the performance of various hedging strategies with finitely many hedging opportunities. Boyle and Emanuel \cite{boyle1980discretely} study the distribution of portfolio returns with discrete hedging. Ma{\v{s}}tinsek \cite{mastinvsek2006discrete} studies the error in piecewise constant hedging strategies as a function of the time interval $\delta t$ between trades in the presence of transaction costs. Mello and Neuhaus \cite{mello1998portfolio} research the accumulated hedging error due to discrete rebalancing, extending the work by Figlewski \cite{figlewski1989options} to imperfect markets. The idea of allowing hedging at the time when fixed relative changes in the stock price occur is explored in \cite{prigent2004option}, where the price dynamic (in an incomplete market) is a marked point process.

Finally, it is worth noticing that optimal multiple stopping has been studied in the context of pricing of swing options in the energy market (see, e.g., Lempa \cite{lempa2014} for a survey). In particular, optimal boundaries for options with a put payoff are studied analytically in Carmona and Touzi \cite{carmona2008optimal} and Carmona and Dayanik \cite{Dayanik2008}, in infinite horizon, and De Angelis and Kitapbayev \cite{de2017optimal}, in finite horizon (\cite{carmona2008optimal} also consider finite horizon but only numerically). In models of optimal multiple stopping there is normally a minimum time-lag between subsequent admissible stopping times. That is imposed as a constraint on the set of admissible stopping sequences and guarantees that simultaneous use of all the stopping times cannot occur. In the case of discrete hedging, there is no need for such constraint: simultaneous use of all stopping times never occurs because, at each stopping time, the portfolio weights are adjusted and therefore each subsequent stopping problem is of a different nature. Moreover, the optimisation of the portfolio weights leads to extremely convoluted analytical expressions for the subsequent stopping problems, so that the resulting optimal multiple stopping problem is much harder to tackle than those in \cite{Dayanik2008}, \cite{carmona2008optimal} and \cite{de2017optimal}. 

The paper is organised as follows. In Section \ref{sec:1} we set the hedging problem in a rigorous mathematical framework. In Section \ref{sec:2} we study the hedging problem for a fixed value of the initial stock holding. We prove continuity and differentiability of the value function with respect to the initial stock price (Theorem \ref{thm:C1}). We determine the existence of optimal trading boundaries (Proposition \ref{prop4.33}) and we prove that the value function solves a suitable variational problem (Theorem \ref{thm:VI}). In Section \ref{sec:3} we prove that the optimal trading boundaries are continuous monotonic functions of the initial stock holding (Theorems \ref{prop5.4.1} and \ref{prop5.5.1}). Section \ref{sec:4} gives necessary first order condition which enable computation of an optimal initial stock holding. We explore the optimisation problem with the right boundary $b=\infty$ in Section \ref{secSETB}. The study is complemented in Section \ref{sec:concl} with an extensive numerical analysis of the properties of the optimal hedging strategies and of the corresponding hedging error.


\section{Problem formulation and background material}\label{sec:1}

Consider a Black-Scholes economy on a complete probability space $(\Omega,\cF,\P)$ with risk neutral measure $\P$. We have one risky stock $S$ and a risk-free bond $B$, following the dynamics
\begin{align}
\label{eq:S}\ud S_t&=rS_t\ud t+\sigma S_t\ud W_t,\quad S_0=x,\\
\label{eq:B}\ud B_t&=rB_t\ud t,\qquad\qquad\qquad B_0=1,
\end{align}
where $W=(W_t)_{t\geq0}$ is a Brownian motion, $r>0$ is the risk-free rate and $\sigma>0$ is the stock's volatility. Let $\mathbb{F}:=(\mathcal{F}_t)_{t\ge 0}$ be the natural filtration generated by $W$ satisfying the usual conditions. When necessary we will denote by $(S^x_t)_{t\geq0}$ the process $S$ starting at $x$. Alternatively we will use the notation $\P_x(\,\cdot\,)=\P(\,\cdot\,|S_0=x)$ and $\E_x[\,\cdot\,]=\E[\,\cdot\,|S_0=x]$. 

An option trader sells one perpetual American put option written on the stock $S$ with the strike price $K$. Such option gives its holder the right but not the obligation to sell one share of the stock $S$ for the price $K$ at any (random) time $\tau\in[0,\infty]$. It is well-known that, if the initial stock price is $x$, the arbitrage-free price $P(x)$ of the option is given by
\begin{equation}\label{eq:AmP}
P(x)=\sup_{\tau} \E\left[e^{-r\tau}(K-S^x_\tau)^{+}\right],
\end{equation}
where the supremum is taken over all $\mathbb{F}$-stopping times. The explicit form of $P(x)$ is known (see, e.g., \cite[Chapter VII]{peskir2006optimal}) and it reads 
\begin{equation}\label{eq:Panalyt}
P(x)= 
\left\{
\begin{array}{ll}
     \frac{1}{d}\am^{1+d}x^{-d},  & \am\leq x<\infty, \cr\\
    K-x,   & 0\leq x \leq  \am,
  \end{array}
  \right.
  \end{equation}
where $d:=2r/\sigma^2$ and 
\begin{equation}
\label{opta}
\am:=\frac{K}{1+\frac{1}{d}}
\end{equation} 
is the so-called {\em exercise boundary}; that is, the holder exercises the option optimally according to the stopping rule
\begin{equation}\label{eq:taua}
\tau_{\am}:=\inf\{t\geq0:S_t\leq \am\}.
\end{equation}

By a straightforward application of It\^o-Tanaka's formula we can derive the dynamics of the discounted option price, that is 
\begin{equation}
\label{optiondyna}
\ud (e^{-rt}P(S_t))=-e^{-rt}rK 1_{\{S_t<\am\}}\ud t+e^{-rt}\sigma S_t P'(S_t)\ud W_t,
\end{equation}
where $1_{\{\cdot\}}$ denotes the indicator function. It is immediate to verify that
\begin{align*}
&t\mapsto e^{-rt}P(S_t)\text{ is a supermartingale and } t\mapsto e^{-r(t\wedge\tau_{\am})}P(S_{t\wedge\tau_{\am}})\text{ is a martingale.}
\end{align*}

According to classical theory the seller of the option should use Delta hedging to construct a replicating portfolio for the perpetual American put. The Delta of the option corresponds to the first derivative 
\[
P'(x)=\max\{-(\am/x)^{1+d},-1\},
\]
which is an increasing function taking values in $[-1,0)$ and is strictly increasing on $(\am,\infty)$. Notice that the Delta appears in the stochastic integral of \eqref{optiondyna}. This highlights that, under the classical Black-Scholes model, if the option holder does not exercise the option at $\tau_{\am}$, the option seller gains instantaneous interests $rK$ with her short position perfectly hedged. 

In our problem formulation, we tacitly assume that the option holder exercises the option optimally, hence as soon as $S_t$ falls below $\am$.        

Our trader faces the following hedging scenario: after selling the option, she constructs a self-financing (hedging) portfolio $\Pi=(\Pi_t)_{t\geq0}$ with bond holding $(m_t)_{t\ge 0}$ and stock holding $(\theta_t)_{t\ge0}$, that is 
\[
\Pi_t=\theta_tS^x_t+m_tB_t,\qquad t\ge 0;
\]
at time $t=0$, she chooses an initial stock holding $\theta_0=h$ and bond holding $m_0=P(x)-hx$. However, in contrast to the classical Delta hedging model, the seller is allowed to rebalance her portfolio only once at a (stopping) time $\tau$ of her choosing before the stock price leaves a given interval $(a,b)$. Her goal is to find an admissible trading strategy (in a sense which will be made precise in Definition \ref{Deftradingstrategy}) so that the variance of the tracking error is minimised (this will also be clarified in a moment). The thresholds $a,b$ can be interpreted as re-assessment price levels set by the option seller. From a practical point of view, the option seller will choose those levels on the grounds of subjective propensity to risk and operational/regulatory constraints.

Since we are assuming that the option holder exercises the option according to the stopping rule \eqref{eq:taua}, it is natural to only allow $a\ge\am$. We also assume that $b<\infty$ and define $\cI:=(a,b)$ and $\closure\cI:=[a,b]$. The (random) time horizon of our problem is given by
\begin{equation*}
\tau^x_\cI:=\inf\{t\geq0: S^x_t\notin \cI\}.
\end{equation*}
We often omit the superscript $x$ if it does not lead to ambiguity. For mathematical completeness the case of $b=\infty$ is discussed separately in Section \ref{secSETB} as it presents some specific technical features. 

In order to formally define admissible trading strategies, we need to introduce some notation. 
Given an initial stock price $S_0=x\in\closure\cI$, we let
\begin{align*}
\Tx:=\left\{ \tau:\,\text{$\tau$ is a $\mathbb{F}$-stopping time such that $\tau\leq\tau^x_\cI$, $\P$-a.s.}\right\}
\end{align*} 
and for any $\tau\in\Tx$ we define
\begin{align}\label{eq:Htau}
\mathcal{H}^\tau:=\{h_1: \Omega \to \R\ :\ \text{$h_1$ is $\mathcal{F}_\tau$-measurable and } \E\left[(h_1)^2\right]<\infty\}.
\end{align}
Since the seller's optimisation problem ends at the time when the price process leaves the interval $\cI$, it is natural to consider an initial stock holding $\theta_0$ which lies in the set
\begin{align*}
\cH:=[P'(a),P'(b)], 
\end{align*}
where it is worth recalling that $P'(x)=-(\am/x)^{1+d}$ for $x\ge \am$.
\begin{definition}[Trading strategy]
\label{Deftradingstrategy}
For an initial stock price $S_0=x\in \closure{\cI}$, the set of admissible trading strategies $\cA_{x}$ consists of pairs $(\tau,\theta)$, such that $\tau\in \Tx$ and
\begin{align*}
\theta_t&:=
 \begin{cases}
    h,  \quad\quad 0\leq t \le\tau, \cr
    h_1,   \quad \tau< t \leq  \tau^x_\cI,
  \end{cases}
\end{align*}
where $h\in\mathcal{H}$ is the initial stock holding and $h_1\in \mathcal{H}^{\tau}$ is the new stock holding after the trade. 
\end{definition}

Given a trading strategy $(\tau, \theta)\in\cA_x$, the trader's self-financing, hedging portfolio follows the dynamics
\begin{align}
\label{portfolio}\ud\Pi^{\tau,\theta}_t&=\theta_t\ud S^x_t+m_t\ud B_t ,\qquad \Pi^{\tau,\theta}_0=h x+m_0 = P(x).
\end{align}
Then, combining \eqref{portfolio} with \eqref{eq:S}--\eqref{eq:B}, it is easy to verify that the discounted portfolio process $t\mapsto e^{-rt}\Pi^{\tau,\theta}_t$ is a local martingale with the dynamics  
\begin{equation}
\label{portdyna}
\ud (e^{-rt}\Pi^{\tau,\theta}_t)=\theta_t\ud (e^{-rt}S^x_t)=e^{-rt}\theta_t \sigma S^x_t\ud W_t.
\end{equation}

Finally, we can formulate the optimisation problem for the option seller. As mentioned above, the seller wants to minimise the {\em variance} of the tracking error (i.e., the difference between the hedging portfolio and the option price) at the terminal time $\tau_\cI$. It is worth remarking that the choice of the variance is natural since the mean of the tracking error is completely uninformative. Indeed
\begin{align}\label{eq:mean}
\E_x\left[e^{-r\tau_\cI}\big(\Pi^{\tau,\theta}_{\tau_\cI}-P(S_{\tau_\cI})\big)\right]=0
\end{align}
thanks to the optional sampling theorem, upon recalling that on the stochastic interval $[0,\tau_\cI]$ the price process $S$ is bounded and $(\theta_t)_{t\ge 0}$ is a square integrable process (cf.~\eqref{eq:Htau}).
Then, given an initial price $S_0=x$ we are interested in the problem 
\begin{align}\label{eq:V}
\mathcal V(x)=&\inf_{(\tau,\theta)\in\cA_x} \mathcal{V}ar_{x}\left[e^{-r\tau_\cI}\bigl(\Pi^{\tau,\theta}_{\tau_\cI}-P(S_{\tau_\cI})\bigr)\right]\\
=&\inf_{(\tau,\theta)\in\cA_x} \E_{x}\left[e^{-2r\tau_\cI}\bigl(\Pi^{\tau,\theta}_{\tau_\cI}-P(S_{\tau_\cI})\bigr)^2\right],\notag
\end{align}
where we use the notation $\mathcal{V}ar_{x}[\,\cdot\,]=\mathcal{V}ar[\,\cdot\,|S_0=x]$ and the second equality follows from \eqref{eq:mean}.

\begin{remark}\label{rem:prem}
It is assumed above that the option is sold for the price $P(x)$ and the seller invests the proceeds in the hedging portfolio, i.e., $\Pi^{\tau,\theta}_0 = P(x)$. However, the seller aware of her trading constraints may sell the option at a premium over the Black-Scholes price, i.e., for $P(x)+\delta$ with $\delta>0$. Denoting by $(\Pi^{\tau,\theta;\delta}_t)_{t\ge 0}$ the associated hedging portfolio, for any trading strategy $(\tau,\theta) \in \cA_x$ it follows from \eqref{portfolio} and \eqref{portdyna} that $\Pi^{\tau,\theta;\delta}_t=e^{rt}\delta+\Pi^{\tau,\theta}_t$ for all $t\ge 0$. The mean tracking error equals (c.f.\ \eqref{eq:mean})
\[
\E_x\left[e^{-r\tau_\cI}\big(\Pi^{\tau,\theta;\delta}_{\tau_\cI}-P(S_{\tau_\cI})\big)\right]=\delta
\]
and consequently
\begin{align*}
\inf_{(\tau,\theta)\in\cA_x} \mathcal{V}ar_x\left[e^{-r\tau_\cI}\bigl(\Pi^{\tau,\theta;\delta}_{\tau_\cI}-P(S_{\tau_\cI})\bigr)\right]&=\inf_{(\tau,\theta)\in\cA_x} \mathcal{V}ar_x\left[e^{-r\tau_\cI}\bigl(\Pi^{\tau,\theta}_{\tau_\cI}-P(S_{\tau_\cI})\bigr)+\delta \right]\\
&=\inf_{(\tau,\theta)\in\cA_x} \E_{x}\left[e^{-2r\tau_\cI}\bigl(\Pi^{\tau,\theta}_{\tau_\cI}-P(S_{\tau_\cI})\bigr)^2\right] = \mathcal{V}(x).
\end{align*}
Hence the problem simplifies to the one studied in the paper.
\end{remark}

One may argue that if all sellers on the market charge a premium on the Black-Scholes price, then the tracking error should be computed accounting for such premium too. As shown in the next remark, if we assume a multiplicative premium we can embed these models in our set-up.

\begin{remark}
Due to trading frictions on real markets, the selling price of the option may be higher than the theoretical Black-Scholes price. Assuming a multiplicative adjustment, the option's selling price is $P(x)(1+\eps)$, where $\eps \ge 0$, and we denote by $(\Pi^{\tau,\theta;\eps}_t)_{t\ge 0}$ the associated hedging portfolio. The trader receives $P(x)(1+\eps)$ at time $0$ and tracks the selling price $P(S_t)(1+\eps)$ (so that she can close the position at time $\tau_\cI$). In view of \eqref{eq:mean}, the mean tracking error is
\begin{align*}
\E_x\left[e^{-r\tau_\cI}\big(\Pi^{\tau,\theta;\eps}_{\tau_\cI}- P(S_{\tau_\cI})(1+\eps)\big)\right]=(1+\eps) \E_x\left[e^{-r\tau_\cI}\big(\Pi^{\tau,\theta'}_{\tau_\cI}- P(S_{\tau_\cI})\big)\right]=0,
\end{align*}
where $\theta'_t = \theta_t/(1+\eps)$, $t \ge 0$, is used along with \eqref{portdyna} to obtain $e^{-r t}\Pi^{\tau,\theta;\eps}_t=(1+\eps)e^{-r t}\Pi^{\tau,\theta'}_t$. Therefore, 
\begin{align*}
&\inf_{(\tau,\theta)\in\cA_x} \mathcal{V}ar_{x}\left[e^{-r\tau_\cI}\bigl(\Pi^{\tau,\theta;\eps}_{\tau_\cI}-P(S_{\tau_\cI})(1+\eps)\bigr)\right]\\
&=(1+\eps)^2 \inf_{(\tau,\theta)\in\cA_x} \mathcal{V}ar_{x}\left[e^{-r\tau_\cI}\bigl(\Pi^{\tau,\theta}_{\tau_\cI}-P(S_{\tau_\cI})\bigr)\right] = (1+\eps)^2 \mathcal{V}(x),
\end{align*}
and the optimisation problem simplifies to the one studied in the paper.
\end{remark}

Using the integral forms of the dynamics \eqref{portdyna} and \eqref{optiondyna} and It\^o's isometry we obtain a more convenient problem formulation:
\begin{align}
\label{1.0}
\mathcal V(x)&=\inf_{(\tau,\theta)\in\cA_x} \E_{x}\left[\Bigl(\int_0^{\tau_\cI}e^{-ru}(\theta_u-P'(S_u))\sigma S_u\ud W_u\Bigr)^2\right]\\
&=\inf_{(\tau,\theta)\in\cA_x} \E_{x}\left[\int_0^{\tau_\cI}e^{-2ru}f(S_u,\theta_u)\ud u\right],\nonumber
\end{align}
where
\begin{equation}
\label{f_x_h}
f(x,\theta):=(\theta- P'(x))^2\sigma^2x^2.
\end{equation}
The final expression in \eqref{1.0} highlights the well-known fact that Delta hedging amounts to controlling the difference between $\theta_t$ and $P'(S_t)$. In the absence of trading constraints the optimal trading strategy would be the Black-Scholes strategy $\theta_t=P'(S_t)$, which would produce no tracking error with certainty, i.e.~$\mathcal V(x)=0$.

Notice that we can rewrite our problem as 
\begin{align}\label{eq:vscript}
\mathcal V(x)=\inf_{h \in\cH} V(x,h),
\end{align}
where
\begin{equation}
\label{1.1}
V(x,h):=\inf_{(\tau,h_1)\in\cT_x\times\mathcal{H}^{\tau}} \E_x\left[\int_0^{\tau}e^{-2ru}f(S_u,h)\ud u+\int_{\tau}^{\tau_\cI}e^{-2ru}f(S_u,h_1)\ud u\right].
\end{equation}
In light of this observation we will first proceed with a detailed analysis of the function $V(x,h)$ and subsequently we will determine $\mathcal V(x)$. By doing this, we will also obtain an optimal control $(\tau^*,\theta^*)$. 

We close this section recalling some useful facts and some notation. Let $\partial_x$ and $\partial_{xx}$ denote partial derivatives with respect to $x$. For future frequent use we introduce the infinitesimal generator of the process $S$, denoted by $\mathcal{L}$, and defined by its action on functions $v\in C^2(\R_+)$ as follows:
\begin{equation*}
\mathcal{L}v(x):=\tfrac{\sigma^2}{2}x^2\partial_{xx}v(x)+rx \partial_xv(x).
\end{equation*}
Recalling that $d=2r/\sigma^2$ we have that 
\begin{equation}\label{eq:q12}
q_1= \frac{1-d+\sqrt{(1-d)^2+8d}}{2}>0,\quad\quad q_2= \frac{1-d-\sqrt{(1-d)^2+8d}}{2}<0,
\end{equation}
are the roots of
\[
q^2+(d-1)q-2d=0.
\]
Since our price process is absorbed at $\{a,b\}$, we will need the functions $\varphi$ and $\psi$ defined, respectively, as the unique (up to multiplication) decreasing and increasing fundamental solutions of the ODE
\begin{align}
\label{funODE}
(\mathcal{L}-2r)v(x)&=0,\quad x\in(a,b),
\end{align}
with boundary conditions $$\psi(a+)=0,\quad\psi'(a+)>0,\quad\varphi(b-)=0,\quad \varphi'(b-)<0.$$
They are conveniently constructed as linear combinations of $\hat\varphi(x):=x^{q_2}$ and $\hat\psi(x):=x^{q_1}$ by taking (see, e.g., \cite{Alv2001})
\begin{align}\label{eq:fund}
\varphi(x)=\hat\varphi(x)-\frac{\hat\varphi(b)}{\hat\psi(b)}\hat\psi(x)\quad\text{and}\quad \psi(x)=\hat\psi(x)-\frac{\hat\psi(a)}{\hat\varphi(a)}\hat\varphi(x).
\end{align}

Finally, using $\varphi$ and $\psi$, we recall an analytical expression of the resolvent for a one-dimensional diffusion, which can be found in \cite[Chapter \rom{2}, p.19]{borodin2012handbook}. For any $x\in \cI$, and any bounded measurable function $g:\cI\to\R$ we have
\begin{align}
\label{diffusionformula}
&\E_{x}\left[\int_{0}^{\tau_\cI}e^{-2ru}g(S_u)\ud u\right]
=\wron^{-1}\left(\varphi(x)\int_{a}^x\psi(z)g(z)m'(z)\ud z+\psi(x)\int_{x}^{b}\varphi(z)g(z)m'(z)\ud z\right),
\end{align}
where $\wron$ is the Wronskian (with the value independent of $x$)
\begin{equation*}
\wron=\psi'(x)\frac{\varphi(x)}{s'(x)}-\varphi'(x)\frac{\psi(x)}{s'(x)} > 0,
\end{equation*}
and $s'(x)$ and $m'(x)$ are the densities of the scale function and of the speed measure of $(S_t)_{t\ge 0}$, respectively. They are explicitly given by
\begin{equation}
\label{speedmeasure}
s'(x)=c\,x^{-d}\qquad\text{and}\qquad m'(x)=2x^{d-2}/c\,\sigma^2,
\end{equation}
where $c>0$ is the same constant in both expressions ($s'$ and $m'$ are uniquely defined up to multiplication). For future reference, we notice that the Wronskian $w$ can be also expressed in terms of the Wronskian $\hat w$ associated to $\hat{\varphi}$ and $\hat \psi$. In particular, it is not hard to check that (recall that $q_2<0<q_1$)
\begin{align}\label{eq:w}
w=\hat w\left(1-(a/b)^{q_1-q_2}\right).
\end{align}
This observation will be useful when we later consider fundamental solutions of \eqref{funODE} on intervals $\cI'\neq\cI$.


\section{A one dimensional optimal stopping problem}\label{sec:2}

In this section, we study problem \eqref{1.1} for each fixed initial stock holding $h\in\cH$. First we find the optimal stock holding $h_1$ and reduce \eqref{1.1} to a standard one dimensional optimal stopping problem, then we solve the optimal stopping problem via associated free boundary problems.


\subsection{Reduction to a Markovian optimal stopping problem}\label{sec2.1}

The first task is to show that it is sufficient to draw $h_1\in\cH^\tau$ from the class of Markovian controls. To this end, we introduce the set of Markovian controls $\cH^{\tau}_m$ defined as
\begin{equation*}
\mathcal{H}_{m}^{\tau}:=\{h_1\in \cH^{\tau}: h_1=\ell(S_\tau)\text{ for some measurable $\ell:\closure{\cI}\to\R$}\}. 
\end{equation*}
Consider an analogue of problem \eqref{1.1} but with the constraint of using Markovian controls and denote its value by
\begin{equation}
\label{markovianvalue}
\widetilde{V}(x,h):=\inf_{(\tau,h_1)\in\cT_x\times\mathcal{H}_{m}^{\tau}} \E_x\left[\int_0^{\tau}e^{-2ru}f(S_u,h)\ud u+\int_{\tau}^{\tau_\cI}e^{-2ru}f(S_u,h_1)\ud u\right].
\end{equation}
Next we show the equivalence of \eqref{1.1} and \eqref{markovianvalue}. 
\begin{proposition}
\label{propmarkoviancontrol}
For all $(x,h)\in\closure{\cI}\times\cH$ we have $\widetilde{V}(x,h)=V(x,h)$.
\end{proposition}
\begin{proof}
Since $h_1$ is $\cF_{\tau}$ measurable, expanding the square in \eqref{f_x_h} and using the tower property of conditional expectation, we can write \eqref{1.1} as
\begin{align}
\label{markprop}
V(x,h)=&\inf_{(\tau,h_1)\in\cT_x\times\mathcal{H}^{\tau}} \E_x\bigg[\int_0^{\tau}e^{-2ru}f(S_u,h)\ud u+h_1^2\E_x\left(\int_{\tau}^{\tau_\cI}e^{-2ru}\sigma^2 S^2_u\ud u\Big| \mathcal{F}_\tau \right)\\
&\qquad\qquad\qquad\quad -2h_1\E_x\left(\int_{\tau}^{\tau_\cI}e^{-2ru}P'(S_u)\sigma^2 S^2_u\ud u\Big| \mathcal{F}_\tau \right)\notag\\
&\qquad\qquad\qquad\quad+\E_x\left(\int_{\tau}^{\tau_\cI}e^{-2ru}(P'(S_u))^2\sigma^2 S^2_u\ud u\Big| \mathcal{F}_\tau \right)\bigg].\notag
\end{align}
Notice that for any trading time $\tau$, the expression under the expectation $\E_x$ is quadratic in $h_1$. Then the optimal stock holding $h_1^*$ is
\begin{align}
\label{h1tau}
h_1^*=&\frac{\E_x\left(\int_{\tau}^{\tau_\cI}e^{-2ru}P'(S_u)\sigma^2 S^2_u\ud u\Big| \mathcal{F}_\tau \right)}{\E_x\left(\int_{\tau}^{\tau_\cI}e^{-2ru}\sigma^2 S^2_u\ud u\Big| \mathcal{F}_\tau \right)}=\frac{\E_{S_\tau}\left[\int_{0}^{\tau_\cI}e^{-2ru}P'(S_u)\sigma^2 S^2_u\ud u\right]}{\E_{S_\tau}\left[\int_{0}^{\tau_\cI}e^{-2ru}\sigma^2 S^2_u\ud u \right]},
\end{align}
where the final equality follows from the strong Markov property of the process $S$. Therefore, the optimal stock holding $h^*_1$ is a measurable function of the stock price $S_\tau$ at time $\tau$. Hence it suffices to consider problem \eqref{markovianvalue} instead of \eqref{1.1}. Notice that a similar result was also obtained by \cite{martini1999variance}.
\end{proof}

\begin{figure}[tb]
\centering
\includegraphics[width=7cm]{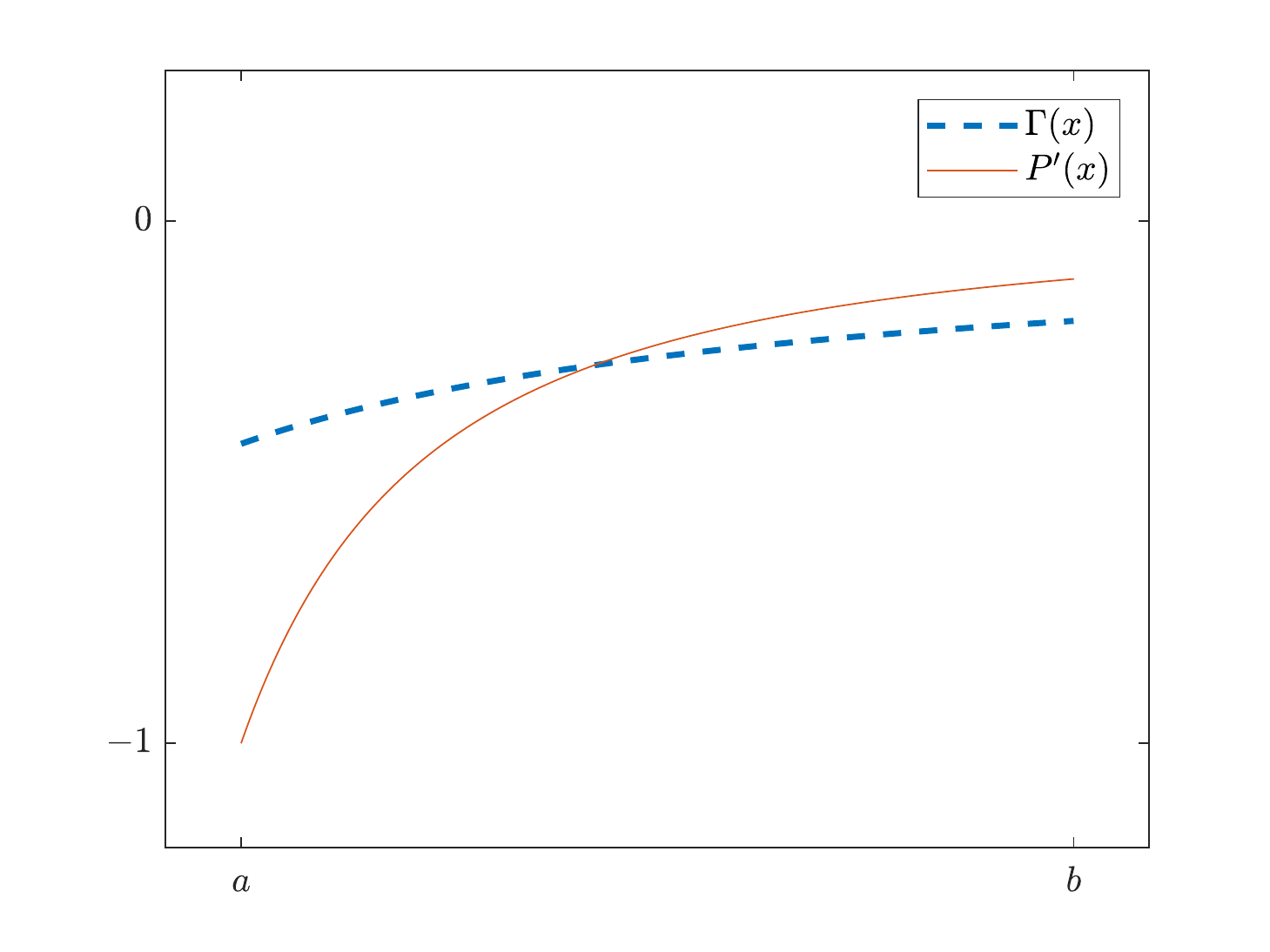}
\caption{Plots of the functions $\Gamma(x)$ and $P'(x)$ using parameters $r=3\%$, $\sigma=30\%$, $K=100$ and $b=150$. Notice that $a=K/(1+d^{-1})=40$.} 
\label{p1.0}
\end{figure}
 
Thanks to Proposition \ref{propmarkoviancontrol}, we can apply the strong Markov property of $(S_t)_{t\ge 0}$ to transform \eqref{1.1} into a canonical impulse control form:
\begin{align}
\label{findh1}
V(x,h)&=\inf_{(\tau,h_1)\in\mathcal{T}_x\times\mathcal{H}_{m}^{\tau}}\E_x\left[\int_0^{\tau}e^{-2ru}f(S_u,h)\ud u +e^{-2r\tau}\widehat{M}(S_\tau,h_1)\right],
\end{align}
where 
\begin{equation}
\label{2.2}
\widehat{M}(x,\zeta):=\E_x\left[\int_{0}^{\tau_\cI}e^{-2ru}f(S_u,\zeta)\ud u\right],\qquad \zeta\in\R, \quad x \in \closure\cI.
\end{equation}
Expanding the square in $f$ yields the following representation for $\widehat M$
\begin{equation}
\label{2.3}
\widehat M(x, \zeta) =\zeta^2\gamma_1(x)-2\zeta \gamma_2(x)+\gamma_3(x),
\end{equation}
where
\begin{equation}
\label{eq:Gamma}
\begin{aligned}
\gamma_1(x) &= \E_{x}\left[\int_{0}^{\tau_\cI}e^{-2ru}\sigma^2 S^2_u\ud u \right],\qquad \gamma_2(x) = \E_{x}\left[\int_{0}^{\tau_\cI}e^{-2ru}P'(S_u)\sigma^2 S^2_u\ud u\right],\\
\gamma_3(x) &= \E_{x}\left[\int_{0}^{\tau_\cI}e^{-2ru}(P'(S_u))^2 \sigma^2S^2_{u} \ud u\right].
\end{aligned}
\end{equation}
Direct calculations, using \eqref{diffusionformula}, lead to explicit formulae for $\gamma_i$, $i=1, 2, 3,$
\begin{align}
\label{eq:gamma1}\gamma_1(x)&=-x^2 + A_1 D_{2}x^{q_1} + A_{2}D_1 x^{q_2},\\
\label{eq:gamma2}\gamma_2(x)&=-\frac1{d}\am^{1+d} x^{1-d}+ A_1C_2x^{q_1}+  A_2C_1x^{q_2},\\
\label{eq:gamma3}\gamma_3(x)&=-\frac1{d^2}\hat{a}^{2+2d}x^{-2d}+A_1B_2x^{q_1}+ A_2B_1x^{q_2},
\end{align}
where, using $q_1$ and $q_2$ given in \eqref{eq:q12},
\[
\begin{array}{ll}
A_i:=[a^{q_i-q_{3-i}}-b^{q_i-q_{3-i}}]^{-1}, &B_i:=d^{-2}\bigl((\am/a)^{2+2d} a^{2-q_i}-(\am/b)^{2+2d}b^{2-q_i}\bigr),\\[5pt]
C_i:=d^{-1}\bigl((\am/a)^{1+d} a^{2-q_i}- (\am/b)^{1+d} b^{2-q_i}\bigr), & D_i:=(a^{2-q_i}-b^{2-q_i}).
\end{array}
\]

The expression \eqref{h1tau} in the proof of Proposition \ref{propmarkoviancontrol} implies that the optimal stock holding after the rebalancing of the portfolio is a measurable function of the stock price at the rebalancing time. It is a minimiser of $\zeta \mapsto \widehat M(x, \zeta)$ which, thanks to the representation \eqref{2.3}, is unique and given by
\begin{align}\label{2.4}
\Gamma(x) = \argmin_{\zeta}\widehat{M}(x,\zeta) = \frac{\gamma_2(x)}{\gamma_1(x)}, \qquad x \in \cI.
\end{align}
Notice that the optimal stock holding $\Gamma$ is defined only on $\cI$. If the trader trades at $\tau_\cI$, her choice of the stock holding becomes irrelevant for the optimisation problem. 

Denoting
\begin{align}\label{eq:M}
M(x):=\widehat{M}(x,\Gamma(x)) \text{\quad for } x\in\cI, \text{\quad and\quad } M(a) = M(b) = 0,
\end{align} 
$V(x,h)$ can be represented as 
\begin{equation}
V(x,h)=\inf_{\tau\in\mathcal{T}_x}\E_x\left[\int_0^{\tau}e^{-2ru}f(S_u,h)\ud u+e^{-2r\tau}M(S_\tau)\right].\label{2.5}
\end{equation}
While \eqref{2.5} defines a standard optimal stopping problem, the explicit expression of $M$ is extremely convoluted and makes the analysis of our problem very challenging. Indeed, it immediately follows from \eqref{2.3} and \eqref{2.4} that
\begin{equation}\label{Mprobability}
M(x) = -\frac{\gamma^2_2(x)}{\gamma_1(x)} + \gamma_3(x),\qquad x\in\cI.
\end{equation}
However, thanks to the analytical expressions we can easily assert the smoothness of $\Gamma$ and $M$ in $\cI$ and their behaviour at the boundary $\partial\cI$. 
\begin{proposition}
\label{prop4.1}
The optimal stock holding $\Gamma$ and the payoff function $M$ belong to $C^\infty(\cI)$. Furthermore,
\begin{itemize}
\item[(i)]
$\Gamma$ is negative, strictly increasing, with bounded first derivative on $\cI$. The limits of $\Gamma$ at $a$ and $b$ satisfy
\begin{equation}\label{eq:limG}
\Gamma(a):=\lim_{x\downarrow a}\Gamma(x)>P'(a)\quad\text{and}\quad \Gamma(b):=\lim_{x\uparrow b}\Gamma(x)<P'(b).        
\end{equation}
\item[(ii)]
Limits of the derivatives $M'$, $M''$ at $a$ and $b$ exist and are finite. Moreover,
\begin{align}\label{eq:limM}
\lim_{x\downarrow a}M(x)=\lim_{x\uparrow b}M(x)=0.
\end{align}
\end{itemize}
\end{proposition}
\begin{proof}
The smoothness of $\Gamma$ and $M$ on $\cI$ can be checked directly from their explicit expressions \eqref{2.4} and \eqref{Mprobability}.

The monotonicity of $\Gamma$ in (i) is hard to obtain directly from its analytical expression \eqref{2.4} with $\gamma_1, \gamma_2$ as in \eqref{eq:gamma1}-\eqref{eq:gamma2}. Instead we exploit the probabilistic formulae for $\gamma_i$'s given in \eqref{eq:Gamma}, combined with \eqref{diffusionformula}. It can be easily verified that
\begin{equation}\label{eq:gamma}
\Gamma(x)=\frac{\varphi(x)p_2(x)+\psi(x)p_4(x)}{\varphi(x)p_1(x)+\psi(x)p_3(x)},
\end{equation}
where
\begin{align*} 
p_1(x)&=\int_{a}^{x}\psi(z)\sigma^2z^2m'(z)\ud z,\quad\quad p_2(x)=\int_{a}^{x}\psi(z)P'(z)\sigma^2z^2m'(z)\ud z,\\
p_3(x)&=\int_{x}^{b}\varphi(z)\sigma^2z^2m'(z)\ud z,\quad\quad p_4(x)=\int_{x}^{b}\varphi(z)P'(z)\sigma^2z^2m'(z)\ud z.
\end{align*} 
From \eqref{eq:gamma} using simple algebra, we obtain
\begin{align}\label{cancel1}
\Gamma'(x)&=\frac{\wron s'(x)}{\big(\varphi(x)p_1(x)+\psi(x)p_3(x)\big)^2}\big(p_1(x)p_4(x)-p_2(x)p_3(x)\big),
\end{align}
where $w$ is the Wronskian and in the calculations we have used
\begin{equation*}
\psi(x) p_4'(x)=-\varphi(x) p_2'(x), \quad \psi(x) p_3'(x)=-\varphi(x) p_1'(x).
\end{equation*}
Since $P'(\,\cdot\,)$ is strictly increasing, we have
\begin{align*}
p_4(x)&>P'(x)\int_{x}^{b}\varphi(z)\sigma^2z^2m'(z)\ud z=P'(x)p_3(x),\\
p_2(x)&<P'(x)\int_{a}^{x}\psi(z)\sigma^2z^2m'(z)\ud z=P'(x)p_1(x). 
\end{align*}
Therefore $p_1(x)p_4(x)>P'(x)p_3(x)p_1(x)>p_2(x)p_3(x)$,
which implies that $\Gamma'(x)>0$. Noticing that 
\[
p_1(a)=p_2(a)=p_3(b)=p_4(b)=0,
\]
\[
p_1'(x)P'(x)=p_2'(x),\quad\text{and}\quad p_3'(x)P'(x)=p_4'(x),
\] 
and using de L'Hospital's rule for the right-hand side of \eqref{cancel1},  we can compute the limits
\begin{align*}
\lim_{x\downarrow a}\Gamma'(x)&=\frac{\wron(p_4(a)-P'(a)p_3(a))}{\psi'(a+)p_3(a)^2}<\infty,\\
\lim_{x\uparrow b}\Gamma'(x)&=\frac{\wron(p_2(b)-P'(b)p_1(b))}{\varphi'(b-)p_1(b)^2}<\infty,
\end{align*}
\begin{equation*}
\lim_{x\downarrow a}\Gamma(x)=\frac{p_4(a)}{p_3(a)}>P'(a)\quad\text{and}\quad\lim_{x\uparrow b}\Gamma(x)=\frac{p_2(b)}{p_1(b)}<P'(b),     
\end{equation*}
which, together with \eqref{cancel1}, concludes the proof of $(i)$.

Now we prove $(ii)$. The boundedness of derivatives follows directly from the explicit representation \eqref{Mprobability}. Limits at $a$ and $b$ are deduced from \eqref{eq:gamma1}-\eqref{eq:gamma3}.
\end{proof}

We close this section by proving the Lipschitz continuity of the value function. Since $M$ is continuous on $\closure{\cI}$, \cite[Theorem 3.4]{palczewski2011stopping} implies that $V$ is continous and the smallest optimal stopping time is in the standard form, i.e., the first hitting time of the set where $V$ coincides with $M$ (see \eqref{eq:tau*} below). However, in the particular case of the optimal stopping problem $V(x,h)$, the Lipschitz continuity, and, therefore, continuity, can be proven directly. Arguments below rely on the Lipschitz continuity of $f$ and $M$ and not on their particular form. Notice that the underlying process is absorbed at $a$ and $b$ which differentiates our setting from results found in the literature.
\begin{proposition}
\label{vlips}
There exists a constant $L$ such that for any $(x,h)$ and $(x',h')$ in $\closure{\cI}\times\cH$
\begin{equation}\label{eq:Lip}
|V(x,h)-V(x',h')|\leq L(|x-x'|+|h-h'|).
\end{equation}
\end{proposition}
\begin{proof}
Take $(x,h)$ and $(x',h')$ in $\closure{\cI}\times\cH$. Let $\tau_1\in \Tx$ be an $\varepsilon$-optimal stopping time for $V(x,h)$ and let $\tilde{\tau}=\tau_1\wedge\tau^{x'}_\cI$, so that $\tilde{\tau}\in \cT_{x'}$. Since $\tilde\tau$ is in general sub-optimal for $V(x',h')$, we have
\begin{align*}
V(x',h')&\leq\E\left[\int_0^{\tilde{\tau}}e^{-2ru}f(S^{x'}_u,h')\ud u+e^{-2r\tilde{\tau}}M(S^{x'}_{\tilde{\tau}})\right]\\
&\leq\E\left[\int_0^{\tau_1}e^{-2ru}f(S^{x'}_u,h')\ud u\right]+\E\left[e^{-2r\tau_1}M(S^{x'}_{\tau_1})1_{\{\tau_1\leq\tau^{x'}_\cI\}}\right],
\end{align*}
where we used that $f\ge 0$ and $M(S^{x'}_{\tau^{x'}_\cI})=0$ by \eqref{eq:limM}.
Now, by direct comparison we obtain
\begin{align}\label{vlipineq1}
&V(x',h')-V(x,h)\\
&\leq\E\left[\int_0^{\tau_1}e^{-2ru}(f(S^{x'}_u,h')-f(S^{x}_u,h))\ud u \right]\nonumber\\
&\quad+\E\left[e^{-2r\tau_1}\big(M(S^{x'}_{\tau_1})-M(S^x_{\tau_1})\big)1_{\{\tau_1\leq\tau^{x'}_\cI\}}-e^{-2r\tau_1}M(S^{x}_{\tau_1})1_{\{\tau_1>\tau^{x'}_I\}}\right] + \varepsilon\nonumber\\
&\leq\E\left[\int_0^{\tau_1}e^{-2ru}|f(S^{x'}_u,h')-f(S^{x}_u,h)|\ud u \right]+\E\left[e^{-2r\tau_1}|M(S^{x'}_{\tau_1})-M(S^{x}_{\tau_1})|1_{\{\tau_1\leq\tau^{x'}_\cI\}}\right] + \varepsilon.\nonumber
\end{align}
The map $(x,h)\mapsto f(x,h)$ is Lipschitz on $K\times\cH$, with $K\subset\R_+$ any compact, and $x\mapsto M(x)$ is also Lipschitz by $(ii)$ in Proposition \ref{prop4.1}. Since $S^x_{t\wedge\tau_1}\in\closure{\cI}$, for all $t\ge 0$, then 
\begin{equation}\label{eqn:Kab}
S^{x'}_{t\wedge\tau_1}=x'/x S^x_{t\wedge\tau_1}\in[a^2/b,b^2/a]=:K_{a,b}.
\end{equation}
Let $L_1, L_2$ be the Lipschitz constants for $f(x,h)$ on $K_{a,b} \times \cH$ and for $M(x)$ on $\closure{\cI}$, respectively. Then, using the explicit expression of $S_t^x$, we can bound \eqref{vlipineq1} with
\begin{align*}
&V(x',h')-V(x,h)\\
&\leq\E\left[\int_0^{\tau_1}e^{-2ru}L_1(|x-x'|S^1_u+|h-h'|)\ud u \right]
+\E\left[e^{-2r\tau_1}L_2|x-x'|S^1_{\tau_1}\right] + \varepsilon\nonumber\\
&\leq (L_1\vee L_2)(|x-x'|+|h-h'|)\left(1+\int_0^{\infty}e^{-ru}\ud u \right) + \varepsilon, \nonumber
\end{align*}
where we used $\E[e^{-r t}S^1_{t}]= 1$ for any $t \ge 0$. Since $\varepsilon > 0$ is arbitrary, we conclude that $V(x',h')-V(x,h) \le (1 + 1/r) (L_1\vee L_2)(|x-x'|+|h-h'|)$. 
A symmetric argument leads to the reverse inequality and \eqref{eq:Lip} is proven with $L=(1 + 1/r)(L_1\vee L_2)$.
\end{proof}

We note here for future use that
\begin{align}\label{eq:V0}
V(a,h)=V(b,h)=0.
\end{align} 

Thanks to the reduction to a standard Markovian setup we can introduce the continuation and stopping set of problem \eqref{2.5}, denoted respectively by $\cC$ and $\cD$, and defined as
\begin{align}
\label{eq:C}&\cC:=\{(x,h)\in \closure{\cI}\times \cH: V(x,h)<M(x)\},\\
\label{eq:D}&\cD:=\{(x,h)\in \closure{\cI}\times \cH: V(x,h)=M(x)\}.
\end{align}
Obviously, we have $\{a,b\}\times\cH \subset\cD$ due to \eqref{eq:V0}.
It is well known (see, e.g., \cite[Chapter \rom{1}, Corollary 2.9]{peskir2006optimal}) that the minimal optimal stopping time in \eqref{2.5} is 
\begin{equation}\label{eq:tau*}
\tau_{x,h}^*:=\inf\{t\geq 0: (S^x_t,h)\in \cD\}.
\end{equation}
For simplicity, in the rest of the paper we also use the notation $\tau^*_{h}=\tau^*_{x,h}$ under $\P_x$.

The slightly odd aspect of \eqref{eq:tau*} is that the two dimensional process $(S,h)$ is actually constant in its second coordinate. This motivates introducing the sets
\begin{align*}
&\cC_{h}:=\{x\in \closure{\cI}: V(x,h)<M(x)\},\\
&\cD_h:=\{x\in \closure{\cI}: V(x,h)=M(x)\},
\end{align*}
for each $h\in\cH$. In terms of these two sets, the optimal stopping time \eqref{eq:tau*} reads 
\begin{equation}\label{eq:tauh}
\tau^*_{x,h}:=\inf\{t\geq 0: S^x_t\in \cD_{h}\}.
\end{equation}
Since functions $M,V$ are continuous, the sets $\cC$ and $\cC_{h}$ are open whereas $\cD$ and $\cD_{h}$ are closed.  

Finally, letting 
\begin{align}
Y^h_t:= e^{-2r {(t\wedge\tau_\cI)}}V(S_{t\wedge\tau_\cI},h)+\int_0^{t\wedge\tau_\cI}e^{-2rs}f(S_s,h)\ud s
\end{align}
we have that, for any $(x,h)\in \cI \times \cH$, the process $(Y^h_t)_{t\ge 0}$ is a $\P^x$-sub-martingale and 
\begin{align}\label{eq:mart}
\text{the process}\:(Y^h_{t\wedge\tau^*_h})_{t\ge 0}\: \text{is a $\P^x$-martingale}.
\end{align}

\subsection{A free boundary problem}\label{sec2.2}
It is expected that, for each $h\in\cH$, the stopping problem \eqref{2.5} be linked to an obstacle problem
\begin{align}
\label{eq:ObsP}&\min\left\{(\cL -2r)u+f,M-u\right\}(x,h)=0,\quad \text{a.e.}~x\in\cI,\\
\label{eq:ObsP2}&u(a,h)=u(b,h)=0.
\end{align}
This problem can be stated as the following free boundary problem
\begin{align}
\label{eq:FBP1}&(\cL -2r)u(x,h)+f(x,h)=0,\quad x\in\{z \in \cI: u(z,h)<M(z)\},\\
\label{eq:FBP2}&(\cL -2r)u(x,h)+f(x,h)\ge 0,\quad \text{a.e. } x\in\cI,\\
\label{eq:FBP3}&u(x,h) \le M(x), \quad x \in \cI, \qquad u(a,h)=u(b,h)=0.
\end{align}
It is also often postulated that the so-called smooth-pasting condition holds, i.e.,
\begin{align}\label{eq:smooth}
\partial_x u(\cdot,h)=M'(\cdot)\quad\text{on}\quad \partial\{z \in \cI:u(z,h)<M(z)\}.
\end{align}

In the literature on one dimensional optimal stopping problems the obstacle problem \eqref{eq:ObsP} is usually solved in its form \eqref{eq:FBP1}--\eqref{eq:FBP3} by first making an educated guess on the shape of the set $\{z \in \cI:u(z,h)<M(z)\}$ and then by solving the corresponding boundary value problem \eqref{eq:FBP1}. The solution of the resulting ODE can be often computed explicitly and the smooth pasting \eqref{eq:smooth} is used to determine the boundary $\partial\{z \in \cI:u(z,h)<M(z)\}$. The latter normally relies on finding roots of nontrivial algebraic equations. Finally, one verifies \eqref{eq:FBP2}-\eqref{eq:FBP3}. 

Since the payoff function $M(x)$ has a very complicated form, the approach sketched above is infeasible, particularly, the verification of \eqref{eq:FBP2}-\eqref{eq:FBP3} from the smooth-pasting condition. Instead, we follow a mixed probabilistic/analytic approach. In this section we determine the shape of the continuation set, while in Section \ref{sec:2.3} we prove the smoothness of the value function and determine in what sense it solves the obstacle problem \eqref{eq:ObsP}-\eqref{eq:ObsP2}.

It is well-known that one can gain insights into the geometry of the stopping set $\cD$ by studying the sign of the function $G: \cI\times\cH\mapsto \R$ defined as
\begin{equation}
\label{G}
G(x,h):=(\mathcal{L}-2r)M(x)+f(x,h),
\end{equation}
where $\cL M$ is well-defined thanks to Proposition \ref{prop4.1}.
\begin{lemma}
\label{propcalculateG}
For each $h\in \cH$, 
\begin{align}\label{eq:simple}
\{x\in\cI:G(x,h)<0\} \subset \mathcal{C}_h.
\end{align}
\end{lemma}
\begin{proof}
The proof of \eqref{eq:simple} is standard but we present arguments for the convenience of the reader. For a fixed $h$, assume there is $\hat{x}\in \cI$ such that $G(\hat{x},h)<0$ and let 
\[
\tau_0:=\inf\{t\geq0 : G(S_t,h)\geq0\}\wedge \tau_\cI.
\]
Then $\tau_0>0$, $\P_{\hat{x}}$-a.s., by continuity of $G$ and $t\mapsto S_t$. Since $M\in C^2(\cI)$ with bounded derivatives (Proposition \ref{prop4.1} $(ii)$), by an application of Dynkin's formula we have
\begin{align*}
V(\hat{x},h)&\leq\E_{\hat{x}}\left[\int_0^{\tau_0}e^{-2ru}f(S_u,h)\ud u+e^{-2r\tau_0}M(S_{\tau_0})\right]\\
      &=\E_{\hat{x}}\left[\int_0^{\tau_0}e^{-2ru}G(S_u,h)\ud u\right]+M(\hat x)<M(\hat{x}),
\end{align*}
hence $\hat x\in\cC_h$.
\end{proof}

The following lemma provides an explicit expression for $G$.
\begin{lemma}
\label{corobdedM}
For all $(x,h)\in\cI\times\cH$ we have
\begin{align}
\label{2.6}
G(x,h)&=\sigma^2x^2\Bigl((h-P'(x))^2 -(\Gamma(x)-P'(x))^2-(\Gamma'(x))^2\gamma_1(x)\Bigr).
\end{align}
\end{lemma}
\begin{proof}
Using \eqref{Mprobability}, we obtain
\begin{align}
\label{G1}
G(x,h)&=(\Gamma(x))^2(\mathcal{L}\gamma_1-2r\gamma_1)(x)-2\Gamma(x)(\mathcal{L}\gamma_2-2r\gamma_2)(x)+(\mathcal{L}\gamma_3-2r\gamma_3)(x)\\
&\quad+\sigma^2 x^2\Gamma'(x)(\Gamma(x)\gamma_1'(x)-\gamma_2'(x))+\sigma^2x^2(h-P'(x))^2.\nonumber
\end{align}
Recall the probabilistic expressions for $\gamma_1$, $\gamma_2$ and $\gamma_3$ given in \eqref{2.4} and \eqref{eq:gamma3}. Hence, for $i=1,2,3$,
\begin{equation}\label{eq:Lgamma}
(\mathcal{L}-2r)\gamma_i(x)=-g_i(x),\quad\text{on $\cI$,}
\end{equation} 
where $g_1(x)=\sigma^2x^2$, $g_2(x)=\sigma^2 x^2 P'(x)$ and $g_3(x)=\sigma^2 x^2(P'(x))^2$.
Furthermore,
\begin{align}\label{eq:GG}
\Gamma(x)\gamma_1'(x)-\gamma_2'(x)=-\Gamma'(x)\gamma_1(x),
\end{align}
since $\Gamma(x)=\gamma_2(x)/\gamma_1(x)$. Finally, inserting \eqref{eq:Lgamma} and \eqref{eq:GG} into \eqref{G1} yields \eqref{2.6}.
\end{proof}

Next we proceed to prove that the continuation and the stopping sets have non-empty intersection with $\cI$ (recall that $\{a, b\} \in \cD_h$). For any $h\in \cH$, it is convenient to define $x_p(h)\in \closure{\cI}$ as the unique root of the equation
$P'(x)-h=0$, that is, 
\begin{equation}
\label{2.7}
x_p(h)=\am(-h)^{-\frac{1}{1+d}}.
\end{equation}
Notice that for $h\in\text{int}(\cH)$ we have $x_p(h)\in\cI$.

\begin{proposition}
\label{prop4.202}
For each $h\in \cH$, we have $\cD_h\cap\cI\neq\varnothing$ and $\mathcal{C}_h\not=\varnothing$.
\end{proposition}
\begin{proof}
First consider $h\in (P'(a),P'(b))$. Then $x_p(h)\in \cI$ and it is immediate to see from \eqref{2.6} that $G(x_p(h),h)<0$. Hence, \eqref{eq:simple} implies $\cC_h\neq\varnothing$. If $h=P'(a)$ then the expression \eqref{2.6} and the fact that $\Gamma(a) > P'(a)$ (Proposition \ref{prop4.1}) imply $G(a, P'(a)) < 0$. By the continuity of $G$, there is $x \in \cI$ with $G(x, P'(a)) < 0$ and an application of \eqref{eq:simple} gives $\cC_h\neq\varnothing$. A similar argument applies for $h = P'(b)$.
 
Assume now that $\cD_h\setminus\{a,b\}=\varnothing$ so that $\mathcal{C}_h=\cI$. Then for any $x\in\cI$ we have 
\begin{align*}
M(x)>V(x,h)=\E_{x}\left[\int_{0}^{\tau_\cI}e^{-2ru}f(S_u,h)\ud u\right] = \widehat M(x, h)
\ge \inf_{l\in \cH} \widehat{M}(x,l) = M(x), 
\end{align*}
hence a contradiction.
\end{proof}

The subsequent analysis will show that the roots of the map $x\mapsto G(x,h)$ for each $h\in\cH$ determine the shape of the continuation and the stopping sets. Due to the complexity of the expression for $G$, it seems very hard to determine analytically the exact number of zeros of the map $x\mapsto G(x,h)$. However, the exercise is trivial from a numerical point of view, thanks to the fully explicit expression in \eqref{2.6}. We performed extensive numerical tests and observed only three possible situations displayed in Figure \ref{p2.0}. It will also follow from the proof of Proposition \ref{prop:Gzeros2} that the map $x\mapsto G(x,h)$ has at least one root if $h\in[P'(a),\Gamma(a)]\cup [\Gamma(b),P'(b)]$ and it has at least two roots if $h\in(\Gamma(a),\Gamma(b))$. The following assumption provides a necessary ingredient to determine the exact number of zeros of $G$ and the shape of the stopping set.

\begin{assumption}\label{ass:zeros1}
For each $h\in \cH $, the equation $G(\,\cdot\,,h)=0$ has at most two roots in $\cI$.
\end{assumption}
\begin{figure}[tb]
\centering
\includegraphics[width=7cm]{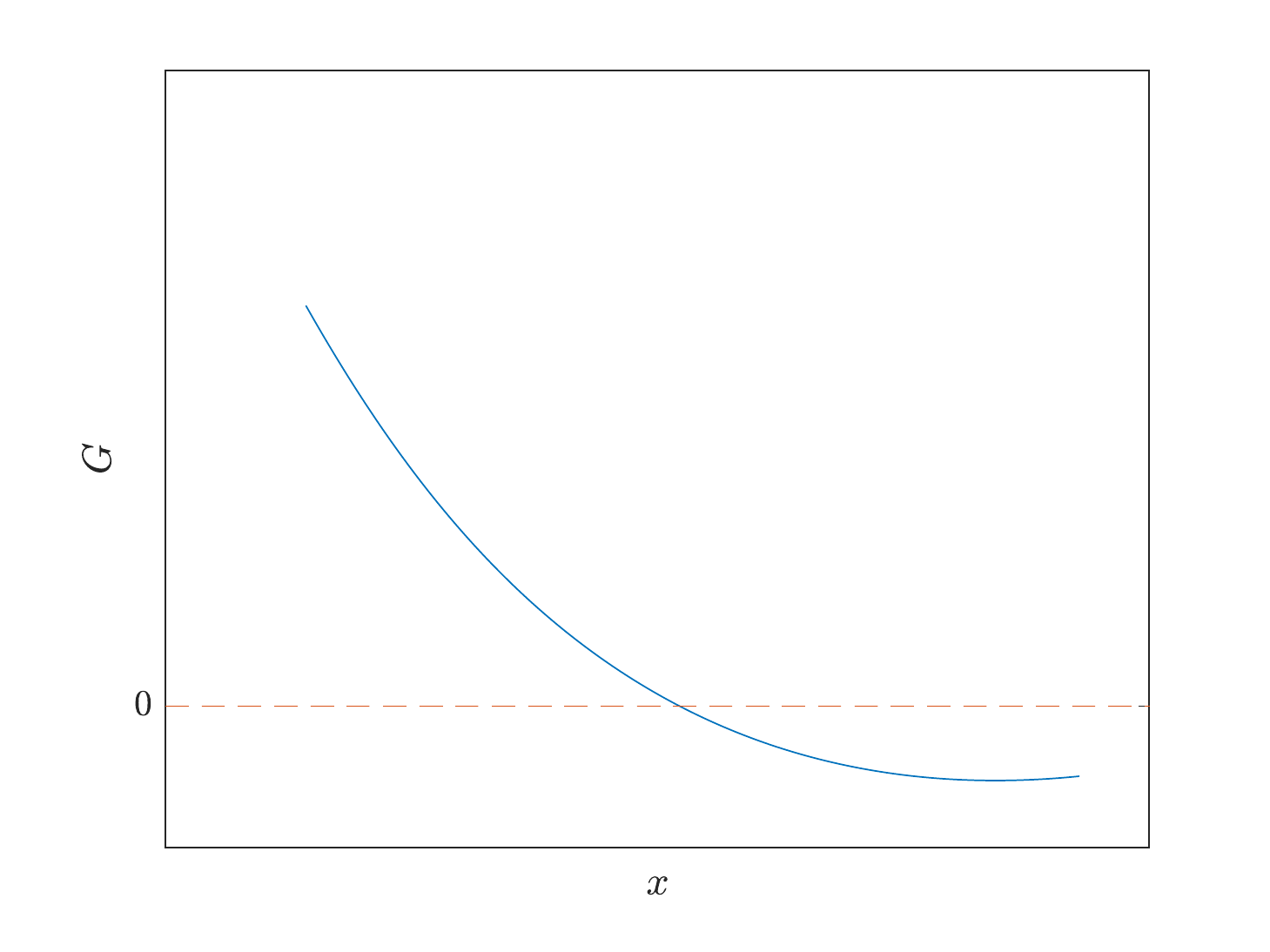}
\includegraphics[width=7cm]{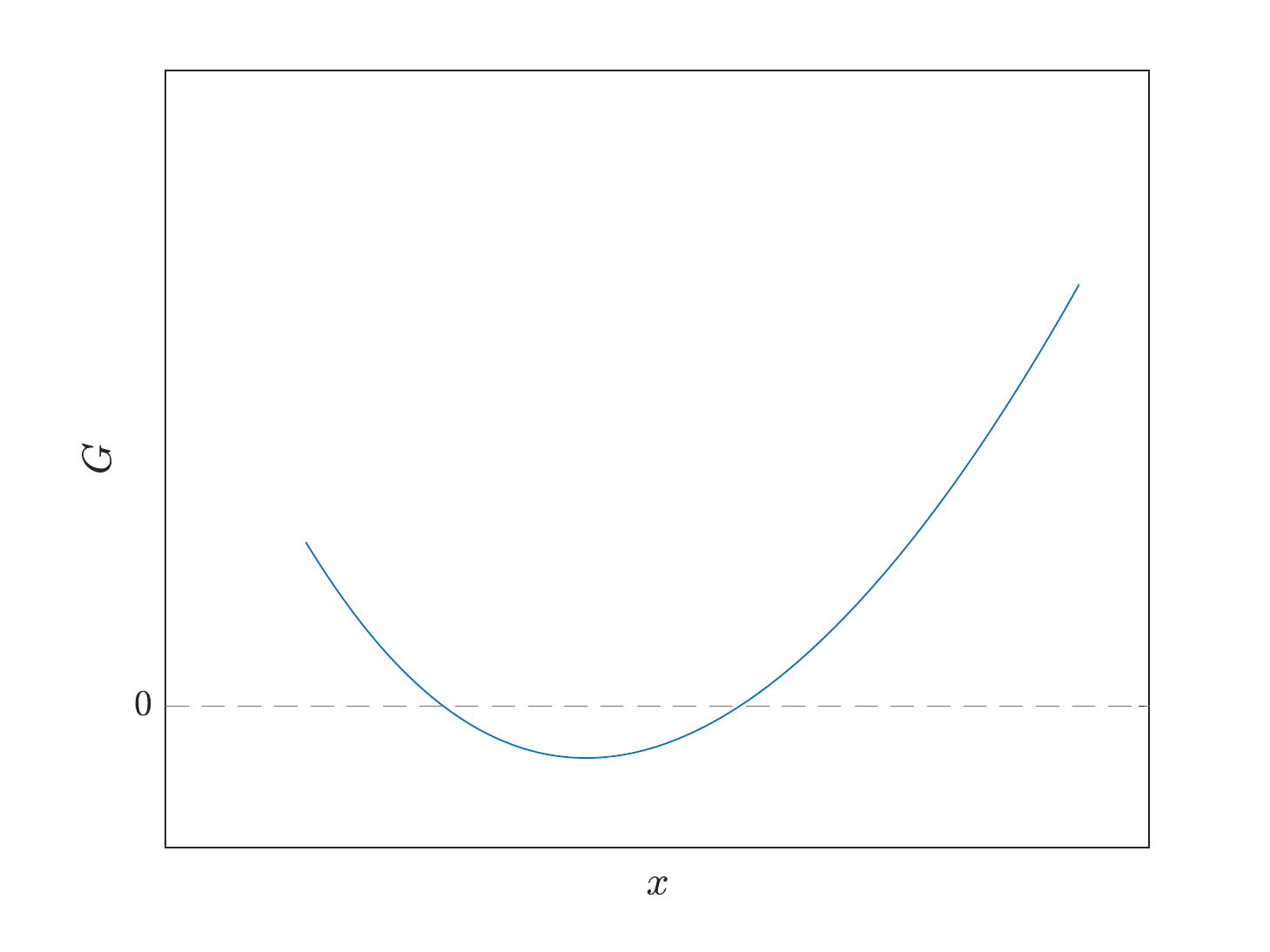}
\includegraphics[width=7cm]{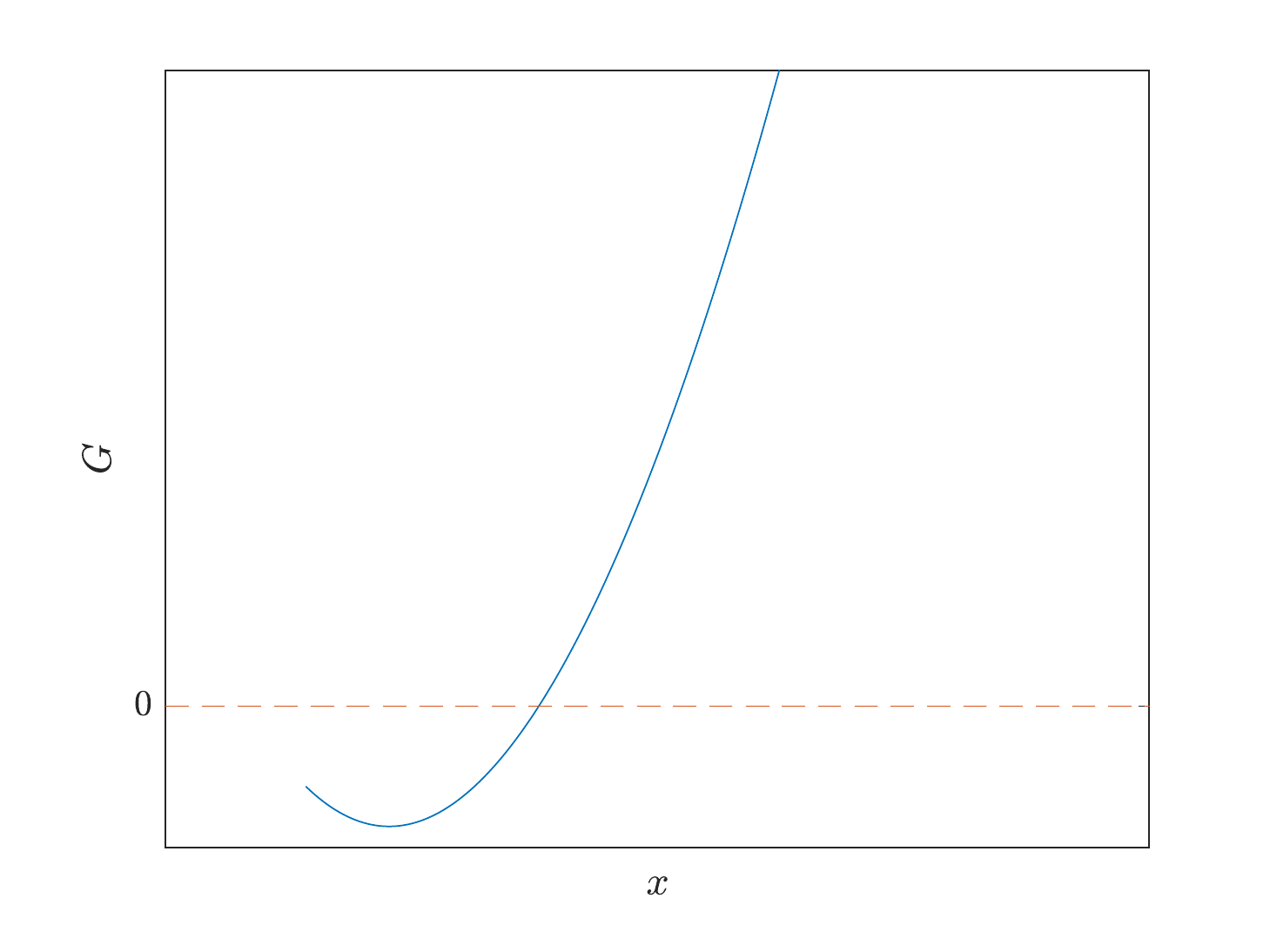}
\caption{Plots of the map $x\mapsto G(x,h)$ for different values of the initial stock holding $h$ using parameters $r=3\%$, $\sigma=30\%$, $K=100$, $b=150$, and $a = \am = K/(1+d^{-1})=40$.} 
\label{p2.0}
\end{figure} 

Denoting the roots of $G(\,\cdot\,,h)=0$ on $\cI$ by $x_{G_1}$ and $x_{G_2}$ (when they both exist) consider the following three cases:
\begin{itemize}
\item[(A.1)] $G(x,h)>0$ for $x\!\in(a,x_{G_1})$ and $G(x,h)<0$ for $x\!\in\!(x_{G_1},b)$, except possibly at $x_{G_2}$;
\item[(A.2)] $G(x,h)>0$ for $x\in(a,x_{G_1})\cup(x_{G_2},b)$ and $G(x,h)<0$ for $x\in(x_{G_1},x_{G_2})$;
\item[(A.3)] $G(x,h)<0$ for $x\in(a,x_{G_1})$ and $G(x,h)>0$ for $x\in(x_{G_1},b)$, except possibly at $x_{G_2}$.
\end{itemize}
\begin{remark}
In (A.1) we mean that, if $G(\,\cdot\,,h)$ has two roots then it must be $\partial_x{G}(x_{G_2},h)=0$. The root $x_{G_2}$ may be on the right or on the left of $x_{G_1}$. An analogous rationale holds in (A.3).
\end{remark}
It turns out that the above cases (A.1)-(A.3) are uniquely linked to the choice of the initial stock holding $h$ as the following proposition demonstrates.
\begin{proposition}
\label{prop:Gzeros2}
Under Assumption \ref{ass:zeros1}, we have:
\begin{itemize}
\item[(i)] Condition (A.1) holds if and only if $h\in[\Gamma(b),P'(b)]$;
\item[(ii)] Condition (A.2) holds if and only if $h\in(\Gamma(a),\Gamma(b))$;
\item[(iii)] Condition (A.3) holds if and only if $h\in[P'(a),\Gamma(a)]$.
\end{itemize}
\end{proposition}
\begin{proof}
Assume (A.1) and $h \in \cH$. Then $G(b,h) \le 0$. Using $\gamma_1(b) = 0$ in \eqref{2.6}, we obtain
\[
(h-P'(b))^2 -(\Gamma(b)-P'(b))^2 \le 0,
\]
which yields $h \ge \Gamma(b)$ and completes the proof of the right implication in (i). 

Consider now $h\in[\Gamma(b),P'(b)]$. Directly from \eqref{2.6} we calculate $G(a,h) > 0$ since $h > \Gamma(a) > P'(a)$ and $\gamma_1(a) = 0$. For $h > \Gamma(b)$ we have $G(b,h) < 0$, which combined with Assumption \ref{ass:zeros1} and the continuity of $G$ proves (A.1). For $h=\Gamma(b)$ we have to use a different argument because $G(b, \Gamma(b)) = 0$. Rewriting \eqref{2.6} yields
\[
G(x,\Gamma(b))=\sigma^2x^2\Bigl((\Gamma(b)-\Gamma(x))(\Gamma(b)+\Gamma(x) - 2P'(x))-(\Gamma'(x))^2\gamma_1(x)\Bigr).
\]
The last term in the bracket is non-positive. We have $\Gamma(b)+\Gamma(x) - 2P'(x) < 0$ for $x \in \cI$ sufficiently close to $b$, and, $\Gamma(b)-\Gamma(x) > 0$ by the monotonicity of $\Gamma$. Hence, $G(x, \Gamma(b)) < 0$ for $x \in \cI$ sufficiently close to $b$, which immediately proves (A.1).

Assume now (A.2). Using arguments from the beginning of the proof, $G(b,h) > 0$ implies $h < \Gamma(b)$. Analogously, $G(a, h) > 0$ implies $h > \Gamma(a)$. For the left implication in (ii), we note that $G(x_p(h), h) < 0$ for $h \in \operatorname{int}(\cH)$. The sign of $G(x, h)$ at $x \in \{a, b\}$ is determined by the sign of
\[
(h-P'(x))^2-(\Gamma(x)-P'(x))^2=(h-\Gamma(x))(h+\Gamma(x)-2P'(x)).
\] 
Recalling that $P'(a) < \Gamma(a) < \Gamma(b) < P'(b)$ (c.f. \eqref{eq:limG}), we have $G(a, h) > 0$ and $G(b, h) > 0$ for $h \in (\Gamma(a), \Gamma(b))$. As above, the continuity of $G$ and Assumption \ref{ass:zeros1} completes the proof of the left implication in (ii).

The proof of (iii) is analogous to (i).
\end{proof}

In light of the above proposition, we will refer to conditions (A.1)--(A.3) as determining the ranges of $h$ as well as the zeros of $G(x,h)$. We now show that they are sufficient to determine shapes of the continuation and the stopping sets $\cC_h$ and $\cD_h$.
\begin{proposition}
\label{prop4.33}
Let Assumption \ref{ass:zeros1} hold and take $h\in\cH$. Then we have 
\begin{enumerate}
\item[(i)] under (A.1) there is $x^*_1\in (a, x_{G_1}]$ such that $\cC_{h}=(x^{*}_1,b)$; 
\item[(ii)] under (A.2) there exist $x^*_1 \in [a,x_{G_1}]$ and $x^*_2 \in [x_{G_2}, b]$ such that $\cC_{h}=(x^{*}_1,x^{*}_2)$. Moreover, at least one of $x^*_1, x^*_2$ is in $\cI$;  
\item[(iii)] under (A.3) there is $x^*_2\in [x_{G_1}, b)$ such that $\cC_{h}=(a,x^{*}_2)$.
\end{enumerate}
\end{proposition}
\begin{proof}
We only give a full proof of $(iii)$ as the other claims follow by analogous arguments. Assume (A.3) and that the root $x_{G_2}$ exists and is smaller than $x_{G_1}$. Inclusion \eqref{eq:simple} implies $\cD_h\cap\cI\subseteq \{x_{G_2}\}\cup[x_{G_1},b)$. We will show that $x_{G_2}\notin \cD_h$. Indeed, for a small $\eps>0$, let 
\[
\tau_\eps:=\inf\{t\geq0: S_t\notin (x_{G_2}-\eps,x_{G_2}+\eps)\}.
\]
Since $\tau_\eps$ is sub-optimal for $V(x_{G_2},h)$, we have
\begin{align*}
V(x_{G_2},h)&\leq\E_{x_{G_2}}\left[\int_0^{\tau_\eps}e^{-2ru}f(S_u,h)\ud u+e^{-2r\tau_\eps}M(S_{\tau_\eps})\right]\\
&=\E_{x_{G_2}}\left[\int_0^{\tau_\eps}e^{-2ru}G(S_u,h)\ud u\right]+M(x_{G_2})< M(x_{G_2}),
\end{align*}
where the equality is an application of Dynkin formula for $M(S_{\tau_\eps})$ and the final strict inequality holds because, under (A.3), we have $G(x,h)<0$ on $(x_{G_2}-\eps,x_{G_2}+\eps)\setminus\{x_{G_2}\}$ for a sufficiently small $\eps$ and $G(x_{G_2},h)=0$. This shows that $x_{G_2}\notin \cD_h$ and therefore $\cD_h\cap\cI\subseteq[x_{G_1},b)$. The latter inclusion trivially holds if $x_{G_2} > x_{G_1}$ or when the second root $x_{G_2}$ does not exist.

Next we show that if $x_0\in[x_{G_1},b)$ and $x_0\in \cD_h$, then $[x_0,b]\subseteq\cD_h$.
Arguing by contradiction, assume there exists such an $x_0$ and an open set $U\subset (x_0,b)$ such that $U\subset \cC_h$. For any $x\in U$, we have 
\[
\tau^*_{x,h}\le \inf\{t\geq0: S^x_t\le x_0\}, \quad\P-a.s. 
\]
Applying Dynkin formula, we obtain
\begin{align*}
V(x,h)&=\E_{x}\left[\int_0^{\tau^*_h}e^{-2ru}f(S_u,h)\ud u+e^{-2r\tau^*_h}M(S_{\tau^*_h})\right]\\
&=\E_x\left[\int_0^{\tau_h^*}e^{-2ru}G(S_u,h)\ud u\right]+M(x)\ge M(x),  
\end{align*}
where the final inequality is due to $G(x,h)\ge 0$ on $(x_0,b)$. Hence a contradiction. Notice that the existence of $x_0\in [x_{G_1},b)$ such that $x_0\in \cD_h$ is guaranteed by $\cD_h \cap \cI \ne \varnothing$ (Proposition \ref{prop4.202}).
\end{proof}

The above proposition shows that under (A.1) and (A.3) the shape of the stopping set is unambiguously determined. Only under (A.2), the set $\cD_h \cap \cI$ may have one or two connected components, depending on the choice of the parameters in the problem.

\subsection{Solution of the free boundary problem}\label{sec:2.3}
Showing that the value function $V$ is a solution to the free boundary problem \eqref{eq:FBP1}-\eqref{eq:FBP3} is relatively easy. However, this provides little value unless one can further ascertain uniqueness. This is done via a verification argument, which typically requires smooth pasting across stopping boundaries. Smooth pasting is also required for efficient calculation of stopping boundaries via a solution of algebraic equations, see Subsection \ref{subsec:algebraic_eqn}. In this section we first show that the value function $V$ of \eqref{2.5} satisfies $V(\,\cdot\,,h)\in C^1(\cI)$ for each $h\in\cH$ (i.e., smooth pasting), then we use this fact to prove that $V$ solves \eqref{eq:FBP1}-\eqref{eq:FBP3} uniquely (Theorem \ref{thm:VI}).

We can immediately claim that $V(\,\cdot\,,h)\in C^2(\cI\setminus\partial\cC_h)$. Indeed, on $\cD_h\setminus\partial\cC_h$, $V=M$, so the result is trivial by $(ii)$ in Proposition \ref{prop4.1}. Instead, on $\cC_h$, the result follows by \eqref{eq:mart} and a standard argument \cite[Chapter 4.2]{karatzas2012brownian} (see also \cite[Chapter III, Section 7]{peskir2006optimal}). Hence, for any $h\in\cH$, $V$ is a classical solution of
\begin{align}\label{eq:FBPV}
&(\cL-2r)V(x,h)=-f(x,h),\qquad x\in\cC_h,
\end{align}
with the boundary condition $V(x,h)=M(x)$ for $x\in\partial\cC_h$.

The difficulty lies in showing the regularity of the value function across the boundary. For that we will use the following well known fact. For any open interval $\cO\subset\R_+$, letting
\begin{align}\label{eq:tauO}
\tau^x_\cO:=\inf\{t\ge 0: S^x_t\notin\cO\}\quad\text{and}\quad \hat \tau^x_{ \cO}:=\inf\{t \ge 0: S^x_t\notin \closure\cO\}
\end{align}
we have (see, \cite[Chapter V, Lemma 46.1]{rogers1994diffusions})
\begin{equation}\label{eq:tauOa}
\P(\tau^x_\cO=\hat \tau^x_{\cO})=1,\quad\text{for all $x\in\R_+$}.
\end{equation}

This fact and a well-known argument based on properties of the process sample paths, guarantee that 
if a sequence $(x_n)_{n\ge 0}\in\R_+ $ converges to $x_0\in\R_+$ as $n\rightarrow \infty$, then 
\begin{align}\label{eq:conv0}
\tau^{x_n}_{\cO}\to \tau^{x_0}_{\cO},\qquad\P-a.s.
\end{align}
The proof of \eqref{eq:conv0} is an easier version of the one we give for the continuity of optimal stopping times in Theorem \ref{prop5.5.1}, hence we omit it here.
In particular, under Assumption \ref{ass:zeros1} and using Proposition \ref{prop4.33}, this implies that for any sequence $(x_n)_{n\ge 0}\in\cC_h$ converging to $x_0\in\partial{\cC_h}$ as $n\rightarrow \infty$, we have
\begin{align}\label{eq:conv}
\tau^*_{x_n,h}\to 0,\qquad\P-a.s.
\end{align}
This is the key tool to the next result, which makes use of an approach developed in \cite{de2018global}.

\begin{theorem} \label{thm:C1}
Under Assumption \ref{ass:zeros1} we have, for each $h\in\cH$, 
\[
V(\cdot,h)\in C(\closure{\cI})\cap C^1(\cI)\cap C^2(\cI\setminus\partial\cC_h)
\]
and for any $x_0\in\partial\cC_h\cap\cI$
\begin{align}\label{eq:C2}
\lim_{\cC_h\ni x\to x_0}\partial_{xx}V(x,h)=2(\sigma x_0)^{-2}\left(-rx_0 M'(x_0)+2r M(x_0)-f(x_0,h)\right).
\end{align}
\end{theorem}
\begin{proof}
The continuity of $V(\cdot, h)$ follows from Proposition \ref{vlips}, whereas \eqref{eq:C2} can be obtained from \eqref{eq:FBPV} provided that $V(\cdot\,,h)\in C^1(\cI)$. Hence, it only remains to show that for any $x_0\in\partial\cC_h\cap\cI$ it holds
\[
\lim_{\cC_h\ni x\to x_0}\partial_{x}V(x,h)=M'(x_0).
\]

\newcommand{\tauEps}{\tau^{x+\eps}_\cI}
Fix $x\in\cC_h$ and denote $\tau^*:=\tau^*_{x,h}$ which is optimal for the problem $V(x,h)$. Fix $\eps>0$ and
notice that the stopping time $\tau^*\wedge \tauEps \in \cT_{x+\eps}$ is admissible for the problem $V(x+\eps,h)$. We get an upper bound 
\begin{equation*}
V(x+\eps,h)\leq\E\left[\int_0^{\tau^*\wedge \tauEps}\!\!e^{-2ru}f(S^{x+\eps}_u,h)\ud u+e^{-2r (\tau^*\wedge \tauEps)}M(S^{x+\eps}_{\tau^*\wedge \tauEps})\right].
\end{equation*}
Using this and the optimality of $\tau^*$ for $V(x,h)$ we obtain
\begin{align*}
&\frac{V(x+\eps,h)-V(x,h)}{\eps}\\
&\leq\frac{1}{\eps}\E\left[\int_0^{\tau^*\wedge\tauEps}\!\!e^{-2ru}\big(f(S^{x+\eps}_u,h)\!-\!f(S^{x}_u,h)\big)\ud u +e^{-2r\tau^*}\big(M(S^{x+\eps}_{\tau^*})\!-\!M(S^{x}_{\tau^*})\big)1_{\{\tau^*\le\tauEps\}}\right]\notag\\
&\quad-\frac{1}{\eps}\E\left[\left(\int_{\tauEps}^{\tau^*}\!\!e^{-2ru}f(S^{x}_u,h)\ud u+e^{-2r\tau^*}M(S^{x}_{\tau^*})\right)1_{\{\tau^* >\tauEps\}}\right]\\
&\leq\frac{1}{\eps}\E\left[\int_0^{\tau^*\wedge\tauEps}\!\!e^{-2ru}\big(f(S^{x+\eps}_u,h)\!-\!f(S^{x}_u,h)\big)\ud u +e^{-2r\tau^*}\big(M(S^{x+\eps}_{\tau^*})\!-\!M(S^{x}_{\tau^*})\big)1_{\{\tau^*\le\tauEps\}}\right],
\end{align*}
where in the first inequality we also use $M(S^{x+\eps}_{\tauEps})=0$, $\P$-a.s., by \eqref{eq:limM}, and the second inequality follows from $f\ge 0$ and $M\ge 0$. The final term in the last inequality can be further estimated by 
\begin{align*}
&\E\left[e^{-2r\tau^*}\big(M(S^{x+\eps}_{\tau^*})-M(S^{x}_{\tau^*})\big)1_{\{\tau^*\leq\tauEps\}}\right]\\
&= \E\left[e^{-2r\tau^*}\big(M(S^{x+\eps}_{\tau^*})-M(S^{x}_{\tau^*})\big)1_{\{\tau^*\leq\tauEps\}\cap\{\tau^*<\tau^x_{\cI}\}}\right]\\
&\quad+\E\left[e^{-2r\tau^*}\big(M(S^{x+\eps}_{\tau^*})-M(S^{x}_{\tau^*})\big)1_{\{\tau^*\leq\tauEps\}\cap\{\tau^*=\tau^x_{\cI}\}}\right]\\
&=\E\left[e^{-2r\tau^*}\big(M(S^{x+\eps}_{\tau^*})-M(S^{x}_{\tau^*})\big)1_{\{\tau^*\leq\tauEps\}\cap\{\tau^*<\tau^x_{\cI}\}}\right]\\
&\quad+\E\left[e^{-2r\tau^*}\big(M(S^{x+\eps}_{\tau^*})-M(S^{x}_{\tau^*})\big)1_{\{\tau^*\leq\tauEps\}\cap\{\tau^*=\tau^x_{\cI}\}\cap\{S^x_{\tau_\cI}=a\}}\right],
\end{align*}
where we use $\{\tau^*\leq\tauEps\}\cap\{\tau^*=\tau^x_{\cI}\}\cap\{S^x_{\tau_\cI}=b\}=\varnothing$ in the second equality. Notice that on $\{\tau^*\leq\tauEps\}\cap\{\tau^*=\tau^x_{\cI}\}\cap\{S^x_{\tau_\cI}=a\}$ we have
\begin{align*}
M(S^{x+\eps}_{\tau^*})-M(S^{x}_{\tau^*})&\le (S^{x+\eps}_{\tau^*}-S^{x}_{\tau^*})\sup_{z\in[a,b]}|M'(z)|\\
&=((1+\eps/x)a-a)\sup_{z\in[a,b]}|M'(z)|\le \eps\sup_{z\in[a,b]}|M'(z)|.
\end{align*}
Hence, 
\begin{align*}
&\E\left[e^{-2r\tau^*}\big(M(S^{x+\eps}_{\tau^*})-M(S^{x}_{\tau^*})\big)1_{\{\tau^*\leq\tauEps\}}\right]\\
&\le \E\left[e^{-2r\tau^*}\big(M(S^{x+\eps}_{\tau^*})-M(S^{x}_{\tau^*})\big)1_{\{\tau^*\leq\tauEps\}\cap\{\tau^*<\tau^x_{\cI}\}}\right]+\eps\sup_{z\in[a,b]}|M'(z)|\P\left(\tau^*=\tau^x_{\cI}\right).
\end{align*}
Then we have 
\begin{align*}
&\frac{V(x+\eps,h)-V(x,h)}{\eps}\\
&\leq\E\left[\int_0^{\tau^*\wedge \tauEps}e^{-2ru}\big(f(S^{x+\eps}_u,h)-f(S^{x}_u,h)\big)\eps^{-1}\ud u\right]\\
&\quad +\E\left[e^{-2r\tau^*}\big(M(S^{x+\eps}_{\tau^*})-M(S^{x}_{\tau^*})\big)\eps^{-1}1_{\{\tau^*\leq\tauEps\}\cap\{\tau^*<\tau^x_{\cI}\}}\right]+ \sup_{z\in[a,b]}|M'(z)|\,\P(\tau^*=\tau^x_{\cI}).
\end{align*}
From \eqref{eq:conv0} we obtain that $\tauEps\rightarrow \tau^{x}_{\cI}$, $\P$-a.s., as $\eps\to 0$. Thus, when $\eps\to 0$ we have 
\[
1_{\{\tau^*\leq\tauEps\}\cap\{\tau^*<\tau^x_{\cI}\}} \to 1_{\{\tau^*<\tau^x_{\cI}\}},\quad\P-a.s.
\]
By the smoothness of $f$ on $K_{a,b}$ (defined in \eqref{eqn:Kab}) and of $M$ on $[a,b]$, we have 
\begin{align*} 
&\int_0^{\tau^*\wedge \tauEps}e^{-2ru}\frac{|f(S^{x+\eps}_u,h)-f(S^{x}_u,h)|}{\eps}\ud u \le \sup_{z\in K_{a,b}}|\partial_x f(z,h)|\int_0^{\tau^*}e^{-(r+\frac{1}{2}\sigma^2) u+\sigma W_u}\ud u,\\
&e^{-2r\tau^*}\frac{|M(S^{x+\eps}_{\tau^*})-M(S^{x}_{\tau^*})|}{\eps}1_{\{\tau^*\leq\tauEps\}\cap\{\tau^*<\tau^x_{\cI}\}}\le \sup_{z\in[a,b]}|M'(z)|.
\end{align*} 
Thus, letting $\eps\to0$ and applying the dominated convergence theorem we get 
\begin{align}
\label{smoothfiteq1}
\partial_x V(x,h)&\le \E\left[\int_0^{\tau^*}e^{-2ru}\ \partial_x f(S^x_u,h)\, S^1_{u}\ud u+e^{-2r\tau^*}M'(S^{x}_{\tau^*})S^1_{\tau^*}1_{\{\tau^*<\tau^x_{\cI}\}}\right]\\
&\quad+\sup_{z\in[a,b]}|M'(z)|\, \P(\tau^*=\tau^x_{\cI}).\notag
\end{align}

Similar arguments, applied to the stopping time $\tau^*\wedge\tau^{x-\eps}_\cI$, which is admissible for $V(x-\eps,h)$, allow us to obtain
\begin{align}\label{eq:Vx}
\partial_x V(x,h) &= \lim_{\eps\to 0}\frac{V(x,h)-V(x-\eps,h)}{\eps}\\
&\geq\E\left[\int_0^{\tau^*}e^{-2ru}\, \partial_x f(S^x_u,h)\,S^1_{u} \ud u+e^{-2r\tau^*}M'(S^{x}_{\tau^*})S^1_{\tau^*}1_{\{\tau^*<\tau^x_{\cI}\}}\right]\notag\\
&\quad-\sup_{z\in[a,b]}|M'(z)|\, \P(\tau^*=\tau^x_{\cI}).\notag
\end{align}

Recall that $\tau^*=\tau^*_{x,h}$ and then let $\cC_h\ni x\to x_0\in\partial\cC_h\cap\cI$. From \eqref{eq:conv} we get
$\P(\tau^*_{x,h}=\tau^x_{\cI})\to 0$ and $1_{\{\tau^*_{x,h}<\tau^x_{\cI}\}}\to 1$, $\P$-a.s. Then using dominated convergence in \eqref{smoothfiteq1} and \eqref{eq:Vx} we obtain
\[
\partial_xV(x,h)\to M'(x_0),\quad\text{as $x\to x_0$},
\]
which concludes the proof.
\end{proof}

\begin{remark}
We could not infer smooth fit across the stopping boundary diectly from \cite{de2018global} because our underlying process is killed at points $a,b$. Instead we adapted the line of arguments in the aforementioned paper and used the particular characteristics of our optimal stopping problem.
\end{remark}

Thanks to the regularity obtained in the theorem above we can rigorously connect the stopping problem \eqref{2.5} to the obstacle problem \eqref{eq:ObsP}--\eqref{eq:ObsP2} (equivalently to the free boundary problem \eqref{eq:FBP1}--\eqref{eq:smooth}).
\begin{theorem}\label{thm:VI}
Let Assumption \ref{ass:zeros1} hold. For each $h\in\cH$ the value function $V(\,\cdot\,,h)$ is the unique solution, in the a.e.~sense, of \eqref{eq:ObsP}--\eqref{eq:ObsP2} (equivalently of \eqref{eq:FBP1}--\eqref{eq:smooth}) in the class of functions $C(\closure{\cI})\cap C^1(\cI)$ whose second order partial derivative lies in $L^\infty_{\ell oc}(\cI)$.
\end{theorem}
\begin{proof}
From Theorem \ref{thm:C1} we know that $V(\,\cdot\,,h)$ has the right regularity. Moreover, $V(\,\cdot\,,h)=M(\,\cdot\,)$ on $\cD_h$, where $(\cL-2r)M\ge -f$ by \eqref{eq:simple}. Then, combining these facts with \eqref{eq:FBPV} we conclude that for any $h\in\cH$ 
\[
\min\{(\cL-2r)V(x,h)+f(x,h), M(x)-V(x,h)\}=0,\quad\text{for $x\in\cI\setminus\partial\cC_h$}
\]
and clearly $V(a,h)=V(b,h)=0$ (cf.~\eqref{eq:V0}). The same argument guarantees that $V(\,\cdot\,,h)$ also solves \eqref{eq:FBP1}--\eqref{eq:FBPV}.
 
Uniqueness of the solution follows by a standard verification argument. Let $u$ be another solution of \eqref{eq:ObsP}--\eqref{eq:ObsP2} in $C(\closure{\cI})\cap C^1(\cI)$ with $u''\in L^\infty_{\ell oc}(\cI)$ (for simplicity of notation we omit $h\in\cH$, given and fixed). Then, by Tanaka's formula and using $(\cL-2r)u\ge -f$, we obtain
\begin{align*}
\E_x\left[e^{-2r \tau}u(S_{\tau})\right]\ge u(x)-\E_x\left[\int_0^{\tau}e^{-2r t}f(S_t,h)\ud t\right],
\end{align*}

for any  stopping time $\tau\in\cT_x$. Rearranging terms and using $u\le M$ we obtain
\begin{align*}
u(x)\le \E_x\left[\int_0^{\tau}e^{-2r t}f(S_t,h)\ud t+e^{-2r \tau}M(S_{\tau})\right].
\end{align*}
Hence $u\le V$. To prove the reverse inequality it is sufficient to choose $\tau=\inf\{t\ge 0: u(S_t)=M(S_t)\}$ and all the inequalities above become equalities.
\end{proof}

\subsection{Analytical formulae}\label{subsec:algebraic_eqn}
We will sketch now the approach we use to compute the value function and the optimal stopping boundaries. Thanks to Proposition \ref{prop4.33} and Theorems \ref{thm:C1} and \ref{thm:VI}, for $h\in\cH$, the value function $V(\,\cdot\,,h)$ is a classical solution of the system
\begin{equation}\label{eq:ODEp}
\begin{cases}
(\mathcal{L}-2r)v(x)=-f(x,h),& x\in (x^*_1,x^*_2),\\[+3pt]
v(x)=M(x),& x=x^*_1, x^*_2,\\[+3pt]
v'(x^*_1)=M'(x^*_1),& \text{if } x^*_1 > a,\\[+3pt]
v'(x^*_2)=M'(x^*_2),& \text{if } x^*_2 < b,\\[+3pt]
a\le x^*_1 < x_2^* \le b, &
\end{cases}
\end{equation}
where $x^*_1, x^*_2$ are the stopping boundaries, i.e., $\cC_h = (x^*_1, x^*_2)$. Conversely, under Assumption \ref{ass:zeros1}, there is at most one solution of the above system. Indeed, using techniques of \cite[Lemma 2.5 and Lemma 2.6]{ruschendorf2008class} we can show that $x^*_1 \le \min(x_{G_1},x_{G_2})$ and $x^*_2 \ge \max(x_{G_1},x_{G_2})$ and $v < M$ on $(x^*_1, x^*_2)$. We further set $v=M$ on $\closure{\cI} \setminus (x^*_1, x^*_2)$. The positioning of $x^*_1, x^*_2$ with respect to $x_{G_1}, x_{G_2}$ toghether with Assumption \ref{ass:zeros1} implies that $(\mathcal{L}-2r)v(x)\ge-f(x,h)$ for all $x \in \cI \setminus [x^*_1, x^*_2]$. Hence, $v$ solves the variational inequality \eqref{eq:ObsP}--\eqref{eq:ObsP2} and, due to Theorem \ref{thm:VI}, coincides with $V(\cdot, h)$.

\begin{remark}
Notice that the condition $x^*_1 < x^*_2$ is necessary. Otherwise, taking $v=M$ and any $x^*_1 = x^*_2 \in \cI$ solves \eqref{eq:ODEp}.
\end{remark}

\begin{remark}
An alternative approach to the one we adopted in the section above, consists in proving directly that there exists a (classical) solution $(v,x^*_1,x^*_2)$ to the system \eqref{eq:ODEp}. The arguments sketched above would then imply that this solution is unique, it is the value function of \eqref{2.5}, and $x^*_1, x^*_2$ are the stopping boundaries. Since solving \eqref{eq:ODEp} is infeasible due to the complexity of the equations arising from the explicit form of the function $M$, we took an alternative route: we used direct methods to obtain the properties of the stopping set and the smoothness of the value function, getting, as a consequence, the existence of a solution to \eqref{eq:ODEp}.
\end{remark}

Having established its existence, computing the solution of the system \eqref{eq:ODEp} is now straightforward. A general solution of the ODE in the first line of \eqref{eq:ODEp} is of the form 
\begin{equation}
\label{V:analytic}
v(x)=C_1x^{q_1}+C_2x^{q_2}-x^2(h-d^{-1}(\am/x)^{1+d})^2,
\end{equation}
with $d$ as in \eqref{eq:Panalyt} and $q_1,q_2$ as in \eqref{eq:q12}. The constants $C_1,C_2$ and the optimal boundaries $x^*_1,x^*_2$ are determined by solving a system of algebraic equations derived from the remaining four conditions in \eqref{eq:ODEp}. The existence and uniqueness of those constants follows from the earlier discussion in this subsection. We mention that we could not solve those algebraic equations analytically, so all examples presented in the paper involve numerical solution of this system of algebraic equations.

Figure \ref{p3.0} displays three possible forms of the stopping set and corresponding value functions. The stopping sets are identified by the values where the solid line (the payoff $M(\cdot)$) coincides with the dashed line (the value function $V(\cdot, h)$).

\begin{figure}[tb]
\centering
\includegraphics[width=7cm]{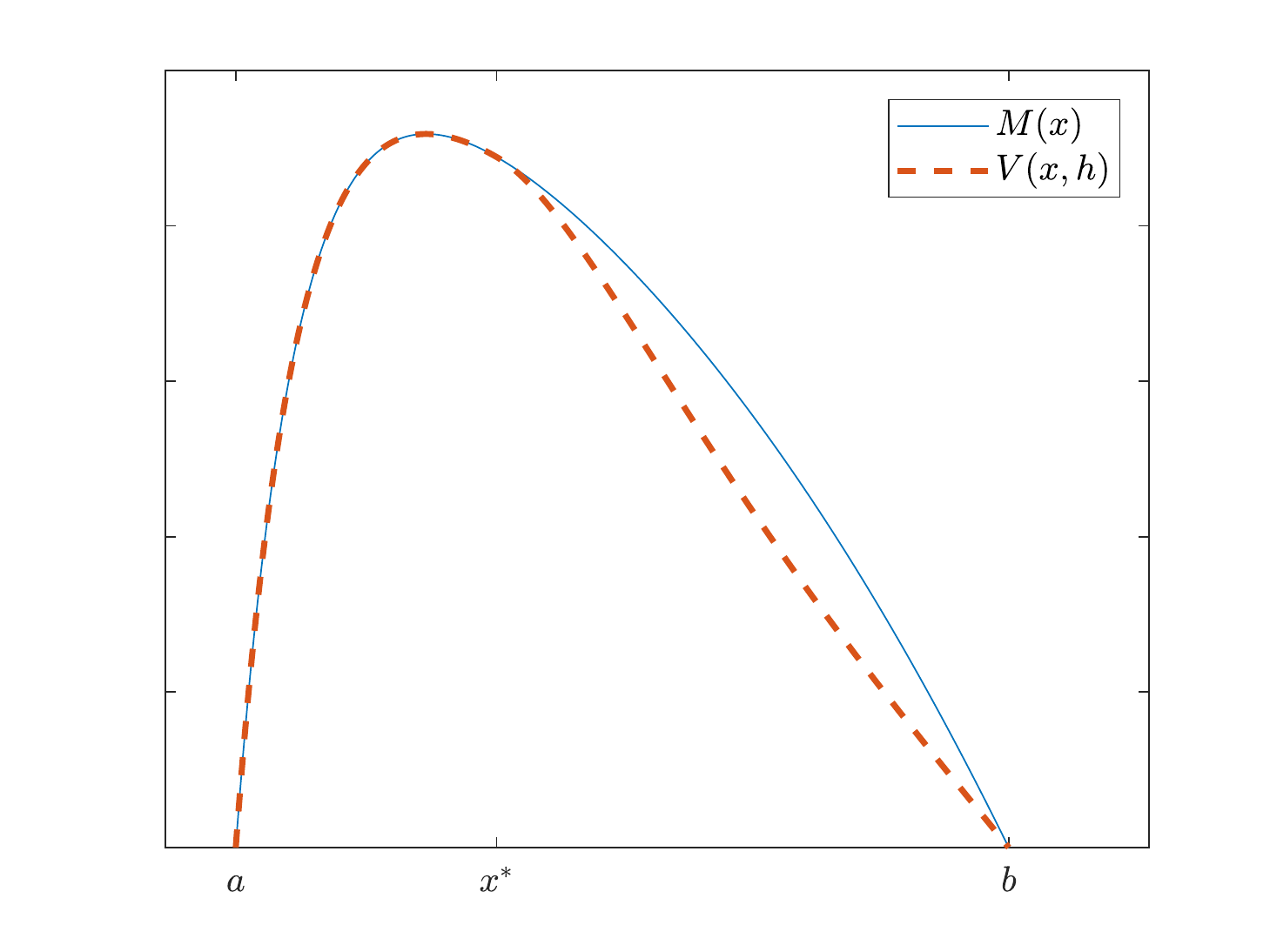}
\includegraphics[width=7cm]{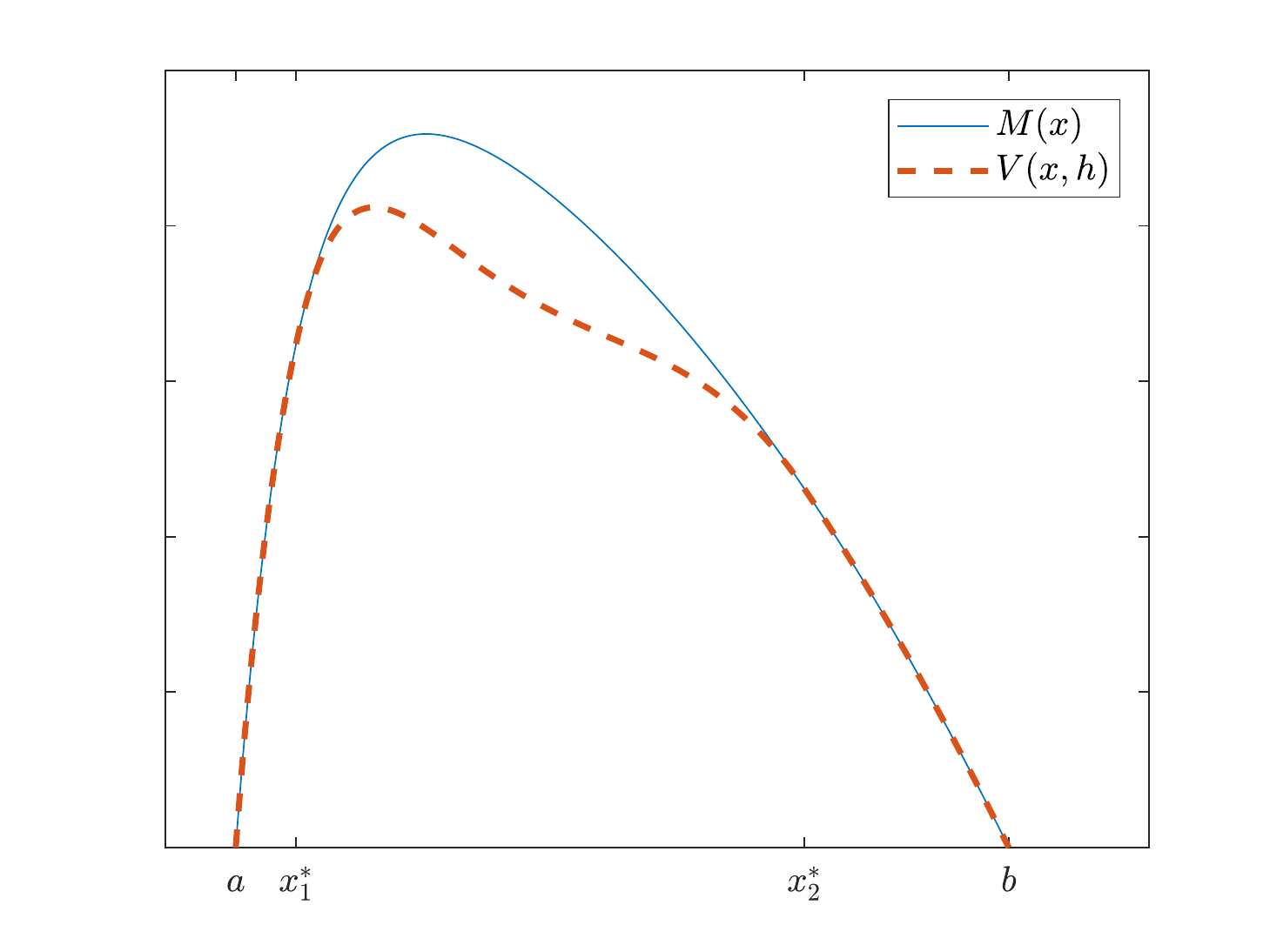}
\includegraphics[width=7cm]{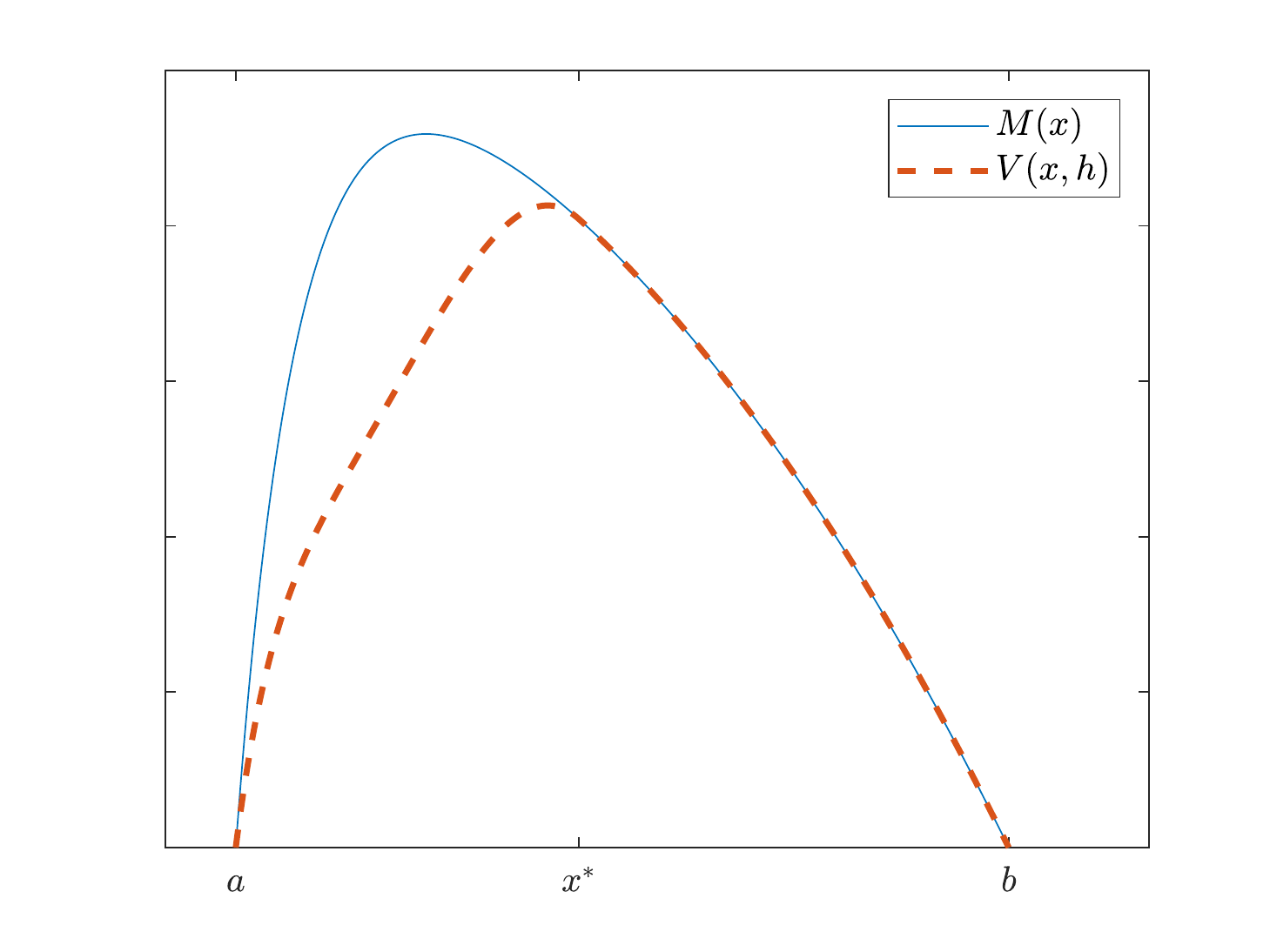}
\caption{Plots of the map $x\mapsto V(x,h)$ and $M(x)$ for different values of the initial stock holding $h$ using parameters $r=3\%$, $\sigma=30\%$, $K=100$, $b=150$, and $a=\am=K/(1+d^{-1})=40$.} 
\label{p3.0}
\end{figure}


\section{Regularity of the stopping boundaries}\label{sec:3}
So far we have studied an optimal control problem for a fixed initial stock holding $h \in \cH$. Optimal stopping boundaries $x^*_1, x^*_2$ from Proposition \ref{prop4.33} obviously depend on $h$; denote them by $x^*_{1,h}$ and $x^*_{2,h}$ with one of them possibly being equal to $a$ or $b$. We will show that $x^*_{1,h}$ and $x^*_{2,h}$ are non-decreasing and continuous in $h$. Apart from these results being of interest on their own, they will be instrumental in studying the mapping $h \mapsto V(x, h)$ and, consequently, in determining, in Section \ref{sec:4}, an optimal initial stock holding $h^*$ in problem \eqref{1.0}.

Recalling $\tau^*_{x,h}$ from \eqref{eq:tauh} we introduce functions $\gg_{h,i}$, $i=1,2$, and $\GG_h$, which are analogues of those in \eqref{eq:Gamma} and \eqref{2.4}:
\begin{equation}
\begin{aligned}
\label{eq:gg1}&\gg_{h,1}(x):=\E_x\Big[\int_{0}^{\tau^*_{h}}e^{-2ru}\sigma^2S_{u}^{2}\ud u\Big],\quad\gg_{h,2}(x):=\E_x\Big[\int_{0}^{\tau^*_{h}}e^{-2ru}P'(S_u)\sigma^2S_{u}^{2}\ud u\Big],\\
&\gg_{h,3}(x):=\E_x\left[\int_{0}^{\tau^*_{h}}e^{-2ru}(P'(S_u))^2\sigma^2S_{u}^{2}\ud u\right]\quad \text{and}\quad\GG_h(x):=\frac{\gamma_{h,2}(x)}{\gamma_{h,1}(x)}.
\end{aligned}
\end{equation}
By minimising \eqref{2.5} with respect to $h$, it is easy to see that the value $\GG_h(x)$ determines an optimal initial stock holding, provided that the next trade happens at the time $\tau^*_{x,h}$. This leads to a fixed point, so that an optimal $h^*$ in \eqref{1.0} must satisfy $\GG_{h^*}(x) = h^*$.

Applying similar arguments as in the proof of Proposition \ref{prop4.1} to $\GG_h$ with $x^*_{1,h}$ and $x^*_{2,h}$ in place of $a$ and $b$ we obtain the next result.
\begin{proposition}
\label{prop5.2.1}
For any $h\in\cH$, we have $\GG_h$ is $C^{\infty}$ and strictly increasing on $(x^*_{1,h},x^*_{2,h})$ with
\begin{align*}
\GG_h(x^*_{1,h}):&=\lim_{x\downarrow x^*_{1,h}}\GG_h(x)>P'(x^*_{1,h}),\\ 
\GG_h(x^*_{2,h}):&=\lim_{x\uparrow x^*_{2,h}}\GG_h(x)<P'(x^*_{2,h}).       
\end{align*}
\end{proposition}

We now prove a technical lemma which is fundamental for showing the monotonicity of the stopping boundaries. 
For $h\in(\Gamma(a),\Gamma(b))$, thanks to the monotonicity of $\Gamma$ (Proposition \ref{prop4.1}) we have that there exists a unique point $x_{\Gamma}(h)\in \cI$ such that 
\begin{equation}
\label{x_h}
\Gamma(x_{\Gamma}(h))=h.
\end{equation}
Moreover $x_{\Gamma}(h)\in \mathcal{C}_h$, because $G(x_{\Gamma}(h),h)<0$ by \eqref{2.6}.

\begin{lemma}
\label{lemma5.3.1}
Fix $h\in\cH$ and let Assumption \ref{ass:zeros1} hold.
\begin{itemize}
\item[(i)] If $x^*_{1,h}>a$, then $h>\Gamma(x^*_{1,h})$ and
$h+\Gamma(x^*_{1,h})\geq 2\GG_h(x^*_{1,h})$.
\item[(ii)] If $x^*_{2,h}<b$, then $h<\Gamma(x^*_{2,h})$ and $h+\Gamma(x^*_{2,h})\leq 2\GG_h(x^*_{2,h})$.
\item[(iii)] If $h>\Gamma(a)$ and $x^*_{1,h}=a$, then $h+\Gamma(x)>2\GG_h(x)$ for all $x\in (a,x_{\Gamma}(h))$.
\item[(iv)] If $h<\Gamma(b)$ and $x^*_{2,h}=b$, then $h+\Gamma(x)<2\GG_h(x)$ for all $x\in (x_{\Gamma}(h),b)$.
\end{itemize}
\end{lemma}
The statement of the lemma has an intuitive financial interpretation. In (i), if the left stopping boundary $x^*_{1,h}$ is non-trivial, then the optimal trade at $x^*_{1,h}$ is to increase the short position in the stock (recall that $\Gamma$ is negative). This is consistent with the Delta hedge $P'$ being an increasing function starting from $-1$ at $\am$ and increasing to $0$ at $\infty$. Analogously, statement (ii) says that if the right stopping boundary $x^*_{2,h}$ is non-trivial, the optimal trade at $x^*_{2,h}$ is to reduce the short position in the stock. Statements (iii)-(iv) formulate a stronger version of the previous two when the stopping boundaries are trivial.

\begin{proof}[Proof of Lemma \ref{lemma5.3.1}]
Recall that $x_p(h) \in \cI$ and $G(x_p(h),h)<0$, where $x_p(h)$ is defined in \eqref{2.7}. Hence, it must be $x^*_{1,h}<x_p(h)<x^*_{2,h}$ for $h\in \cH$. The relative placement of $x^*_{1,h}$, $x^*_{2,h}$, $x_p(h)$ and $x_\Gamma(h)$ will be central in this proof. 

{\em Proof of $(i)$:}
If $x^*_{1,h}>a$, then $G(x^*_{1,h},h)\geq0$ by \eqref{eq:simple} and
\begin{equation}\label{5.3}
\begin{aligned}
0 \leq G(x^*_{1,h},h) &<\sigma^2(x^*_{1,h})^2\left((h-P'(x^*_{1,h}))^2-(\Gamma(x^*_{1,h})-P'(x^*_{1,h}))^2\right)\\
&=\sigma^2(x^*_{1,h})^2(h-\Gamma(x^*_{1,h}))(h+\Gamma(x^*_{1,h})-2P'(x^*_{1,h})),
\end{aligned}
\end{equation}
where the strict inequality comes from \eqref{2.6} upon noting that $\big(\Gamma'(x^*_{1,h})\big)^2\gamma_1(x^*_{1,h})>0$ since $x^*_{1,h}\in\cI$.
Recalling that $x^*_{1,h}<x_p(h)$ and $P'$ is strictly increasing, we have $P'(x^*_{1,h})<h$, so 
\begin{align}\label{eq:pos}
h+\Gamma(x^*_{1,h})-2P'(x^*_{1,h})>\Gamma(x^*_{1,h})-h.
\end{align}
If $h- \Gamma(x^*_{1,h}) < 0$, combining \eqref{5.3} and \eqref{eq:pos} gives $(h-\Gamma(x^*_{1,h}))^2 < 0$, which is impossible. The equality $h - \Gamma(x^*_{1,h}) = 0$ contradicts \eqref{5.3}. Hence, $h- \Gamma(x^*_{1,h})>0$, which is the first claim in $(i)$.

For the second claim we expand the square in $f(S_u,h)$ and obtain
\begin{equation}\label{eq:fh}
\E_x\left[\int_0^{\tau^*_{h}}e^{-2ru}f(S_u,h)\ud u\right]=h^2\gg_{h,1}(x)-2h\gg_{h,2}(x)+\gg_{h,3}(x)
\end{equation}
with the notation introduced in \eqref{eq:gg1}. The explicit formulae for $\gg_{h,i}$, $i=1,2,3$, can be derived from \eqref{diffusionformula} upon replacing $a$ and $b$ by $x^*_{1,h}$ and $x^*_{2,h}$, and $\varphi$ and $\psi$ by $\varphi_{h}$ and $\psi_{h}$. The latter are, respectively, the decreasing and increasing fundamental solutions of the ODE
\begin{align*}
(\mathcal{L}-2r)u(x)&=0,\quad x\in(x^*_{1,h},x^*_{2,h}),
\end{align*} 
with the boundary conditions 
\[
\psi_{h}(x^*_{1,h}+)=0,\quad\psi'_{h}(x^*_{1,h}+)>0,\quad\varphi_{h}(x^*_{2,h}-)=0, \quad\varphi'_{h}(x^*_{2,h}-)<0.
\]
Again, these can be calculated explicitly using \eqref{eq:fund}.
Later we will also use that
\begin{align}
\label{eq:gineq1}\gg_{h,1}'(x^*_{1,h})&=\wron^{-1}_h\psi_{h}'(x^*_{1,h})\int_{x^*_{1,h}}^{x^*_{2,h}}\varphi_{h}(z)\sigma^2z^2m'(z)\ud z>0,
\end{align}
where $w_h=\hat w\big(1-(x^*_{1,h}/x^*_{2,h})^{q_1-q_2}\big)$ is the Wronskian (c.f. \eqref{eq:w}). 

From \eqref{eq:M} and \eqref{2.5} we can write the value function $V$ and the stopping payoff $M$ as
\begin{align}
\label{eq:V12}V(x,h)&=\E_x\left[\int_0^{\tau^*_{h}}e^{-2ru}f(S_u,h)\ud u\right]+\E_x\left[\int_{\tau^*_{h}}^{\tau_\cI}e^{-2ru}f(S_u,\Gamma(S_{\tau^*_{h}}))\ud u\right]\\
             &=V_1(x)+V_2(x),\nonumber\\
\label{eq:M12}M(x)&=\E_x\left[\int_0^{\tau^*_{h}}e^{-2ru}f(S_u,\Gamma(x))\ud u\right]+\E_x\left[\int_{\tau^*_{h}}^{\tau_\cI}e^{-2ru}f(S_u,\Gamma(x))\ud u\right]\\
          &=M_1(x)+M_2(x),\nonumber
\end{align}
where in $V_1,V_2, M_1, M_2$ we omit the dependence on $h\in\cH$ which is fixed. Thanks to the explicit formulae for $\psi_h$ and $\varphi_h$,  \eqref{diffusionformula} and $\Gamma \in C^\infty(\cI)$ (Proposition \ref{prop4.1}) it is not hard to verify that $M_1,V_1\in C^1([x^*_{1,h},x^*_{2,h}])$. For $V_2$, using the strong Markov property we have 
\begin{align*}
V_2(x)&=\widehat M(x^*_{1,h},\Gamma(x^*_{1,h}))\E_{x}\left[e^{-2r\tau^*_{1,h}}1_{\{\tau^*_{1,h}<\tau^*_{2,h}\}}\right]
+\widehat M(x^*_{2,h},\Gamma(x^*_{2,h}))\E_{x}\left[e^{-2r\tau^*_{2,h}}1_{\{\tau^*_{1,h}>\tau^*_{2,h}\}}\right],
\end{align*}
where $\tau^*_{1,h}$, $\tau^*_{2,h}$ denote the first entry time to $[a,x^*_{1,h}]$ and $[x^*_{2,h},b]$, respectively. It is well-known that the two expected values on the right-hand side of the equation above can be expressed in terms of $\psi_h$ and $\varphi_h$ (\cite[Chapter II, Par.~10]{borodin2012handbook}), hence proving $V_2 \in C^1([x^*_{1,h},x^*_{2,h}])$. An analogous argument applies for $M_2$. 

Since $V=M$ at $x^*_{1,h}$ and the smooth-fit holds we have
\begin{align}
\label{eq:CF}V_1(x^*_{1,h})+V_2(x^*_{1,h})&=M_1(x^*_{1,h})+M_2(x^*_{1,h}),\\
\label{eq:SF}V_1'(x^*_{1,h})+V_2'(x^*_{1,h})&=M_1'(x^*_{1,h})+M_2'(x^*_{1,h}).
\end{align}
Noticing that $\P_{x^*_{1,h}}(\tau^*_h=0)=1$ we have $V_1(x^*_{1,h})= M_1(x^*_{1,h}) = 0$ and hence,
\begin{align}
V_2(x^*_{1,h})= M_2(x^*_{1,h}).
\end{align} 
Using the optimality of $\Gamma(x)$ for $\widehat M(x, \cdot)$ and the strong Markov property, for $x\in(x^*_{1,h},x^*_{2,h})$ we have
\begin{align*}
V_2(x)=\E_x\left[e^{-2r\tau^*_{h}}M(S_{\tau^*_{h}})\right]\leq\E_x\left[e^{-2r\tau^*_{h}}\widehat{M}(S_{\tau^*_{h}},\Gamma(x))\right]=M_2(x).
\end{align*}
Hence, $V_2'(x^*_{1,h})\leq M_2'(x^*_{1,h})$. Inserting the latter into \eqref{eq:SF} we deduce
\begin{equation}
\label{v1'>m1'}
V_1'(x^*_{1,h})\geq M_1'(x^*_{1,h}).
\end{equation}

Our task is now to rewrite both sides of \eqref{v1'>m1'} using \eqref{eq:fh} and \eqref{Mprobability}. For an arbitrary $x\in[x^*_{1,h},x^*_{2,h}]$ we have 
\begin{align*} 
V_1'(x)&=h^2\gg_{h,1}'(x)-2h\gg_{h,2}'(x)+\gg_{h,3}'(x)\\
M_1'(x)&=\Gamma^2(x)\gg_{h,1}'(x)-2\Gamma(x)\gg_{h,2}'(x)+\gg_{h,3}'(x)
                    +2\Gamma(x)\Gamma'(x)\gg_{h,1}(x)-2\Gamma'(x)\gg_{h,2}(x).
\end{align*}
Inserting the above in \eqref{v1'>m1'} we obtain 
\begin{equation}
\label{5.8}
 (h^2-\Gamma^2(x^*_{1,h}))\gg_{h,1}'(x^*_{1,h})-2(h-\Gamma(x^*_{1,h}))\gg_{h,2}'(x^*_{1,h})\geq0.
\end{equation}
Since $\gg_{h,1}'(x^*_{1,h})>0$ by \eqref{eq:gineq1} and we have shown above that $h-\Gamma(x^*_{1,h})>0$, we can
divide both sides of \eqref{5.8} by $(h-\Gamma(x^*_{1,h}))\gg_{h,1}'(x^*_{1,h})$, thus obtaining
\begin{equation*}
h+\Gamma(x^*_{1,h})\geq2\frac{\gamma_{h,2}'(x^*_{1,h})}{\gamma_{h,1}'(x^*_{1,h})}=\lim_{x\downarrow x^*_{1,h}}\GG_h(x) =: \GG_h(x^*_{1,h}),
\end{equation*}
where the first equality follows from d'Hospital's rule (see \eqref{eq:gg1}). This concludes the proof of $(i)$.
\vspace{+4pt}

{\em Proof of (ii):} This is analogous to that of (i), hence we omit further details.
\vspace{+4pt}

{\em Proof of (iii) and (iv):} We give a full argument only for (iv) as the case of (iii) can be treated analogously. Fix $h \in \cH$ such that $h<\Gamma(b)$ and $x^*_{2,h}=b$. 

First we notice that for all $x\in\cI\setminus\{x_{\Gamma}(h)\}$ we have
\begin{equation}
\label{intsplit}
M(x)=\E_x\left[\int_0^{\tau_\cI}e^{-2ru}f(S_u,\Gamma(x))\ud u\right]<\E_x\left[\int_0^{\tau_\cI}e^{-2ru}f(S_u,h)\ud u\right],
\end{equation}
where the strict inequality is due to the fact that for each $x\in\cI$, the mapping $\zeta \mapsto \widehat M(x, \zeta)$ is strictly convex and attains its minimum at $\zeta=\Gamma(x)$.

Now fix an arbitrary point $\hat x\in (x_\Gamma(h),b)$. With the notation introduced in \eqref{eq:M12} we rewrite \eqref{intsplit} as 
\begin{equation}
\label{comp1}
\begin{aligned}
&M_1(\hat{x})+M_2(\hat{x})
<\E_{\hat{x}}\left[\int_0^{\tau^*_{h}}e^{-2ru}f(S_u,h)\ud u\right]+\E_{\hat{x}}\left[\int_{\tau^*_{h}}^{\tau_\cI}e^{-2ru}f(S_u,h)\ud u\right]
=:V_1(\hat{x})+ \widetilde{V}_2(\hat{x}).
\end{aligned}
\end{equation}
Here we are again omitting the dependence of $V_1$ and $\widetilde V_2$ on $h$ and note that $V_1$ is the same as in \eqref{eq:V12}, whereas $\widetilde V_2$ is not.
Since $x^*_{2,h}=b$, we have $\{S_{\tau^*_{h}}=b\}=\{\tau^*_{h}=\tau_\cI\}$, so $\{S_{\tau^*_{h}}=x^*_{1,h}\}=\{\tau^*_{h}<\tau_\cI\}$.
Using this fact and the strong Markov property we obtain
\begin{align}\label{eq:V2M2}
\widetilde{V}_2(\hat{x})-M_2(\hat{x})
&=\E_{\hat{x}}\left[e^{-2r\tau^*_{h}}\E_{S_{\tau^*_{h}}}\left[\int_0^{\tau_\cI}e^{-2ru}\big(f(S_u,h)-f(S_u,\Gamma(\hat x))\big)\ud u\right]\right]\notag\\
&=\E_{\hat{x}}\left[e^{-2r\tau^*_{h}}1_{\{\tau^*_{h}<\tau_\cI\}}\right]\E_{x^*_{1,h}}\left[\int_0^{\tau_\cI}e^{-2ru}\big(f(S_u,h)-f(S_u,\Gamma(\hat x))\big)\ud u\right]\notag\\
&=\E_{\hat{x}}\left[e^{-2r\tau^*_{h}}1_{\{\tau^*_{h}<\tau_\cI\}}\right] \big(\widehat M(x^*_{1,h}, h) - \widehat M(x^*_{1,h}, \Gamma(\hat x)) \big)<0,
\end{align}
where it remains to justify the final inequality. Since $\hat{x}>x_{\Gamma}(h)>x^*_{1,h}$, by the monotonicity of $\Gamma$ (Proposition \ref{prop4.1}), we have 
\begin{align}\label{eq:h0G}
\Gamma(\hat{x})>h>\Gamma(x^*_{1,h}). 
\end{align}
Since the mapping $\zeta \mapsto \widehat M(x^*_{1,h}, \zeta)$ is strictly convex and attains its minimum at $\Gamma(x^*_{1,h})$, it is strictly increasing for $\zeta>\Gamma(x^*_{1,h})$. Hence the inequality in \eqref{eq:V2M2} holds and $\widetilde{V}_2(\hat{x})<M_2(\hat{x})$ upon noticing that $\P_{\hat x}(\tau^*_{h}<\tau_\cI)>0$.

Combining \eqref{comp1} with \eqref{eq:V2M2} implies $V_1(\hat{x})>M_1(\hat{x})$. Rewriting this inequality in terms of the functions $\gg_{h,i}$, $i=1,2,3$, given in \eqref{eq:gg1}, we obtain
\begin{equation}
\label{5.9}
 (h^2-\Gamma^2(\hat{x}))\gg_{h,1}(\hat{x})-2(h-\Gamma(\hat{x}))\gg_{h,2}(\hat{x})>0.
\end{equation}
It is clear from \eqref{eq:gg1} that $\gg_{h,1}(\hat x)>0$ since $\hat x\in(x^*_{1,h},x^*_{2,h})$. Then, using also \eqref{eq:h0G} we can divide both sides of \eqref{5.9} by $(h-\Gamma(\hat{x}))\gg_{h,1}(\hat{x})<0$ to obtain
\begin{equation*}
h+\Gamma(\hat{x})<2\frac{\gg_{h,2}(\hat{x})}{\gg_{h,1}(\hat{x})}=2\GG_{h}(\hat{x}).
\end{equation*}
\end{proof}

With Lemma \ref{lemma5.3.1} in place we can now show that the optimal stopping boundaries $x^*_{1,h}$, $x^*_{2,h}$ are non-decreasing in $h$.  
\begin{theorem}
\label{prop5.4.1}
Let Assumption \ref{ass:zeros1} hold. Then, the mappings $h\mapsto x^*_{1,h}$ and $h\mapsto x^*_{2,h}$ are non-decreasing on $\cH$.
\end{theorem}     
\begin{proof}
We only show that $h\mapsto x^*_{2,h}$ is non-decreasing as the arguments for the monotonicity of $h \mapsto x^*_{1,h}$ are analogous. For the clarity of notation let us set $x^*_i(h)=x^*_{i,h}$ for $i=1,2$.

Fix $h<\tilde h$ in $\cH$. If $\tilde h\ge \Gamma(b)$ then $x^*_{2}(\tilde h)=b$ by Propositions \ref{prop:Gzeros2} and \ref{prop4.33}, so trivially $x^*_2(h) \le x^*_2(\tilde h)$. Assume now that $\tilde h<\Gamma(b)$. We split the proof into two cases.
\vspace{+4pt}

({\em Case 1}). Let us first consider $x^*_2(h)=b$ (this can occur under (A.2); see Proposition \ref{prop4.33}). Arguing by contradiction we assume $x^*_2(\tilde h)<b$. Then, we have
\begin{equation}\label{eq:VM}
V(x,\tilde{h})=M(x)>V(x,h),\quad \text{for all $x\in (x^*_2(\tilde h)\vee x^*_1(h), b)$}.
\end{equation}
Taking $\tau^*_h$ optimal for $V(x,h)$ and noticing that it is also admissible for $V(x,\tilde h)$, it is easy to check that \eqref{eq:VM} implies 
\begin{equation*}
\E_x\left[\int_0^{\tau^*_{h}}e^{-2ru}f(S_u,\tilde{h})\ud u\right]>\E_x\left[\int_0^{\tau^*_{h}}e^{-2ru}f(S_u,h)\ud u\right]\quad \text{for all $x\in (x^*_2(\tilde h)\vee x^*_1(h), b)$}.
\end{equation*}

Both expected values above can be written using the functions $\gg_{h,i}$, $i=1,2,3$, introduced in \eqref{eq:gg1} (see also \eqref{eq:fh}). This gives
\begin{equation*}
(\tilde{h}^2-h^2)\gg_{h,1}(x)-2 (\tilde{h}-h)\gg_{h,2}(x)>0.
\end{equation*}
Dividing both sides by $(\tilde{h}-h)\gg_{h,1}(x)>0$ we obtain
 \begin{equation}\label{eq:C1}
h+\tilde{h}>2\frac{\gg_{h,2}(x)}{\gg_{h,1}(x)}=2\GG_h(x).
\end{equation}
Since $x_{\Gamma}(\tilde{h})\in \mathcal{C}_{\tilde{h}}$ by \eqref{x_h} and $\Gamma$ is strictly increasing, we have $x_\Gamma(\tilde h)<x^*_{2}(\tilde h)$ and $\tilde h <\Gamma(x)$ for $x\in (x^*_2(\tilde{h})\vee x^*_1(h), b)$. Hence,
\begin{equation*}
\Gamma(x)+h>\tilde{h} + h>2\GG_h(x),
\end{equation*}
which contradicts (iv) in Lemma \ref{lemma5.3.1}.
\vspace{+4pt}

({\em Case 2}). Let us now consider $x^*_2(h)<b$. In this case we have $\Gamma(x^*_2(h))<\Gamma(b)$, which gives rise to two sub-cases.

({\em Case 2a}). If $h<\Gamma(x^*_2(h))\leq\tilde{h}<\Gamma(b)$, by monotonicity of $\Gamma$ we obtain 
$x^*_2(h)\leq x_{\Gamma}(\tilde{h})$. Moreover, using that $x_{\Gamma}(\tilde{h})\in \mathcal{C}_{\tilde{h}}$, it must be $x_{\Gamma}(\tilde{h})<x^*_2(\tilde{h})$. Hence the claim.

({\em Case 2b}). If $h<\tilde{h}<\Gamma(x^*_2(h))<\Gamma(b)$, we adapt arguments from Case 1 above. Assume, by contradiction, that $x^*_2(\tilde{h}) < x^*_2(h)$. Then, as in \eqref{eq:C1}, we have $h+\tilde{h}>2\GG_h(x)$ for all $x\in (x^*_2(\tilde{h})\vee x^*_1(h), x^*_2(h))$.
By assumption $\tilde{h}<\Gamma(x^*_2(h))$, hence
\begin{equation*}
h+\Gamma(x^*_2(h))>h+\tilde{h}\geq 2\lim_{x\uparrow x^*_2(h)}\GG_h(x),
\end{equation*}
which contradicts (ii) in Lemma \ref{lemma5.3.1}.
\end{proof}

Theorem \ref{prop5.4.1} allows us to prove the continuity of the optimal boundaries and the continuity of the optimal stopping time with respect to $x$ and $h$ (jointly). This is needed to prove that $\partial_hV$ exists and it is (jointly) continuous, which will then allow to establish first order conditions for a minimiser in \eqref{eq:vscript}.

\begin{theorem}
\label{prop5.5.1}
Let Assumption \ref{ass:zeros1} hold. Then the mappings $h\mapsto x^*_{1,h}$ and $h\mapsto x^*_{2,h}$ are continuous on $\cH$. Moreover, $(x,h)\mapsto \tau^*_{x,h}$ is continuous on $\cI\times\cH$, $\P$-a.s.
\end{theorem}
\begin{proof}
First we show continuity of the optimal boundaries and then continuity of the stopping times. For the clarity of notation let us set $x^*_i(h)=x^*_{i,h}$ for $i=1,2$.
\vspace{+4pt}

({\em Continuity of the boundaries}). We only give full arguments for the upper boundary $x^*_2$ as the case of the lower boundary $x^*_1$ can be handled analogously. First we show that $x^*_2$ is left-continuous using a standard argument (see, e.g., \cite[Chapter \rom{7}]{peskir2006optimal}). Fix $h\in \cH$ and consider an increasing sequence $(h_n)_{n\ge 1} \subset \cH$ such that $h_n\uparrow h$ as $n\rightarrow \infty$. For each $n\geq1$, we have $(x^*_2(h_n),h_n)\in \cD$ and in the limit
\begin{equation*}
\lim_{n\rightarrow \infty}(x^*_2(h_n),h_n)=(x^*_2(h-), h),
\end{equation*}
where the left limit $x^*_2(h-)$ is well-defined by the monotonicity of $x^*_2$. Since $\cD$ is closed it must be $(x^*_2(h-), h)\in \cD$ and then $x^*_2(h-)\geq x^*_2(h)$. However, since $x^*_2(\cdot)$ is increasing we also have $x^*_2(h-)\leq x^*_2(h)$, so that left-continuity follows.

The proof of right-continuity of $x^*_2$ follows ideas contained in \cite{de2015note}. If $x^*_2(h)=b$ the claim is trivial. Consider the case $x^*_2(h)<b$. Arguing by contradiction let us assume that $x^*_2(h+)>x^*_2(h)$. Then we can find $x_d$ and $x_u$, such that $x^*_2(h)<x_d<x_u<x^*_2(h+)$, and a sufficiently small $\eps>0$ such that $(x_d,x_u)\times(h,h+\eps]\subset\cC$. Recalling \eqref{eq:FBPV}, we have
\begin{equation}\label{eq:ODE}
(\mathcal{L}-2r)V(x,h+\eps)=-f(x,h+\eps),\quad \text{for $x\in(x_d,x_u)$}.
\end{equation}
Take any $\Psi\in C_{c}^{\infty}((x_d,x_u))$ with $\Psi\ge 0$. Multiplying \eqref{eq:ODE} by $\Psi$, integrating over $[x_d,x_u]$ and using integration by parts we obtain
\begin{equation}\label{eq:ODE2}
\int_{x_d}^{x_u}V(z,h+\eps)(\mathcal{L}^*\Psi-2r\Psi)(z)\ud z=-\int_{x_d}^{x_u}f(z,h+\eps)\Psi(z)\ud z,
\end{equation}
where $\cL^*$ is the adjoint of $\cL$:
\begin{equation*}
(\mathcal{L}^*-2r)\Psi(x)=\frac{1}{2}\frac{\partial^2}{\partial x^2}(\Psi(x)\sigma^2x^2)-\frac{\partial}{\partial x}(\Psi(x)rx)-2r\Psi(x).
\end{equation*}

By the continuity of $V$, we have $\lim_{\eps \to 0} V(z,h+\eps)= V(z, h) = M(z)$ for all $z\in(x^*_{2}(h),x^*_2(h+))$. Then, using the dominated convergence theorem in \eqref{eq:ODE2} to pass to the limit as $\eps\to 0$ we get
\begin{equation*}
-\int_{x_d}^{x_u}f(z,h)\Psi(z)\ud z=\int_{x_d}^{x_u}M(z)(\mathcal{L}^*-2r)\Psi(z)\ud z
=\int_{x_d}^{x_u}\Psi(z)(\mathcal{L}-2r)M(z)\ud z.
\end{equation*}
This is equivalent to $\int_{x_d}^{x_u}\Psi(z) G(z, h) = 0$. However, $[x_d,x_u]$ is in the stopping region $\cD_h$, so $G(z, h) \ge 0$. Recalling that $\Psi$ is arbitrary and non-negative, we conclude that $G(z, h) = 0$ for almost all $z \in [x_d, x_u]$, which contradicts Assumption \ref{ass:zeros1}.

\vspace{+4pt}

({\em Continuity of optimal stopping times}). This part of the proof is based on ideas from \cite{de2018global} (see also, e.g., \cite{menaldi1980optimal}). 
Let
\[
\hat \tau^*_{x,h}:=\inf\{t\ge 0: S^x_t\notin [x^*_1(h),x^*_2(h)]\}
\]
and $\Omega^0 = \{ \tau^*_{x,h} = \hat \tau^*_{x,h} \}$. By \eqref{eq:tauOa}, we have $\P(\Omega^0) = 1$.

Fix $(x,h)\in \cI\times\cH$ and let $(x_n,h_n)_{n\ge 1}$ be a sequence converging to $(x,h)$ as $n\rightarrow \infty$. For any $\omega \in \Omega^0$, if $\tau^*_{x,h}(\omega)=0$, then lower semi-continuity holds trivially. If $\tau^*_{x,h}(\omega)>0$, then for any $t>0$ such that $\tau^*_{x,h}(\omega)>t$, there exists $\eps>0$ (depending on $(t,x,h,\omega)$) such that 
 \begin{align}\label{eq:d1}
 \inf_{0\leq u\leq t}d((S^x_u(\omega),h), \partial \cC)\ge \eps>0,
 \end{align}
 where we use the standard Euclidean distance 
 \[
 d(y, \partial \cC):=\inf_{\hat{y}\in \partial \cC}d(y,\hat{y}), \quad\text{for $y\in \closure{\cI}\times \cH$}.
 \]
 By uniform continuity of $(t,x)\mapsto S_t^x(\omega)$ on compact sets, for $n$ sufficiently large we have 
 \begin{align}\label{eq:d2}
 \inf_{0\le u\le t}d((S_u^{x_n}(\omega),h_n), (S^x_u(\omega),h))\le \eps/2.
 \end{align}
 Combining \eqref{eq:d1} and \eqref{eq:d2} we obtain 
 \begin{equation*}
 \inf_{0\leq u\leq t}d((S_u^{x_n}(\omega),h_n), \partial \cC)>\eps/2,
 \end{equation*}
 for all sufficiently large $n$. Hence $\tau^*_{x_n,h_n}(\omega) > t$ for all such $n$. Since $t>0$ was arbitrary we have 
 \[
 \liminf_{n\to \infty}\tau^*_{x_n,h_n}(\omega)\ge\tau^*_{x,h}(\omega). 
 \]

To prove the upper semi-continuity we use $\hat\tau^*_{x,h}$ which is identical to $\tau^*_{x,h}$ on $\Omega^0$. Recall that $\P(\hat\tau^*_{x,h}<\infty)=1$ as $\hat\tau^*_{x,h}$ is the exit time of a geometric Brownian motion from a bounded interval. For any $\omega \in \Omega^0$, there is $t>\hat \tau^*_{x,h}(\omega)$ and arbitrarily close to $\hat \tau^*_{x,h}(\omega)$ such that $S^x_t(\omega) \notin [x^*_1(h), x^*_2(h)]$. By the continuity of $x^*_1(\cdot)$, $x^*_2(\cdot)$ and $x \mapsto S^x_t(\omega)$, we have $S^{x_n}_t(\omega) \notin [x^*_1(h_n), x^*_2(h_n)]$ and $\hat\tau^*_{x_n,h_n} (\omega) < t$ for sufficiently large $n$. Hence $\limsup_{n \to \infty} \hat\tau^*_{x_n,h_n} (\omega) \le \hat\tau^*_{x,h}(\omega)$. Combined with the lower semi-continuity proved above, this implies the a.s. continuity of $(x,h) \mapsto \tau^*_{x,h}$.
\end{proof}
\begin{figure}[tb]
\centering
\includegraphics[width=9cm]{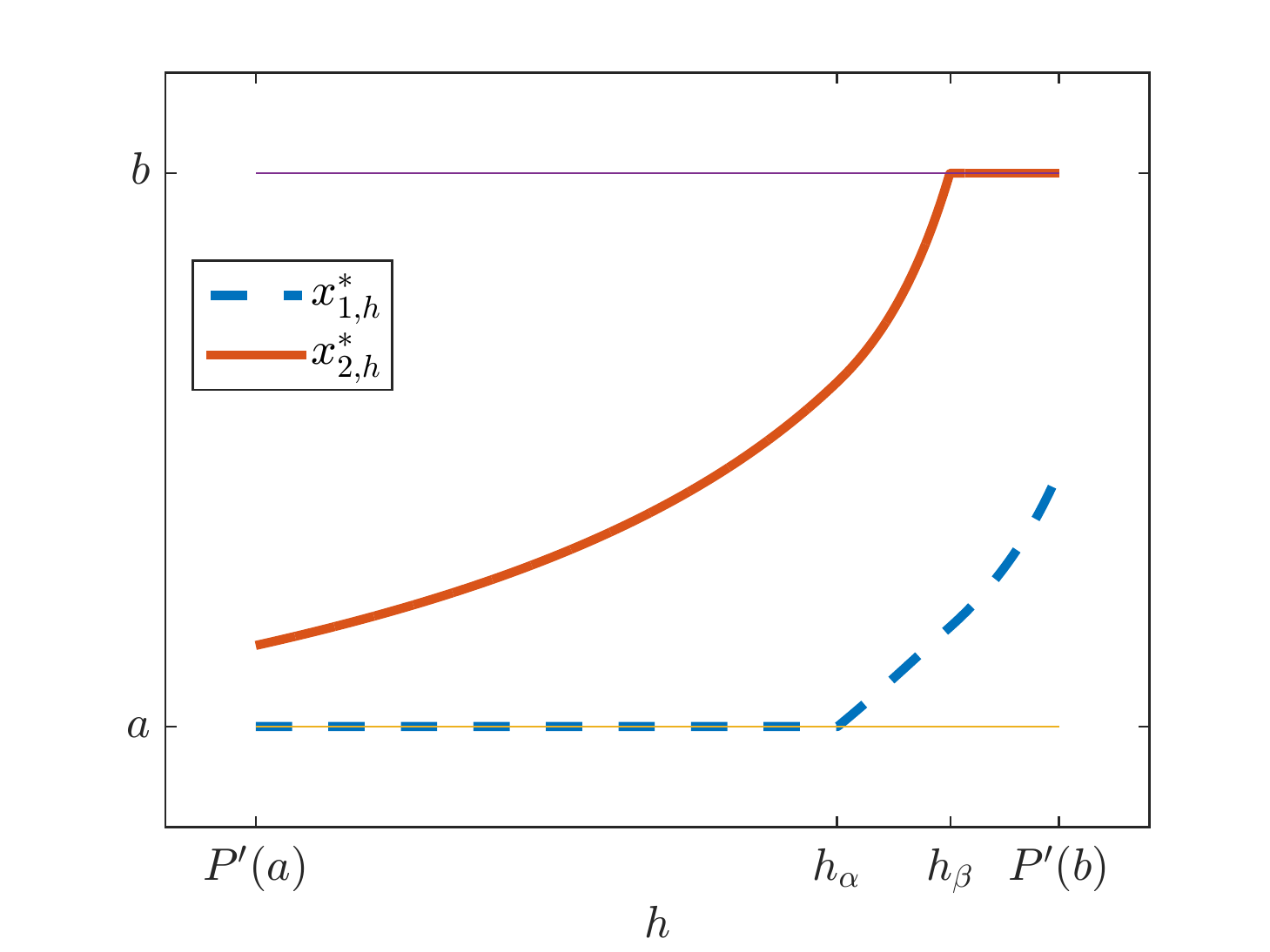}
\caption{Plots of the optimal stopping boundaries $x^*_{1,h},x^*_{2,h}$ as functions of $h$ using parameters $r=3\%$, $\sigma=30\%$, $K=100$, $b=150$ and $a=\hat{a}=K/(1+d^{-1})=40$.} 
\label{p4.0}
\end{figure} 
Figure \ref{p4.0} illustrates the optimal stopping boundaries $x^*_{1,h}$ and $x^*_{2,h}$ when $h\in\cH$ is varying. We highlight points $h_\alpha$ and $h_\beta$ where the continuation region changes from $(a, x^*_{2,h})$ to $(x^*_{1,h},x^*_{2,h})$ and from $(x^*_{1,h},x^*_{2,h})$ to $(x^*_{1,h},b)$, respectively. The three regimes (i)--(iii) of Proposition \ref{prop4.33} are clearly visible on the graph. 


\section{Optimal initial stock holding}\label{sec:4}

The existence of an optimal initial stock holding in \eqref{eq:vscript} follows from compactness of $\cH$ and continuity of $h \mapsto V(x,h)$. Here we show that the minimum of $V(x,\,\cdot\,)$ is attained in the interior of $\cH$. Moreover, although an optimal $h^*$ cannot be obtained explicitly, we show that it must solve a simple algebraic equation whose numerical solution is straightforward.

\begin{proposition}
\label{prop5.7.1}
Under Assumption \ref{ass:zeros1}, we have $V(x,\,\cdot\,)\in C^1(\cH)$ for all $x\in\cI$. Moreover, we have
\begin{equation}
\label{derivative v}
\partial_hV(x,h)=\E_x\left[\int_0^{\tau^*_{h}}e^{-2ru}2(h-P'(S_u))\sigma^2 S^2_u\ud u\right],
\end{equation}
and $\partial_hV\in C(\cI\times\cH)$.
\end{proposition}
\begin{proof}
The argument of proof is analogous to the one used to prove Theorem \ref{thm:C1}, so we only provide a sketch.
Let $\eps>0$ and denote by $\tau^*_{x, h}$ an optimal stopping time for $V(x,h)$. Since $\tau^*_{x,h}$ is admissible but sub-optimal for $V(x,h+\eps)$, an application of the mean value theorem yields
\begin{align*}
&V(x,h+\eps)-V(x,h)\le\eps\E_{x}\left[\int_0^{\tau^*_{h}}e^{-2ru}2(h_\eps-P'(S_u))\sigma^2 S^2_u\ud u\right],
\end{align*}
where $h_\eps\in[h,h+\eps]$. Dividing both sides of the inequality by $\eps$ and letting $\eps\to 0$, we obtain 
\begin{equation}
\label{5.12}
\limsup_{\eps\rightarrow0}\frac{V(x,h+\eps)-V(x,h)}{\eps}\leq \E_{x}\left[\int_0^{\tau^*_{h}}e^{-2ru}2(h-P'(S_u))\sigma^2 S^2_u\ud u\right].
\end{equation}
For the lower bound we denote by $\tau^*_{x, h+\eps}$ the optimal stopping time for $V(x,h+\eps)$ and arguing as above we get
\begin{align*}
&V(x,h+\eps)-V(x,h)\ge \eps\E_{x}\left[\int_0^{\tau^*_{h+\eps}}e^{-2ru}2(h_\eps-P'(S_u))\sigma^2 S^2_u\ud u\right].
\end{align*}
Dividing by $\eps$ both sides of the inequality, letting $\eps\to 0$ and recalling the continuity of the map $h \mapsto \tau^*_{x,h}$ (Theorem \ref{prop5.5.1}) we obtain
\begin{equation}
\label{5.13}
\liminf_{\eps\rightarrow0}\frac{V(x,h+\eps)-V(x,h)}{\eps}\geq \E_{x}\left[\int_0^{\tau^*_{h}}e^{-2ru}2(h-P'(S_u))\sigma^2 S^2_u\ud u\right].
\end{equation}
Combining \eqref{5.13} and \eqref{5.12} gives
\begin{equation*}
\partial_h^{+} V(x,h)=\E_{x}\left[\int_0^{\tau^*_{h}}e^{-2ru}2(h-P'(S_u))\sigma^2 S^2_u\ud u\right],
\end{equation*}
where $\partial^+_h$ denotes the right partial derivative. The same arguments can be applied to obtain the same expression as above also for the left partial derivative $\partial_h^{-} V$, hence \eqref{derivative v} holds.

Continuity of the map $(x,h)\mapsto \partial_hV(x,h)$ is easily deduced from $\P$-a.s. continuity of the maps 
\[
(x,h)\mapsto(h-P'(S^x_u))\sigma^2 (S^x_u)^2\quad\text{and}\quad (x,h)\mapsto\tau^*_{x,h},
\]
and the dominated convergence theorem.
\end{proof}

Finally, we give our result regarding an optimal initial stock holding $h^*$.
\begin{theorem}
\label{optimalh}
Under Assumption \ref{ass:zeros1}, for each initial stock price $S_0=x\in \cI$,
\begin{equation*}
\argmin_{h\in\cH}V(x,h) \subseteq (P'(a),P'(b))=\textrm{int}(\cH).
\end{equation*}
Moveover, each minimiser $h^* \in \argmin_{h\in\cH}V(x,h)$ is a solution of the following equation   
\begin{equation}\label{hstar}   
h^*=\GG_{h^*}(x),
\end{equation}
where $\GG_h$ was defined in \eqref{eq:gg1}.
\end{theorem}
\begin{proof}
Fix $x \in \cI$ and let $\cC^{x}\!:=\!\{h\!\in\!\cH: V(x,h)\!<\!M(x)\}$. We have $\cC^x\!\neq\!\varnothing$ due to Proposition \ref{prop4.202}. 
Hence $\argmin_{h\in\cH}V(x,h) \subset \cC^x$ given that $V\le M$ and $M$ is independent of $h$.

Although it is possible that $P'(a)$ or $P'(b)$ are in $\cC^x$, we will show that the minimum of $V(x,\cdot)$ cannot be attained there. For that purpose, notice that 
\[
\partial_h V(x,h) = 2 \hat \gamma_{h,1}(x) \big(h  - \GG_h(x)\big)
\]
thanks to \eqref{derivative v} and with the notation of \eqref{eq:gg1}. If $P'(a) \in \cC^x$, 
the inequality $\GG_{h}(P'(a)+)>P'(a)$ (see Proposition \ref{prop5.2.1}) implies that $\partial_h V(x,P'(a)+)<0$. Hence the minimum of $V(x, \cdot)$ is not attained at $P'(a)$. Similarly, if $P'(b) \in \cC^x$, then $\partial_h V(x,P'(b)-)>0$, so the minimum of $V(x, \cdot)$ cannot be attained at $P'(b)$.

Consequently, each minimiser $h^*$ of $V(x, \cdot)$ is in $(P'(a),P'(b))$ and must satisfy $\partial_h V(x,h^*) = 0$, which is equivalent to \eqref{hstar}.
\end{proof}

\begin{figure}[tb]
\centering
\includegraphics[width=8cm]{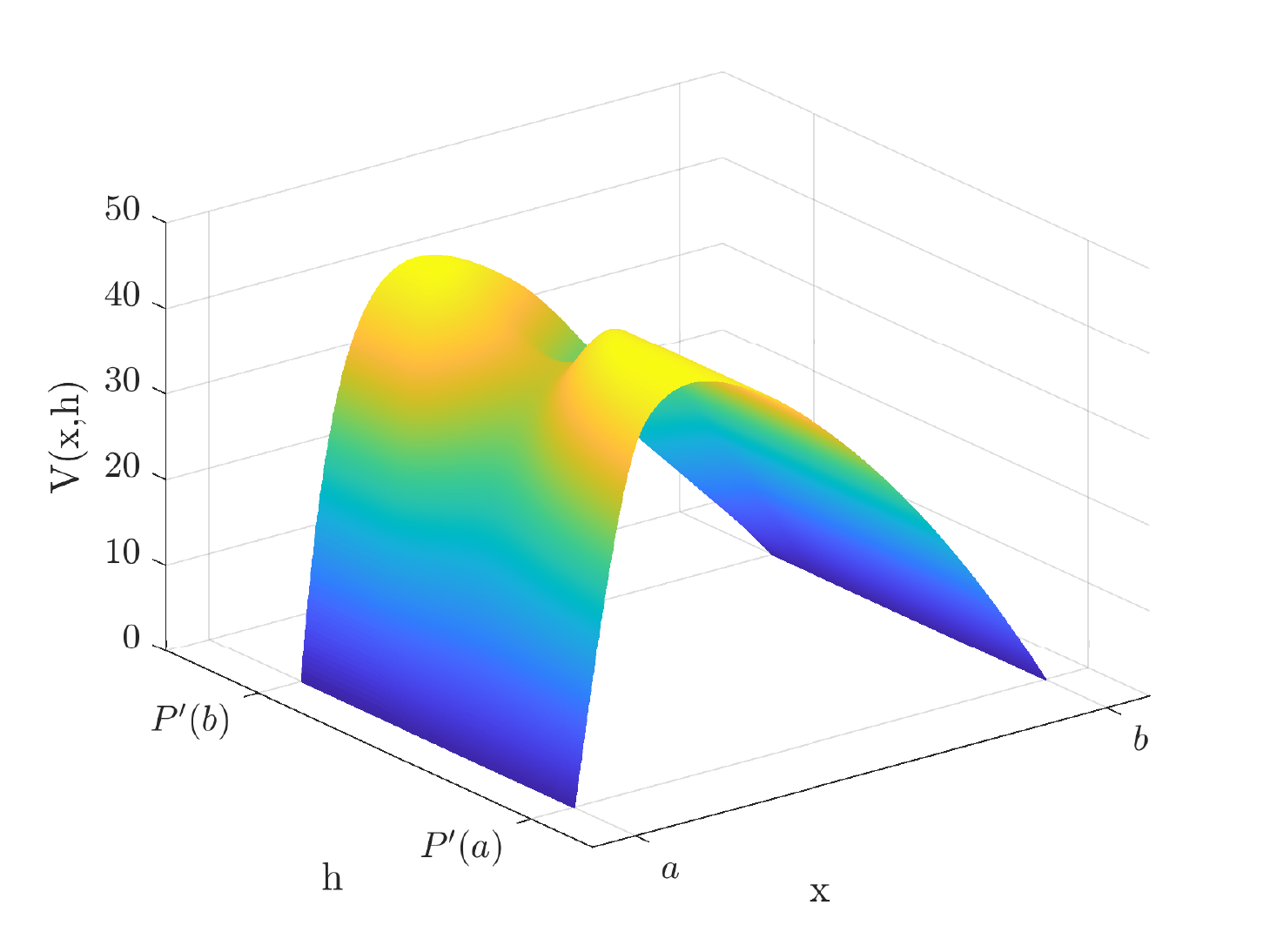}
\includegraphics[width=8cm]{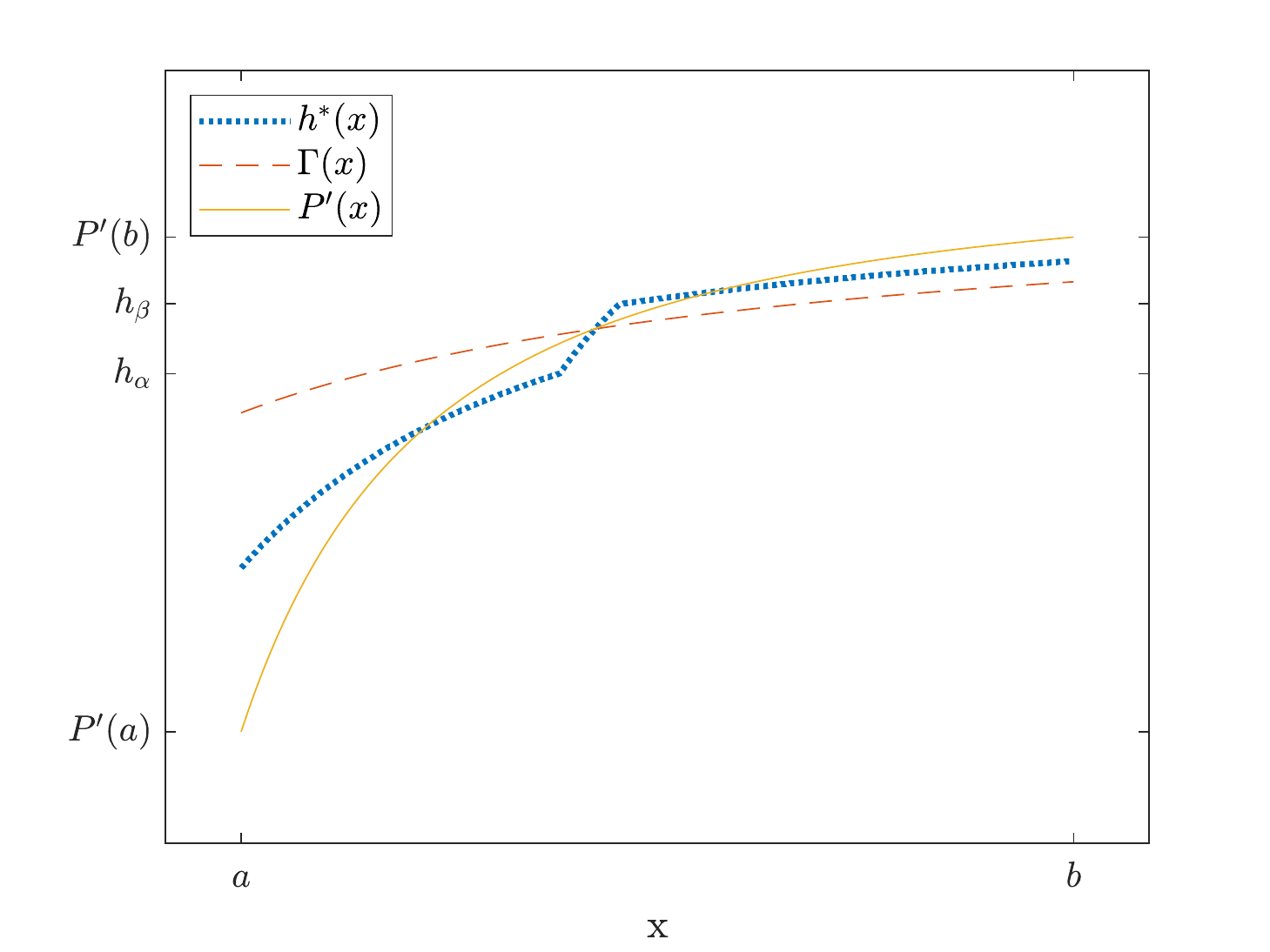}
\caption{Left panel: 3-D plot of the value function $(x,h)\mapsto V(x,h)$. Right panel: plot of optimal stock holdings $x\mapsto h^*(x)$, $x\mapsto\Gamma(x)$ and the Black-Scholes Delta $x\mapsto P'(x)$ using parameters $r=3\%$, $\sigma=30\%$, $K=100$, $b=150$ and $a=\hat{a}=K/(1+d^{-1})=40$. } 
\label{p6.0}
\end{figure} 
We used the first order condition \eqref{hstar} to numerically compute the optimal initial stock holding and it turned out that \eqref{hstar} admitted a unique solution in all examples we considered. 

The left panel of Figure \ref{p6.0} displays the three dimensional plot of the value function $V$. The right panel plots the optimal initial stock holding $h^*(x)$, the optimal hedge $\Gamma(x)$ at the rebalance time and the benchmark Black-Scholes Delta $P'(x)$. Notice that the optimal stock holding after rebalancing $\Gamma(x)$ is flatter then the Delta $P'(x)$, thus the constrained trader under/over-hedges, compared to the Black-Scholes benchmark, if the option is in-the-money/out-of-the-money.\footnote{We are grateful to an anonymous reviewer for this observation.} This reflects the fact that no further trades are possible before $\tau_{\cI}$. For example, if rebalancing occurs when the option is out of the money (close to $b$), there is still a positive probability of reaching the left boundary $a$ before hitting $b$. Therefore, the optimal stock holding $\Gamma(x)$ strikes a balance between optimal Black-Scholes hedges $P'(b)$ at $b$ and $P'(a)$ at $a$. This is unnecessary in the Black-Scholes setting because the portfolio can be rebalanced continuously reacting to changes in the underlying price. The optimal initial stock holding $h^*(x)$ exhibits similar flatter characteristics as $\Gamma(x)$ close to boundaries $a, b$ but is steeper than the Black-Scholes hedge $P'(x)$ in the middle of the graph. The kinks in the map $x\mapsto h^*(x)$ correspond to the points $h_\alpha,h_\beta$ from Figure \ref{p4.0}. They are the points at which the transition between single and double boundaries is observed. The steep part of the graph of $h^*(x)$ coincides with the region where the rebalancing occurs at two boundaries.

\section{Remarks on the role of the upper bound $b$}\label{secSETB}
Before moving on to the numerical illustration, it is worth turning our attention to the question of what happens if we take $b=+\infty$.

In this case, $\tau_\cI=\inf\{t\ge 0: S_t\le a\}=:\tau_a$ and since $S_t$ has a positive drift, we have $\P(\tau_a=\infty)>0$. The martingale $(e^{-rt}S_t)_{t\ge 0}$ is not uniformly integrable and neither is the one defined by \eqref{portdyna}, for a general admissible trading strategy $(\tau, \theta)$. Then the derivation of \eqref{eq:V} via optional sampling is not possible (since \eqref{eq:mean} does not hold) and the whole problem formulation becomes less transparent. We propose here two possible problem formulations and their corresponding solutions. We note that such solutions appear to be structurally different as a consequence of different mathematical ways in which we can interpret the event $\{\tau_a=\infty\}$ in our model.

Thanks to the explicit dynamics of $S$ we can easily derive $\lim_{t\to\infty}e^{-rt}S_t=0$, $\P_x$-a.s., for all $x\in (0,\infty)$. Then, using a standard convention on the event $\{\tau_a=\infty\}$, we have
\begin{align}\label{eq:SS}
e^{-r\tau_a}S_{\tau_a}&=e^{-r\tau_a}S_{\tau_a}1_{\{\tau_a<\infty\}}+e^{-r\tau_a}S_{\tau_a}1_{\{\tau_a=\infty\}}\\
&=e^{-r\tau_a}a1_{\{\tau_a<\infty\}}+\lim_{t\to\infty}e^{-r t}S_{t}1_{\{\tau_a=\infty\}}=e^{-r\tau_a}a1_{\{\tau_a<\infty\}}.\notag
\end{align}
Analogously, recalling that the put option price is bounded by $K$ we also have
\begin{align}\label{eq:PP}
e^{-r\tau_a}P(S_{\tau_a})=e^{-r\tau_a}P(S_{\tau_a})1_{\{\tau_a<\infty\}}=e^{-r\tau_a}P(a)1_{\{\tau_a<\infty\}}.
\end{align}

\subsection{Zero-mean tracking}\label{subsec:zero_mean}
With the aim of retaining a zero-mean tracking error analogue to \eqref{eq:mean} we set 
\[
\tau_n:=\inf\{t\geq 0: S_t\ge n\},\qquad \text{for $n\in [a,\infty)$},
\]
and, recalling that $\tau_\cI=\tau_a$, we study the problem
\begin{align}\label{eq:V-inf-b}
\mathcal V(x):=&\inf_{(\tau,\theta)\in\cA^{\infty}_x}\limsup_{n\uparrow \infty} \mathcal{V}ar_{x}\left[e^{-r\tau_a\wedge\tau_n}\bigl(\Pi^{\tau,\theta}_{{\tau_a\wedge\tau_n}}-P(S_{{\tau_a\wedge\tau_n}})\bigr)\right],
\end{align}
where $\cA^{\infty}_x$ is defined in the same way as Definition \ref{Deftradingstrategy} with $\cI$ replaced by $(a,\infty)$ and $\tau^x_{\cI}$ replaced by $\tau^x_a$. Notice also that we have $h\in\cH=[P'(a),0]$. With this approach the mean tracking error can be computed as
\[
\lim_{n\to\infty}\E_x\left[e^{-r\tau_a\wedge\tau_n}\bigl(\Pi^{\tau,\theta}_{{\tau_a\wedge\tau_n}}-P(S_{{\tau_a\wedge\tau_n}})\bigr)\right]=0
\]
by an application of optional sampling for each $n\ge a$ given and fixed. Clearly for $b<\infty$ problem formulations \eqref{eq:V} and \eqref{eq:V-inf-b} are equivalent since $\tau_\cI=\tau_\cI\wedge\tau_n$ for all $n>b$.

As in Sections \ref{sec:1} and \ref{sec2.1} (with a slight abuse of notation) we have
\begin{equation}
\label{eq:V-inf-b-1}
\begin{aligned}
\mathcal V(x)&=\inf_{(\tau,\theta)\in\cA^{\infty}_x}\limsup_{n\uparrow \infty} \E_{x}\left[\int_0^{\tau_a\wedge\tau_n}e^{-2ru}f(S_u,\theta_u)\ud u\right]
=:\inf_{h\in\mathcal H} V(x,h),
\end{aligned}
\end{equation} 
where
\begin{align}
\label{eq:V-inf-b-2}
V(x,h)&=\inf_{\tau\le\tau_a, h_1\in\mathcal{H}_{m}^{\tau}}\limsup_{n\uparrow \infty}\E_x\left[\int_0^{\tau\wedge\tau_n}e^{-2ru}f(S_u,h)\ud u +e^{-2r(\tau\wedge\tau_n)}\widehat{M}_n(S_{\tau\wedge\tau_n},h_1)\right],
\end{align}
and
\begin{equation}
\label{eq:V-inf-b-3}
\widehat{M}_n(x,\zeta):=\E_x\left[\int_{0}^{\tau_a\wedge\tau_n}e^{-2ru}f(S_u,\zeta)\ud u\right],\qquad \zeta\in\R, \quad x \in \closure\cI.
\end{equation}
First, we show that $\lim_{n\uparrow \infty}\widehat M_n(x,\zeta)=\infty$ for all $\zeta\neq 0$. Then we will use it to argue that the infimum in \eqref{eq:V-inf-b-2} is attained for $h_1 \equiv 0$. 
 
For each $n> a$ and $x \in (a,n)$ we have 
\begin{align}\label{eq:taun}
&\E_x\left[\int_{0}^{\tau_a\wedge\tau_n}e^{-2ru}f(S_u,\zeta)\ud u\right]\\
&=\zeta^2\E_{x}\left[\int_{0}^{\tau_a\wedge\tau_n}e^{-2ru}\sigma^2 S^2_u\ud u \right]-2\zeta\E_{x}\left[\int_{0}^{\tau_a\wedge\tau_n}e^{-2ru}P'(S_u)\sigma^2 S^2_u\ud u \right]\notag\\
&\quad+\E_{x}\left[\int_{0}^{\tau_a\wedge\tau_n}e^{-2ru}(P'(S_u))^2\sigma^2 S^2_u\ud u \right].\notag
\end{align}
The first term on the right-hand side can be written using \eqref{diffusionformula} as
\begin{align} \label{setb1}
&\E_{x}\left[\int_{0}^{\tau_a\wedge\tau_n}e^{-2ru}\sigma^2 S^2_u\ud u\right]\\
&=\wron^{-1}_n\left(\varphi_n(x)\int_{a}^x\psi(z)\sigma^2z^2m'(z)\ud z+\psi(x)\int_{x}^{n}\varphi_n(z)\sigma^2z^2m'(z)\ud z\right),\nonumber
\end{align}
where $\psi$ and $\varphi_n$ are, respectively, the increasing and decreasing fundamental solutions to \eqref{funODE}, with boundary conditions $\psi(a+)=0$, $\psi'(a+)>0$ and $\varphi_n(n-)=0$, $\varphi_n'(n-)<0$, while $w_n$ is the associated Wronskian. 
These quantities can be computed explicitly as in \eqref{eq:fund} and \eqref{eq:w}, and they read
\[
\psi(x)= x^{q_1}-a^{q_1-q_2}x^{q_2},\quad\varphi_n(x)=x^{q_2}-n^{q_2-q_1} x^{q_1},\quad w_n=\hat w (1-(a/n)^{q_1-q_2}),
\]
where $q_2<0<q_1$ are given in \eqref{eq:q12}. 

Clearly $w_n\uparrow \hat w$ and $\varphi_n(x)\uparrow x^{q_2}$ as $n\to \infty$. Then, the first integral on the right-hand side of \eqref{setb1} remains bounded as $n\to\infty$. For the second integral we have, by monotone convergence,
\[
\lim_{n\to\infty}\int_{x}^{n}\varphi_n(z)\sigma^2z^2m'(z)\ud z=\int_{x}^{\infty}z^{q_2}\sigma^2z^2m'(z)\ud z=+\infty,
\] 
where the final equality can be easily obtained by recalling the expression of $m'(z)$ (see, \eqref{speedmeasure}) and upon noticing that $q_2+d+1>0$.
Using the same method one can check that the second and third terms on the right-hand side of \eqref{eq:taun} remain finite as $n\to \infty$, due to the damping effect of $P'(x)$ as $x\to\infty$. Then, we have $\lim_{n \uparrow \infty} \widehat{M}_n(x,\zeta)=+\infty$ unless $\zeta\equiv 0$.

For any $\tau \le \tau_a$ and $h_1 \in \mathcal{H}_{m}^{\tau}$, using $\widehat{M}_n(n, h_1) = 0$ we obtain
\begin{align*}
&\E_x\left[\int_0^{\tau\wedge\tau_n}e^{-2ru}f(S_u,h)\ud u +e^{-2r(\tau\wedge\tau_n)}\widehat{M}_n(S_{\tau\wedge\tau_n},h_1)\right]\\
&=\E_x\left[\int_0^{\tau\wedge\tau_n}e^{-2ru}f(S_u,h)\ud u +e^{-2r\tau}\widehat{M}_n(S_{\tau},h_1)1_{\{\tau<\tau_n\}}\right].
\end{align*}
Since $\tau_n\uparrow \infty$ as $n\to\infty$, $f\ge 0$ and $\widehat M_n$ is non-negative and increasing in $n$, we can apply monotone convergence theorem to pass the limit under expectation. Hence, 
\begin{align*}
&\limsup_{n\uparrow \infty}\E_x\left[\int_0^{\tau\wedge\tau_n}e^{-2ru}f(S_u,h)\ud u +e^{-2r\tau}\widehat{M}_n(S_{\tau},h_1)1_{\{\tau<\tau_n\}}\right]\\
&=
\E_x\left[\int_0^{\tau} e^{-2ru}f(S_u,h)\ud u + \lim_{n\uparrow \infty} e^{-2r\tau}\widehat{M}_n(S_{\tau},h_1) 1_{\{\tau < \tau_n\}} \right].
\end{align*}
Recalling that $\lim_{n \uparrow \infty} \widehat{M}_n(x,\zeta)=+\infty$ for $\zeta \ne 0$, we have that the second term above is infinite unless $h_1 = 0$, $\P_x$-a.s. It follows that the infimum in \eqref{eq:V-inf-b} must necessarily be attained for $h_1 = 0$, $\P_x$-a.s., and using the tower property we have
\[
\limsup_{n\uparrow \infty}\E_x\left[e^{-2r\tau}\widehat{M}_n(S_{\tau},0) 1_{\{\tau < \tau_n\}} \right] = \E_{x}\left[\int_{\tau}^{\tau_a}e^{-2ru}(P'(S_u))^2\sigma^2 S^2_u\ud u \right] =: \E_x \left[e^{-2r\tau} M(S_{\tau})\right].
\]

In light of the above, the hedging problem becomes
\begin{equation*}
V(x,h)=\inf_{\tau\leq\tau_a}\E_x\left[\int_0^{\tau}e^{-2ru}f(S_u,h)\ud u+e^{-2r\tau}M(S_\tau)\right].
\end{equation*}
In this case, we have
\begin{equation*}
G(x,h)=(\mathcal{L}-2r)M(x)+f(x,h)=\sigma^2x^2 h(h-2P'(x)),
\end{equation*}
and it is easy to check that, for each $h\in \cH$, the map $x\mapsto G(x,h)$ has a unique root $x_G = x_p(h/2)$ on $(a,\infty)$ (see \eqref{2.7}). It follows that $G(x,h)<0$ for $x\in(a,x_G)$ and $G(x,h)>0$ for $x>x_G$. By the same argument as in the proof of Proposition \ref{prop4.33} we have that $\cC_h=(a,x^*_h)$ for some $x^*_h\geq x_G$ that can be found explicitly by solving an analogue of \eqref{eq:ODEp}. The corresponding optimal hedging strategy prescribes to clear the stock position (i.e., $h^*_1\equiv0$) as soon as the stock price $S_t$ enters the interval $[x^*_h,\infty)$. 

\subsection{Non-zero mean tracking error}\label{subsec:non-zero-mean}

We can formulate the problem directly with the random time horizon $\tau_\cI=\tau_a$. With the same notation as in Section \ref{subsec:zero_mean}, here we want to solve 
\[
\mathcal V(x)=\inf_{(\tau,\theta)\in\mathcal A^\infty_x}\mathcal{V}ar_x\left[e^{-r\tau_a}\bigl(\Pi^{\tau,\theta}_{\tau_a}-P(S_{\tau_a})\bigr)\right],
\]
and we will indeed produce explicit solutions.

Consider an admissible strategy
\[
\tau=0\quad\text{and}\quad h_1=P(a)a^{-1},
\]
i.e., the rebalancing is immediate at $t=0$ and the bond holding after the trade is $\bar m = P(x) - h_1 x$. 
The discounted portfolio value associated to the above strategy is $\hat\Pi_t:= e^{-rt} \Pi_t = P(a) a^{-1}e^{-rt}S_t+ \bar m$. Using \eqref{eq:SS} and \eqref{eq:PP} the tracking error at time $\tau_a$ is deterministic and amounts to
\[
e^{-r\tau_a}\bigl(\Pi_{\tau_a}-P(S_{\tau_a})\bigr)=\bar m.
\]
Hence, the associated variance is zero and the proposed strategy is optimal.

There is, however, a catch: the hedging portfolio under-replicates the claim. Indeed, recalling the expression for $\bar m$ we have
\[
e^{-r\tau_a}\bigl(\Pi_{\tau_a}-P(S_{\tau_a})\bigr)=\bar m= P(x)-(x/a)P(a)<0, \quad \text{$\P_x$-a.s.},
\]
for all $x>a$, where we used that $P(x)<P(a)$.

One can, however, construct a strategy with a non-negative tracking error (the portfolio value $\Pi_{\tau_a}$ dominates $P(S_{\tau_a})$, $\P_x$-a.s.), but with non-zero variance. This strategy prescribes to initially take a position $h = P'(a)$ in stocks (recall that $P'(a) < 0$ so this is short-selling), and buy $m_0 = P(x) - h x$ bonds. We will show that, on the one hand, if the stock price approaches the boundary $a$, the value of this portfolio grows and allows us to rebalance to a perfect hedge for the boundary $a$. On the other hand, if the stock price diverges to $\infty$ before rebalancing, the discounted portfolio value converges to $m_0$ thanks to \eqref{eq:SS}; it follows from $h < 0$ and $P(x) > 0$ that $m_0 > 0$, so the tracking error is non-negative as required.

Before rebalancing the hedging portfolio evolves according to $e^{rt}m_0+hS_t$.  We choose the rebalancing time so that a perfect hedge can be constructed. We set
\begin{equation}\label{eqn:tau_st}
\tau=\tau^* := \inf\big\{t \ge 0: e^{rt} m_0 + hS_t = S_t P(a) a^{-1}\big\},
\end{equation}
and $h_1 = P(a)a^{-1}$ so that the whole portfolio wealth $e^{rt} m_0 + hS_t$ is invested in stocks. The strategy is clearly self-financing and in order to show that it provides a hedge we first show that $\tau^*<\tau_a$, $\P_x$-a.s., i.e., the rebalancing occurs before hitting the boundary $a$.

First of all the stopping time $\tau^*$ can be rewritten as the first time the discounted stock price $\hat S_t:=e^{-rt}S_t$ falls below a certain threshold:
\[
\tau^*=\inf\bigg\{t \ge 0: \hat S_t \le \frac{m_0}{a^{-1}P(a)-h}\bigg\}.
\]
Using that $a^{-1}P(a)>x^{-1}P(x)$ for all $x>a$ by the monotonicity of $y\mapsto P(y)$, we have 
\[
\frac{m_0}{a^{-1}P(a)-h}=\frac{x^{-1}P(x)-h}{a^{-1}P(a)-h}x<x=\hat S_0=S_0.
\]
Therefore $\tau^*>0$, $\P_x$-a.s. Since the mapping $y \mapsto P(y) - h y$ is strictly increasing for $y > a$ (because $P'(y) > P'(a) = h$) we also have
\[
\frac{P(x)-hx}{a^{-1}P(a)-h}=\frac{P(x)-hx}{P(a)-ha}a>a,
\]
so that $\hat S_t = (P(x)-hx)/(a^{-1}P(a)-h)$ implies $S_t>a$ and therefore $\tau_a>\tau^*$, $\P_x$-a.s., as needed.

Denoting by $\Pi_t$ the portfolio value, we have $\Pi_{\tau_a} = P(a)$ on $\{\tau_a < \infty\}$.  Using \eqref{eq:SS} and \eqref{eq:PP} the discounted tracking error at time $\tau_a$ amounts to
\[
e^{-r\tau_a}\bigl(\Pi_{\tau_a}-P(S_{\tau_a})\bigr)= e^{-r\tau_a}\bigl(\Pi_{\tau_a}-P(S_{\tau_a})\bigr) 1_{\{\tau^* = \infty\}} = m_0 1_{\{\tau^* = \infty\}}, \quad \text{$\P_x$-a.s.}
\]
Hence, neither the associated variance nor the expectation is zero but the portfolio value dominates the payoff at time $\tau_a$.

The hedging strategies obtained in Section \ref{subsec:non-zero-mean} seem economically unintuitive as they prescribe to take a long position in the stock and null investment in bonds after the rebalance. It should be clear that this is a mathematical artefact due to the infinite-time horizon.

\section{Numerical comparisons}\label{sec:concl}

We assess the performance of our optimal hedging strategy against the performance of ad-hoc strategies inspired by those often used in practice (see \cite[Chapter 6]{sinclair2011}). The quality of each strategy is measured in terms of the variance of the tracking error at $\tau_\cI$.

We consider five hedging strategies:
\begin{description}
\item[(Strategy 1)] Our optimal strategy $(\theta^*,\tau^*)$.
\item[(Strategy 2)]
Start with an initial stock holding $\theta_0=P'(x)$ and rebalance at the stopping time 
\[
\zeta:=\inf\left\{t\ge 0: S^x_t \notin \left(\tfrac{1}{2}(a+x),\tfrac{1}{2}(b+x)\right)\right\},
\]
with the classical Delta hedge $\theta_{\zeta}=P'(S^x_{\zeta})$; then hold until $\tau_\cI$.
\item[(Strategy 3)]
Start with an initial stock holding $\theta_0=P'(x)$ and rebalance when the Delta of the current stock price leaves a certain region. That is, let 
\[
\rho:=\inf\left\{t\ge 0: P'(S^x_t) \notin \left(\tfrac{1}{2}(P'(a)+P'(x)),\tfrac{1}{2}(P'(b)+P'(x))\right)\right\},
\]
and rebalance at $\rho$ with the classical Delta hedge $\theta_{\rho}=P'(S^x_{\rho})$; then hold until $\tau_\cI$.
\item[(Strategy 4)]
Start with $\theta_0=\Gamma(x)$ and hold the same amount of stock until $\tau_\cI$.
\item[(Strategy 5)]
Start with $\theta_0=P'(x)$ and hold the same amount of stock until $\tau_\cI$.
\end{description}

An illustration of Strategy 1, which is the optimal one from our analysis, is given in Figure \ref{p5.0}. Strategies with fixed thresholds, like Strategy 2 and 3 above, are popular in the finance sector (see \cite[p.~95]{sinclair2011}) and have an intuitive meaning: the trader makes the portfolio Delta-neutral when the underlying stock price or the associated Delta diverge by a `fixed amount' from their initial values. Such an `amount' of course can be chosen in several different ways; here we only display results for the specific choices made above. However, other specifications of the intervals in the stopping rules $\zeta$ and $\rho$ give results qualitatively consistent with those presented in this section.

Strategies 4 and 5 are so-called static hedging strategies. In Strategy $4$ the static hedging is optimal in the sense that $\Gamma(x)$ minimises the variance of the tracking error when no other rebalancing is allowed. 
\begin{figure}[tb]
\centering
\includegraphics[width=8cm]{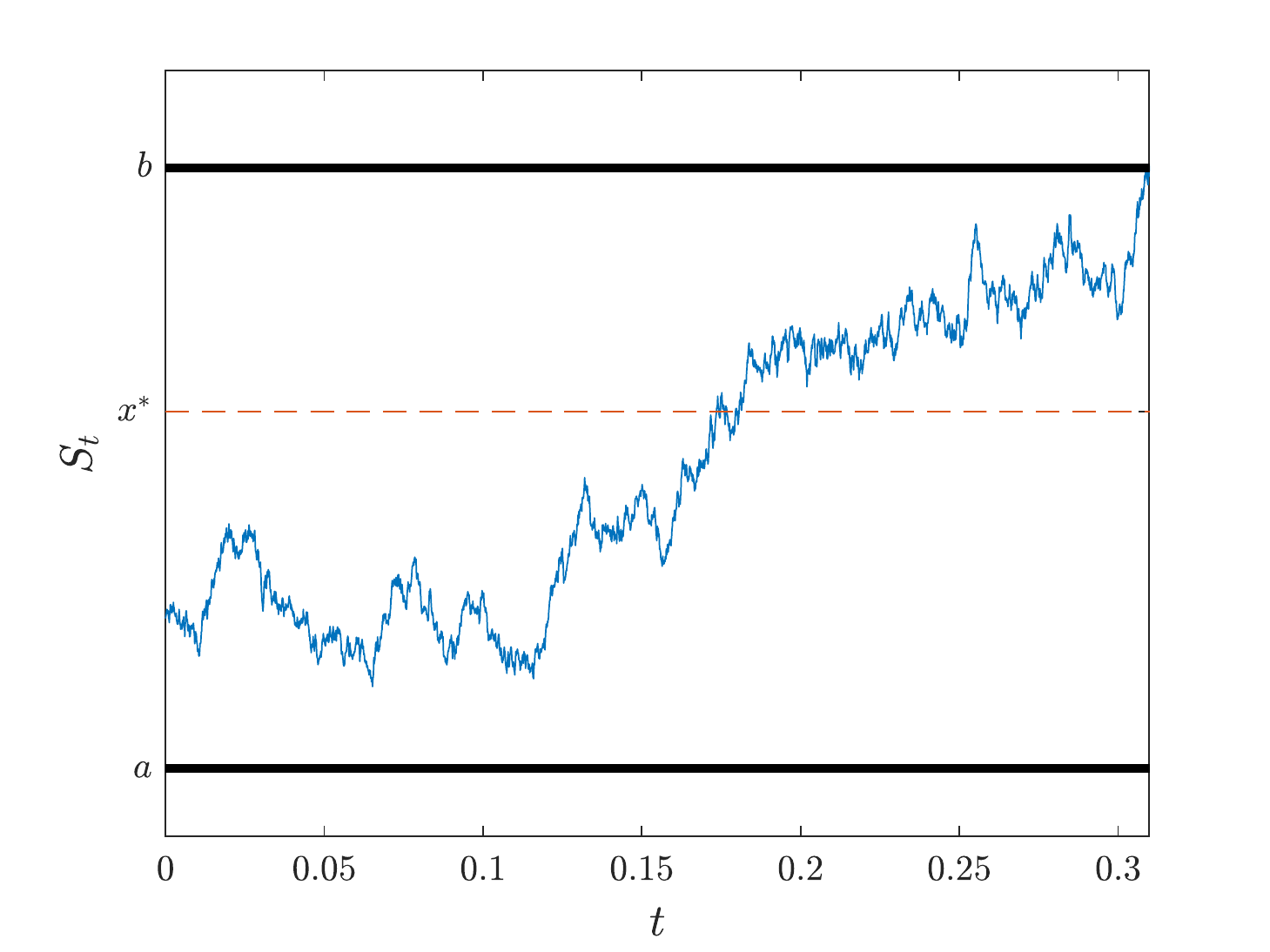}
\includegraphics[width=8cm]{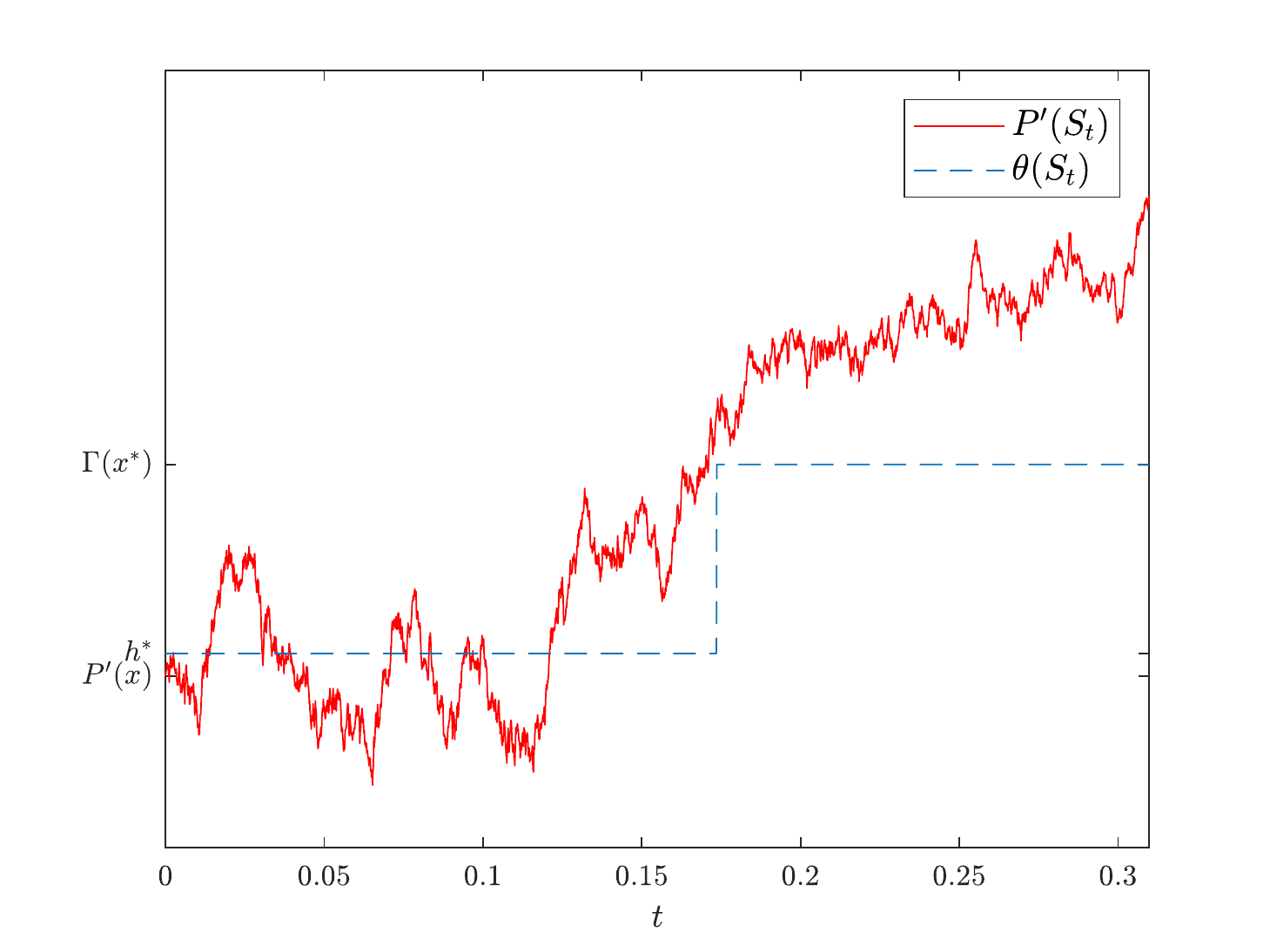}
\caption{A simulation of stock price and the optimal hedging strategy $1$ using parameters $r=3\%,\sigma=30\%, K=100, a=90, x=100,b=130$.} 
\label{p5.0}
\end{figure} 

We evaluate the performance of these five strategies by conducting three experiments: we calculate the sample variance of tracking error with different values of initial stock price $S_0$, volatility $\sigma$ and upper re-assessment boundary $b$, respectively, when other parameters are fixed.  In all experiments, the estimates are based on the same $N=1000$ sample paths of the stock price and model parameters are fixed as
\begin{equation}
\label{par:1}
r=3\%,\sigma=30\%, K=100, S_0 = 100, a=90, b=110.
\end{equation}

\begin{figure}[tb]
\centering
\includegraphics[width=10cm]{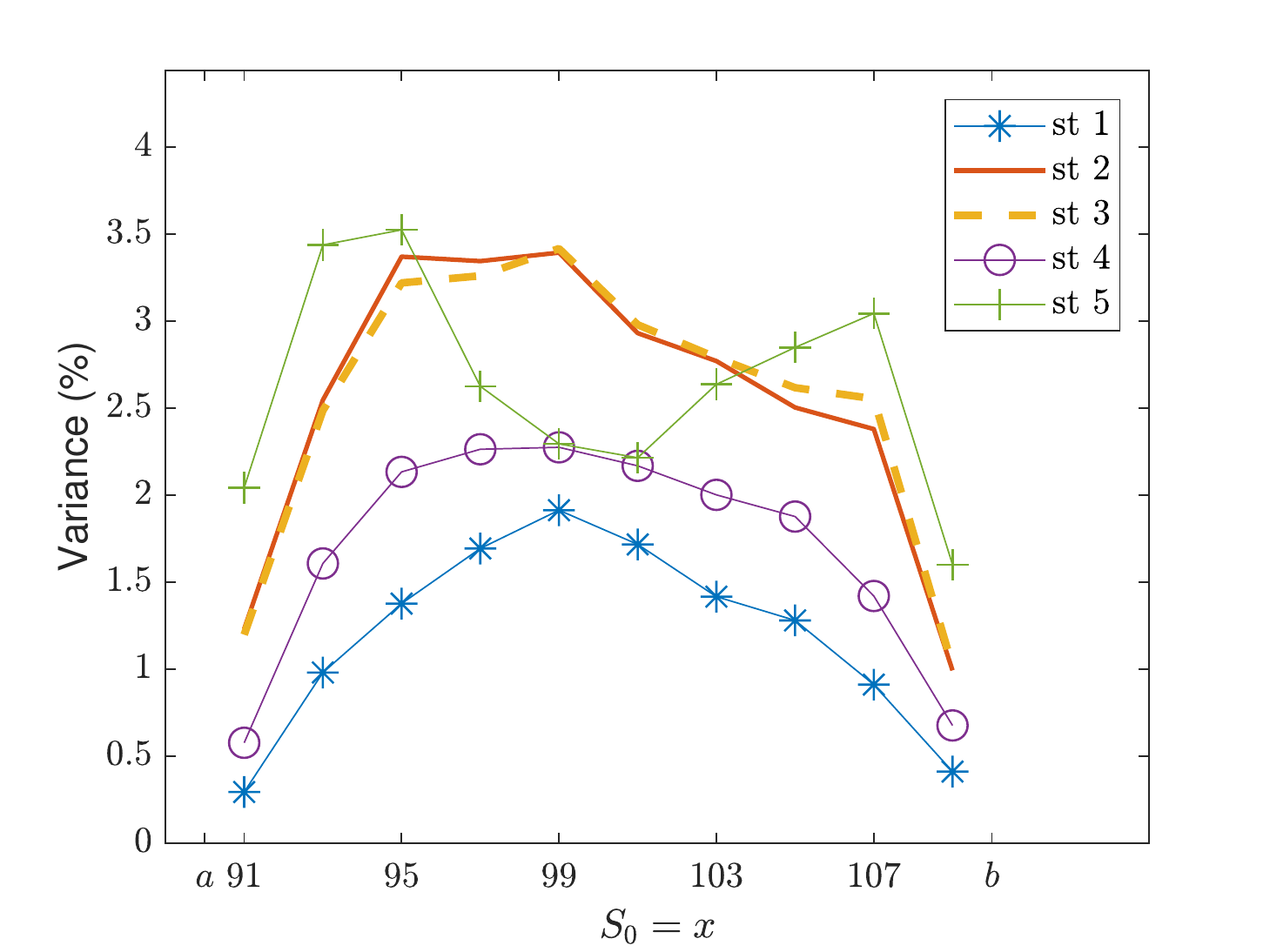}
\caption{Sample variance of the hedging error for different values of $S_0$, with parameters as in \eqref{par:1}.} 
\label{p-compare-x}
\end{figure} 
\begin{table}[ht]
\caption{Sample variance of the hedging error (\%) for different values of $S_0$ with parameters as in  \eqref{par:1}.}
\centering 
\begin{tabular}{lccccc}
\hline
\hline 
         & Strategy 1      &Strategy 2      &Strategy 3      & Strategy 4      &Strategy 5      \\[0.5ex]
\hline
$S_0=$ 91     & 0.29 & 1.23 & 1.20 & 0.58 & 2.04 \\
$S_0=$ 93     & 0.98 & 2.55 & 2.49 & 1.61 & 3.44 \\
$S_0=$ 95     & 1.38 & 3.37 & 3.22 & 2.13 & 3.53 \\
$S_0=$ 97     & 1.69 & 3.35 & 3.26 & 2.26 & 2.63 \\
$S_0=$ 99     & 1.91 & 3.39 & 3.42 & 2.27 & 2.29 \\
$S_0=$ 101    & 1.72 & 2.93 & 2.98 & 2.17 & 2.21 \\
$S_0=$ 103    & 1.42 & 2.77 & 2.79 & 2.00 & 2.64 \\
$S_0=$ 105    & 1.28 & 2.50 & 2.62 & 1.88 & 2.85\\
$S_0=$ 107    & 0.91& 2.38 & 2.55 & 1.42 & 3.04 \\
$S_0=$ 109    & 0.41 & 0.99 & 1.01 & 0.68 & 1.60\\[1ex]
\hline 
\end{tabular}\label{table:1}
\end{table}

\begin{figure}[tb]
\centering
\includegraphics[width=10cm]{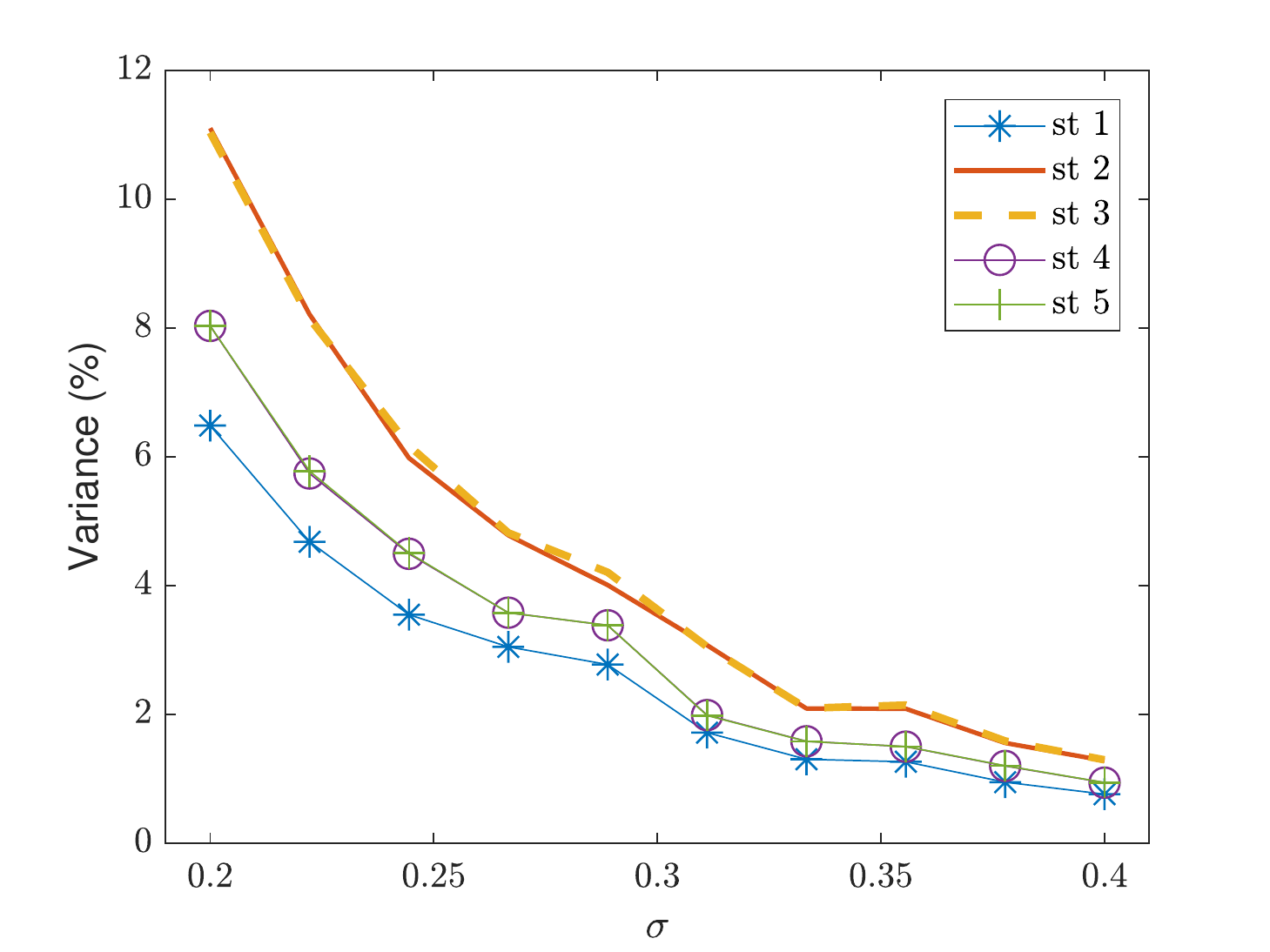}
\caption{Sample variance of the hedging error for different values of $\sigma$ with parameters as in \eqref{par:1}.} 
\label{p-compare-sigma}
\end{figure}

\begin{figure}[tb]
\centering
\includegraphics[width=0.45\textwidth]{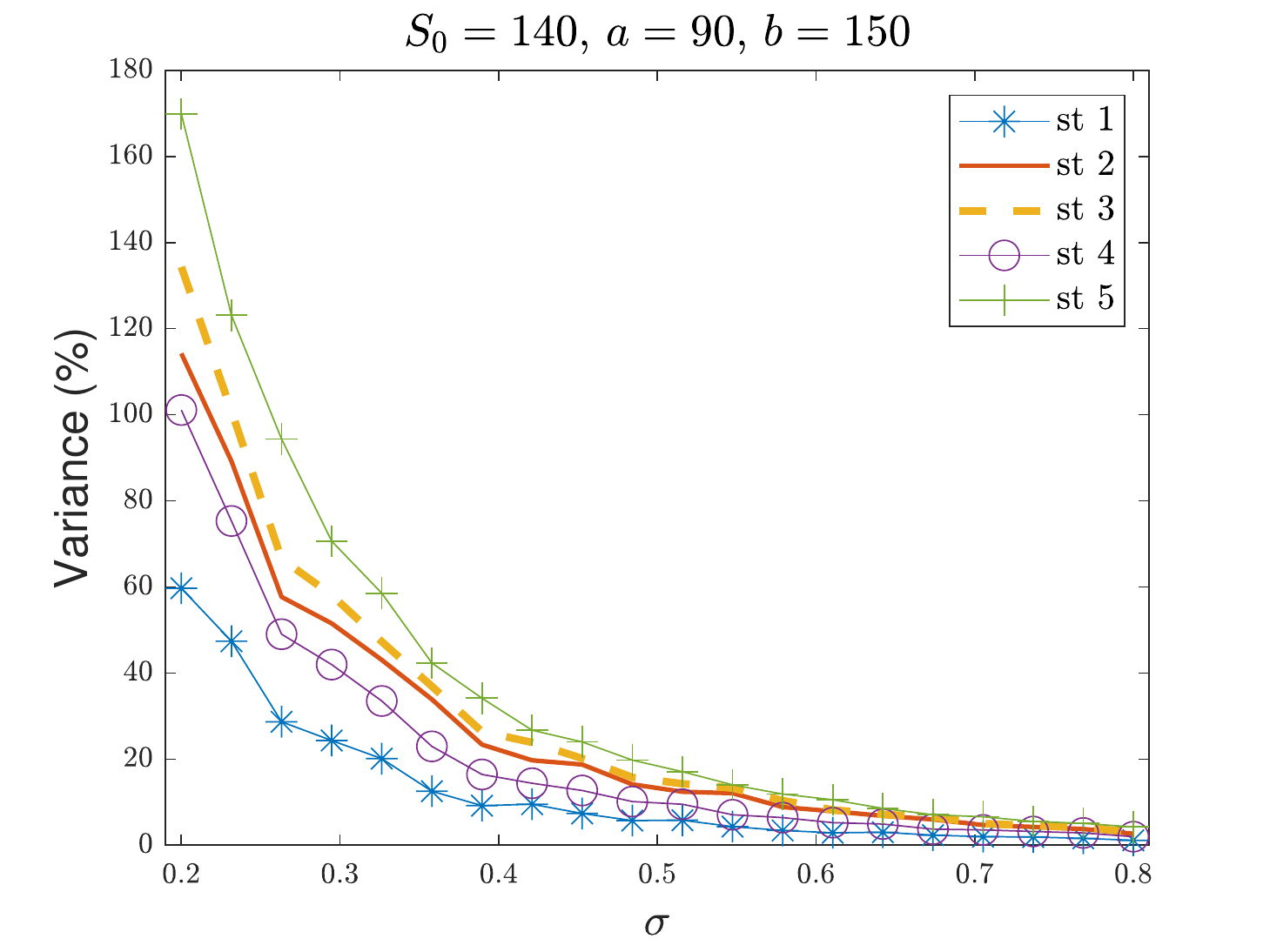}
\includegraphics[width=0.45\textwidth]{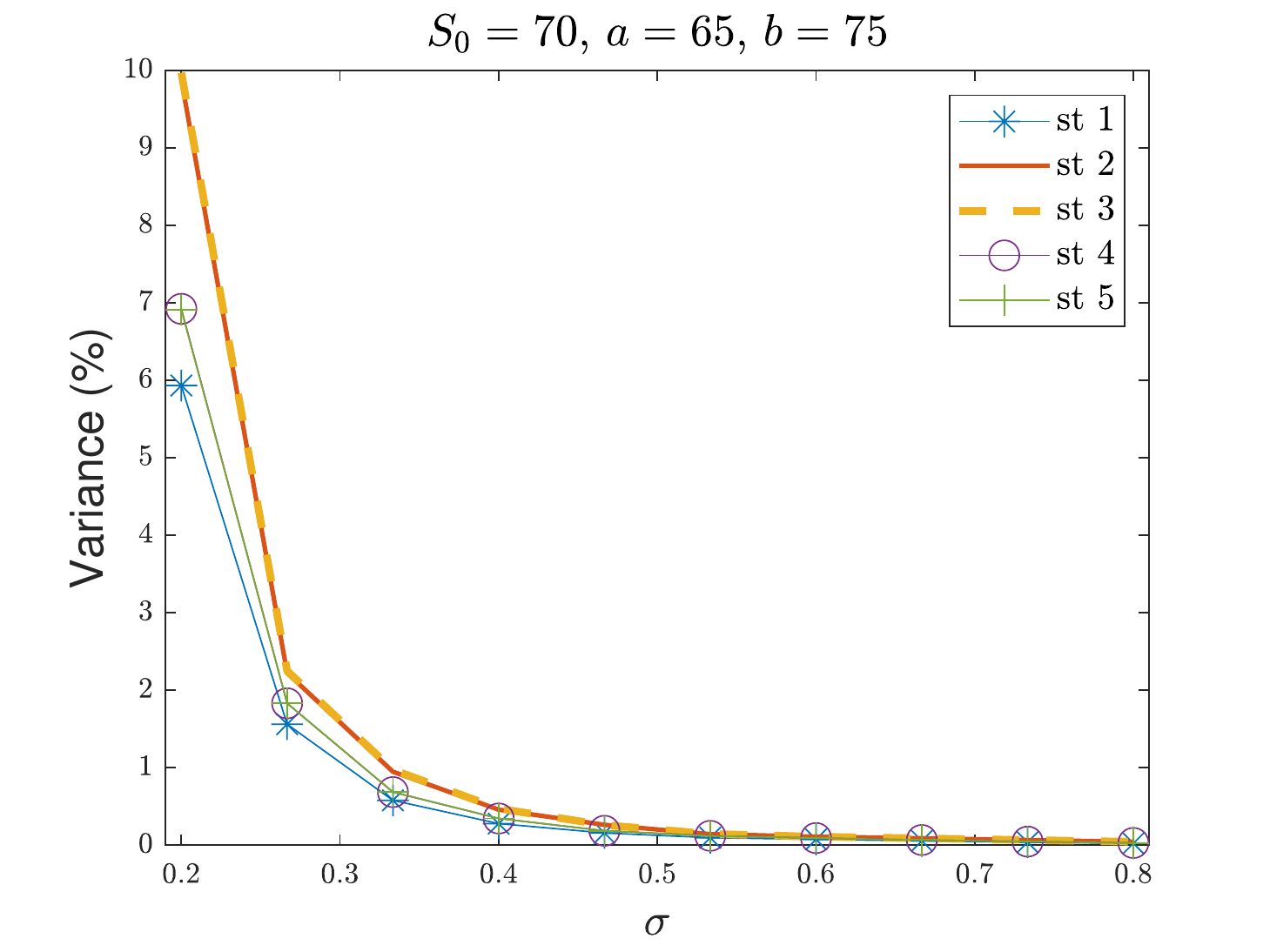}
\caption{Sample variance of the hedging error (\%) for different values of $\sigma$ with parameters $r=3\%$, $K=100$ and deep out-of-the-money option (left panel) and deep in-the-money option (the right panel).} 
\label{p-compare-sigma-2}
\end{figure}

\begin{table}[tb]
\caption{Sample variance of the hedging error (\%) for different values of $\sigma$ with parameters as in \eqref{par:1}.}

\centering 
\begin{tabular}{lccccc}
\hline
\hline 
             &Strategy   1      &Strategy   2      &Strategy   3      &Strategy   4      &Strategy   5      \\[0.5ex]
\hline
$\sigma=$ 20.00\%     & 6.48 & 11.10 & 11.03 & 8.03 & 8.03 \\
$\sigma=$ 22.22\%    & 4.68 & 8.21  & 8.14  & 5.74 & 5.78 \\
$\sigma=$ 24.44\%     & 3.55 & 5.98 & 6.17 & 4.49 & 4.50 \\
$\sigma=$ 26.67\%    & 3.05& 4.78  & 4.82  & 3.58 & 3.57 \\
$\sigma=$ 28.89\%    & 2.77 & 4.00 & 4.20  & 3.38 & 3.38 \\
$\sigma=$ 31.11\%  & 1.71 & 3.07  & 3.05  & 1.98 & 1.98 \\
$\sigma=$ 33.33\%    & 1.30 & 2.09 & 2.09  & 1.58 & 1.58 \\
$\sigma=$ 35.56\%  & 1.26 & 2.08  & 2.14  & 1.49& 1.50 \\
$\sigma=$ 37.78\%    & 0.94 & 1.55  & 1.58  & 1.20 & 1.20 \\
$\sigma=$ 40.00\%  & 0.76 & 1.28  & 1.29  & 0.93 & 0.94\\[1ex]
\hline 
\end{tabular}\label{table:2}
\end{table}

\begin{figure}[tb]
\centering
\includegraphics[width=10cm]{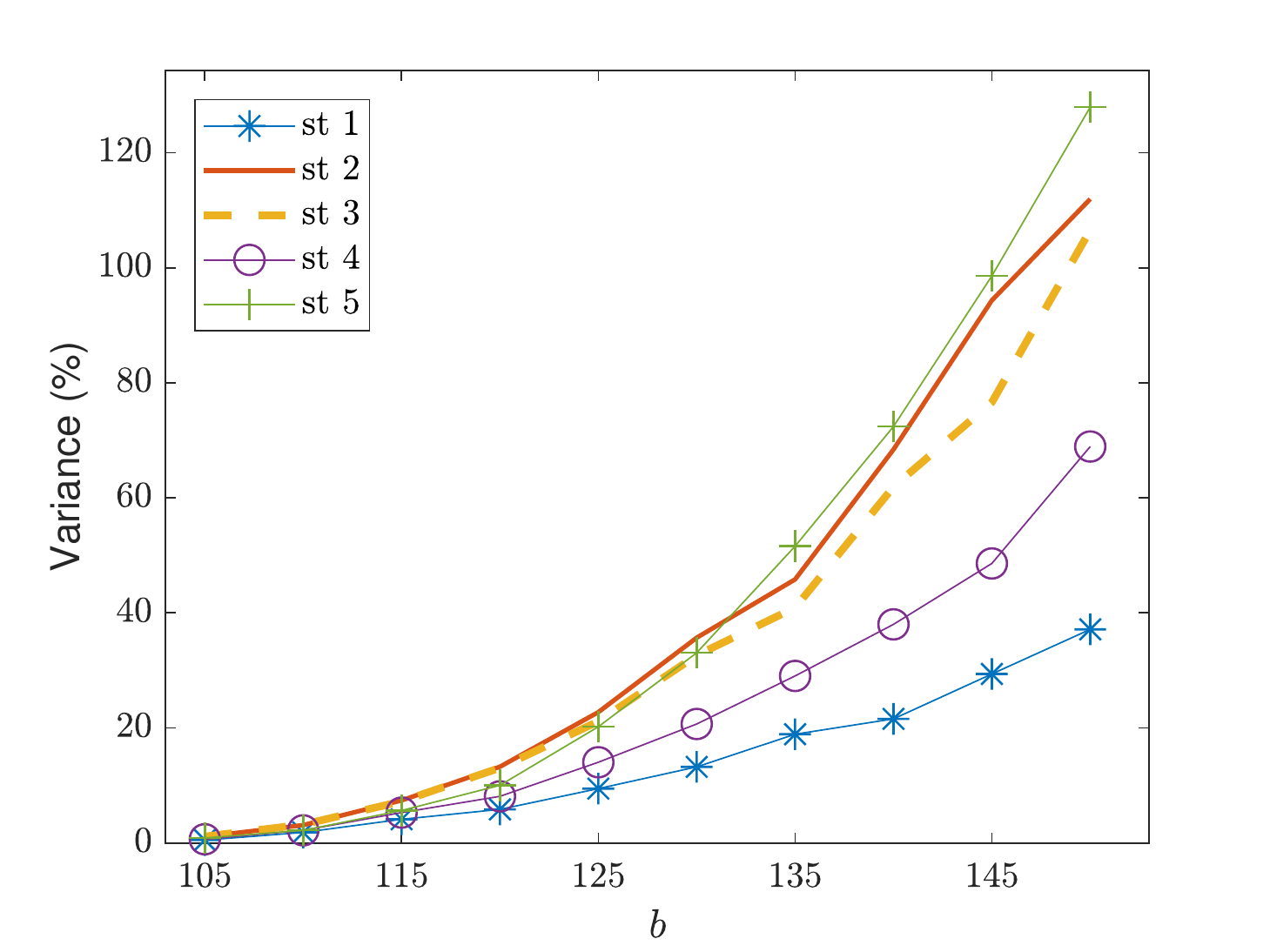}
\caption{Sample variance of the hedging error for different values of $b$ with parameters as in \eqref{par:1}.} 
\label{p-compare-b}
\end{figure} 

\begin{table}[tb]
\caption{Sample variance of the hedging error (\%) for different values of $b$ with parameters as in \eqref{par:1}.}
\centering 
\begin{tabular}{lccccc}
\hline
\hline 
             &Strategy 1      &Strategy 2      &Strategy 3      &Strategy 4      &Strategy 5      \\[0.5ex]
\hline
$b=$ 105  & 0.51  & 1.12  & 1.14   & 0.66  & 0.88   \\
$b=$ 110  & 1.84  & 3.10   & 3.21  & 2.20  & 2.20   \\
$b=$ 115  & 4.12  & 7.40   & 7.24  & 5.29  & 5.65  \\
$b=$ 120  & 5.85  & 13.26  & 13.06  & 8.11  & 10.03  \\
$b=$ 125  & 9.45  & 22.75  & 21.14  & 14.04 & 20.21 \\
$b=$ 130 & 13.23 & 35.68  & 32.63  & 20.69 & 33.07  \\
$b=$ 135 & 18.91 & 45.82  & 40.91  & 29.04 & 51.60  \\
$b=$ 140 & 21.58 & 68.40 & 61.91  & 38.00 & 72.40  \\
$b=$ 145 & 29.38 & 94.32  & 76.60  & 48.59 & 98.58  \\
$b=$ 150 & 37.14 & 111.93 & 106.53 & 68.92 & 127.93\\[1ex]
\hline 
\end{tabular}\label{table:3}
\end{table}

For the first experiment we consider $10$ different values for the initial stock price $S_0$ evenly spaced in the interval 
\[
[91,109]
\]
with all other parameters fixed as in \eqref{par:1}.
As shown in Figure \ref{p-compare-x} and Table \ref{table:1}, the variance of the tracking error for Strategy 1 is at least $40\%$ lower than the variance for the dynamic strategies 2 and 3, and at least $15\%$ lower than the variance for the static hedging strategies 4 and 5. It is worth noticing that the static strategy 4 outperforms the dynamic strategies 2 and 3.

In the second experiment, we take $10$ values of the volatility $\sigma$ evenly spaced in the interval
\[
[20\%,40\%]
\]
with all other parameters fixed as in \eqref{par:1}.
Results are shown in Figure \ref{p-compare-sigma} and Table \ref{table:2}. Our optimal strategy (Strategy 1) produces the variance of the tracking error which is about $30-40\%$ lower than the variance for strategies 2 and 3, and about $15-20\%$ lower than the variance for strategies 4 and 5. The relative gap between different strategies does not vary significantly as the volatility changes. Strategies 4 and 5 produce almost the same results; this happens because $S_0$ is taken as the middle point in $(a,b)$ and therefore the difference between $P'(S_0)$ and $\Gamma(S_0)$ is very small (for example, when $\sigma=31.11\%$ we have $P'(100)=-0.2110$ and $\Gamma(100)=-0.2115$). Strategies 2 and 3 also give very similar results, but they are out-performed by the static strategies.

The steep decline of all graphs in Figure \ref{p-compare-sigma} may seem at odds with the intuition that a high volatility corresponds to a risky trading environment. However, a high volatility also causes the option price to change slower as a function of the stock price: the difference $P(a) - P(b)$ is above $0.14$ for $\sigma = 20\%$ and less than $0.05$ for $\sigma = 40\%$. A larger volatility makes the tracked values closer to each other and, hence, the tracking problem easier. This intution was confirmed by extensive numerical studies with representative results for in-the-money and out-of-the-money options displayed on Figure \ref{p-compare-sigma-2}.

The third experiment studies the effect of the upper boundary $b$. We take $10$ values of $b$ evenly spaced in the interval
\[
[105,150]
\]
with all other parameters fixed as in \eqref{par:1}. The results are displayed in Figure \ref{p-compare-b} and Table \ref{table:3}. For values of $b$ close to $S_0$, the variance of the tracking error for all strategies is low because the stock price leaves the interval $(a,b)$ quickly. Observe that when $b$ is large, the dynamic optimal strategy 1 produces the variance which is $40\%$ lower than the second best (Strategy 4). This gap shrinks to about $20\%$ when $b$ is small. This indicates that both dynamic hedging and optimisation are important when one of the re-assessment boundaries is far away from $S_0$. Quite remarkably, in all the above experiments, the optimised static hedging (Strategy 4) gives a smaller variance of the tracking error than strategies 2 and 3, despite the fact that the latter two allow for one rebalancing opportunity.

\vspace{+15pt}
\noindent{\bf Acknowledgment}: C.~Cai gratefully acknowledges support by China Scholarship Council. T.~De Angelis gratefully acknowledges support by EPSRC grant EP/R021201/1. All authors are indebted to the reviewer for insightful comments, and for suggesting the interpretation in Remark \ref{rem:prem}.

\bibliographystyle{abbrv} 
\bibliography{References}

\end{document}